\documentclass[12pt]{article}
\usepackage{mystyle}

\usepackage{amsmath,amssymb}
\usepackage{graphicx}
\usepackage[colorlinks=true,citecolor=black,linkcolor=black,urlcolor=blue]{hyperref}

\theoremstyle{plain}
\newtheorem{myclaim}{Claim}
\let\oldproof\proof

\renewcommand{\proof}{\setcounter{myclaim}{0}\oldproof}
\newcommand*{\claimproofname}{Proof}
\newenvironment{claimproof}[1][\claimproofname]%
{\begin{oldproof}[#1]}{\end{oldproof}}

\usepackage{url}
\expandafter\def\expandafter\UrlBreaks\expandafter{\UrlBreaks
  \do\-}

\usepackage{enumerate}

\graphicspath{{Figures/}}

\usepackage{tikz}
  \usetikzlibrary{calc,arrows,arrows.meta,angles,quotes}
  \tikzstyle{color1}=[color=red]
  \tikzstyle{color1light}=[color=red!60]
  \tikzstyle{color1verylight}=[color=red!30]
  \tikzstyle{color2}=[color=green]
  \tikzstyle{color2light}=[color=green!60]
  \tikzstyle{color2verylight}=[color=green!30]
  \tikzstyle{color3}=[color=blue]
  \tikzstyle{color3light}=[color=blue!60]
  \tikzstyle{color3verylight}=[color=blue!30]
  \tikzstyle{color4}=[color=orange]
  \tikzstyle{color4light}=[color=orange!60]
  \tikzstyle{color4verylight}=[color=orange!30]
  \tikzstyle{color5}=[color=violet]
  \tikzstyle{color5light}=[color=violet!60]
  \tikzstyle{color5verylight}=[color=violet!30]
  \tikzstyle{fill1}=[fill=red]
  \tikzstyle{fill1light}=[fill=red!60]
  \tikzstyle{fill2}=[fill=green]
  \tikzstyle{fill2light}=[fill=green!60]
  \tikzstyle{fill3}=[fill=blue]
  \tikzstyle{fill3light}=[fill=blue!60]
  \tikzstyle{fill4}=[fill=orange]
  \tikzstyle{fill4light}=[fill=orange!60]
  \tikzstyle{fill5}=[fill=violet]
  \tikzstyle{fill5light}=[fill=violet!60]
  
\newcount\magproz 
\newcount\vorne
\newcount\hinten
\newcount\zwischen
\def\fpath{Figures}
\def\calc_figscale#1{
   \magproz=#1
    \vorne=\magproz
    \divide\vorne by 100
    \hinten=\magproz
    \zwischen=\vorne
    \multiply\zwischen by 100
    \advance\hinten by -\zwischen
    \def\figscale{\the\vorne.\the\hinten}
}

\newcommand{\sG}{G^\star}
\newcommand{\sgn}{\operatorname{sgn}}
\newcommand{\cG}{\mathcal{G}}
\newcommand{\skel}{G_{\operatorname{skel}}}
\newcommand{\ro}{\operatorname{r-out}}
\newcommand{\lo}{\operatorname{l-out}}

\dateline{Aug 1, 2018}


\title{Pentagon contact representations%
\footnote{This paper has appeared in the Electronic Journal of Combinatorics (E-JC 25.3 (2018), P3.39).
  The conference version of this paper has appeared in the proceedings of Eurocomb'17 (ENDM 61C, pp. 421--427) under the same title.}}

\newcommand*\samethanks[1][\value{footnote}]{\footnotemark[#1]}

\author{Stefan Felsner\thanks{Partially supported by DFG grant FE-340/11-1.} \qquad Hendrik Schrezenmaier\samethanks \qquad Raphael Steiner\\
\small Institut f\"ur Mathematik\\[-0.8ex]
\small Technische Universit\"at Berlin\\[-0.8ex] 
\small Germany\\
\small\tt \{felsner,schrezen,steiner\}@math.tu-berlin.de}

\begin{document}

\maketitle

\begin{abstract}
  Representations of planar triangulations as contact graphs of a set
  of internally disjoint homothetic triangles or of a set
  of internally disjoint homothetic squares have received quite some
  attention in recent years.  In this paper we investigate
  representations of planar triangulations as contact graphs of a set
  of internally disjoint homothetic pentagons.  Surprisingly such a
  representation exists for every triangulation whose outer face is a
  $5$-gon.  We relate these representations to \emph{five color
    forests}.  These combinatorial structures resemble Schnyder woods
  and transversal structures, respectively.  In particular there is a
  bijection to certain $\alpha$-orientations and consequently a
  lattice structure on the set of five color forests of a given graph.
  This lattice structure plays a role in an algorithm that is supposed
  to compute a contact representation with pentagons for a given
  graph. Based on a five color forest the algorithm builds a system
  of linear equations and solves it, if the solution is non-negative,
  it encodes distances between corners of a pentagon
  representation. In this case the representation is constructed and
  the algorithm terminates.  Otherwise negative variables guide a
  change of the five color forest and the procedure is restarted with
  the new five color forest. Similar algorithms have been proposed for
  contact representations with homothetic triangles and with squares.
\end{abstract}

\section{Introduction}

A \emph{pentagon contact system}~$\mathcal{S}$ is a finite system of
convex pentagons in the plane such that any two pentagons intersect in
at most one point.  If all pentagons of~$\mathcal{S}$ are regular
pentagons with a horizontal side at the bottom, we
call~$\mathcal{S}$ a \emph{regular pentagon contact representation}.
Note that in this case any two pentagons of~$\mathcal{S}$ are
homothetic.  The contact system is \emph{non-degenerate} if every
contact involves exactly one corner of a pentagon.  The \emph{contact
  graph}~$\cG(\mathcal{S})$ of~$\mathcal{S}$ is the graph that has
a vertex for every pentagon and an edge for every contact of two
pentagons in~$\mathcal{S}$.  Note that~$\cG(\mathcal{S})$ inherits
a crossing-free embedding into the plane from~$\mathcal{S}$.  For a
given plane graph~$G$ and a pentagon contact system~$\mathcal{S}$
with~$\cG(\mathcal{S})=G$ we say that~$\mathcal{S}$ is a
\emph{pentagon contact representation} of~$G$.

We will only consider the case that~$G$ is an \emph{inner
  triangulation of a $5$-gon}, i.e., the outer face of~$G$ is a
$5$-cycle with vertices $a_1,\dotsc,a_5$ in clockwise order, all inner
faces are triangles, there are no loops nor multiple edges, and the
only edges between the vertices $a_1,\dotsc,a_5$ are the five edges of
the outer face.  Our interest lies in a variant of regular pentagon contact
representations of~$G$ with the property that
$a_1,\dotsc,a_5$ are not represented by regular pentagons, but by line segments $s_1,\dotsc,s_5$
which together form a pentagon with all internal angles equal
to~$(3/5)\pi$.  The line segment~$s_1$ is always horizontal and at the
top, and~$s_1,\dotsc,s_5$ is the clockwise order of the segments of
the pentagon. Since this variant of regular pentagon contact representations
is the only kind of contact representations we deal with in this paper, we refer to
these also as regular pentagon contact representations. Figure~\ref{fig:regular_pcr} shows an example.

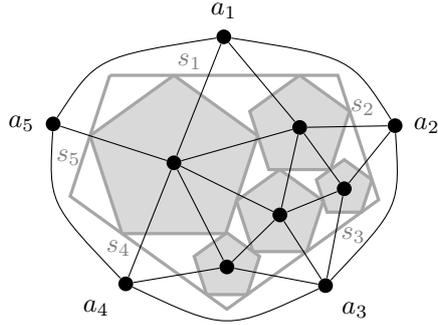
\begin{figure}

\centering

\tikzstyle{normal vertex}=[circle,fill,scale=0.5]
\tikzstyle{edge}=[]

\tikzstyle{pcr}=[very thick,color=gray!70]
\tikzstyle{pcr fill}=[fill=gray!30]

\tikzstyle{segment label}=[color=gray]

\begin{tikzpicture}[scale=3]

\def\A{0.4546294640647149}
\def\B{0.2707279827787506}
\def\C{0.1507367088399636}
\def\D{0.2356059251438577}
\def\E{0.1775675733633562}
\def\aa{0.2809764610791428}
\def\ab{0.4546294640647149}
\def\ba{0.2809764610791428}
\def\bb{0.4546294640647149}
\def\ca{0.1097427956383953}
\def\cb{0.1775675733633562}
\def\da{0.09316040941539433}
\def\db{0.1507367088399636}
\def\ea{0.167319095062964}
\def\eb{0.2707279827787506}
\def\fa{0.1136573660161787}
\def\fb{0.2707279827787506}
\def\fc{0.4546294640647149}
\def\fd{0.5517044438578933}
\def\ga{0.04853748989658916}
\def\gb{0.1570706167625718}
\def\gc{0.2356059251438577}
\def\gd{0.3026830864523393}
\def\ha{0.1775675733633562}
\def\hb{0.171233665440748}
\def\hc{0.4585440344424989}
\def\hd{0.4546294640647149}
\def\ia{0.006333907922608149}
\def\ib{0.1416978993119834}
\def\ic{0.1519463776123757}
\def\id{0.2356059251438577}
\def\ja{0.07415868564756978}
\def\jb{0.1507367088399636}
\def\jc{0.2707279827787506}
\def\jd{0.3180558039029276}
\def\ka{0.09828464856559038}
\def\kb{0.07657802319239386}
\def\kc{0.2356059251438577}
\def\kd{0.2221904928821613}
\def\la{0.1507367088399636}
\def\lb{0.05245206027437295}
\def\lc{0.296349178529731}
\def\ld{0.2356059251438577}
\def\ma{0.03586967405137269}
\def\mb{0.1775675733633562}
\def\mc{0.2356059251438577}
\def\md{0.3231800430531234}

\draw[pcr] (0,0) -- ++(0:   \aa+\fd+\ea)
            -- ++(-72: \ea+\jd+\da)
            -- ++(-144:\da+\lc+\md+\ca)
            -- ++(-216:\ca+\hc+\ba)
            -- ++(-288:\ba+\aa);
            
\coordinate (At) at (0: \aa);
            
\draw[pcr,pcr fill] (At)
        -- ++(-144:\A)
        -- ++(-72:\A)
        -- ++(0:\A)
        -- ++(72:\A)
        -- ++(144:\A);
        
\coordinate (Bt) at (0: \aa+\fd);       
        
\draw[pcr,pcr fill] (Bt)
        -- ++(-144:\B)
        -- ++(-72:\B)
        -- ++(0:\B)
        -- ++(72:\B)
        -- ++(144:\B);
        
\coordinate (Ctr) at ($(0: \aa+\fd+\ea)+(-72:\ea+\jd)$);
        
\draw[pcr,pcr fill] (Ctr)
        -- ++(-216:\C)
        -- ++(-144:\C)
        -- ++(-72:\C)
        -- ++(0:\C)
        -- ++(72:\C);
        
\coordinate (Dbr) at ($(0: \aa+\fd+\ea)+(-72:\ea+\jd+\da)+(-144:\da+\lc)$);
        
\draw[pcr,pcr fill] (Dbr)
        -- ++(-288:\D)
        -- ++(-216:\D)
        -- ++(-144:\D)
        -- ++(-72:\D)
        -- ++(0:\D);
        
\coordinate (Ebl) at ($(-108: \aa+\ba)+(-36:\ba+\hc)$);
        
\draw[pcr,pcr fill] (Ebl)
        -- ++(0:\E)
        -- ++(-288:\E)
        -- ++(-216:\E)
        -- ++(-144:\E)
        -- ++(-72:\E);
        
\node[normal vertex] (Ac) at ($(At)+(-90:0.8506*\A)$) {};
\node[normal vertex] (Bc) at ($(Bt)+(-90:0.8506*\B)$) {};
\node[normal vertex] (Cc) at ($(Ctr)+(-162:0.8506*\C)$) {};
\node[normal vertex] (Dc) at ($(Dbr)+(-234:0.8506*\D)$) {};
\node[normal vertex] (Ec) at ($(Ebl)+(-306:0.8506*\E)$) {};

\coordinate (c1) at (0,0);
\coordinate (c2) at (0:\aa+\fd+\ea);
\coordinate (c3) at ($(c2)+(-72:\ea+\jd+\da)$);
\coordinate (c4) at ($(c3)+(-144:\da+\lc+\md+\ca)$);
\coordinate (c5) at ($(c4)+(-216:\ca+\hc+\ba)$);

\node[normal vertex] (a1) at ($($(c1)!0.5!(c2)$)+(90:0.17)$) [label={above}: \footnotesize $a_1$] {};
\node[normal vertex] (a2) at ($($(c2)!0.5!(c3)$)+(18:0.17)$) [label={right}: \footnotesize $a_2$] {};
\node[normal vertex] (a3) at ($($(c3)!0.5!(c4)$)+(-54:0.17)$) [label={below right}: \footnotesize $a_3$] {};
\node[normal vertex] (a4) at ($($(c4)!0.5!(c5)$)+(-126:0.17)$) [label={below left}: \footnotesize $a_4$] {};
\node[normal vertex] (a5) at ($($(c5)!0.5!(c1)$)+(-198:0.17)$) [label={left}: \footnotesize $a_5$] {};

\node[segment label] at ($($(c1)!0.5!(c2)$)+(90:0.06)+(180:0.15)$) {\footnotesize $s_1$};
\node[segment label] at ($($(c2)!0.5!(c3)$)+(18:0.06)+(108:0.12)$) {\footnotesize $s_2$};
\node[segment label] at ($($(c3)!0.5!(c4)$)+(-54:0.06)+(36:0.23)$) {\footnotesize $s_3$};
\node[segment label] at ($($(c4)!0.5!(c5)$)+(-126:0.06)+(-36:-0.12)$) {\footnotesize $s_4$};
\node[segment label] at ($($(c5)!0.5!(c1)$)+(-198:0.06)+(-108:0.12)$) {\footnotesize $s_5$};

\draw[edge] (a1) .. controls ($(c2)+(54:.1)$) .. (a2);
\draw[edge] (a2) .. controls ($(c3)+(-18:.1)$) .. (a3);
\draw[edge] (a3) .. controls ($(c4)+(-90:.1)$) .. (a4);
\draw[edge] (a4) .. controls ($(c5)+(-162:.1)$) .. (a5);
\draw[edge] (a5) .. controls ($(c1)+(-234:.1)$) .. (a1);

\draw[edge] (Ac) -- (a1)
  (Ac) -- (a5)
  (Ac) -- (a4)
  (Ac) -- (Ec)
  (Ac) -- (Bc)
  (Bc) -- (a1)
  (Bc) -- (a2)
  (Bc) -- (Cc)
  (Cc) -- (a2)
  (Cc) -- (a3)
  (Dc) -- (a3)
  (Dc) -- (Ac)
  (Dc) -- (Bc)
  (Dc) -- (Cc)
  (Dc) -- (Ec)
  (Ec) -- (a3)
  (Ec) -- (a4);
        
\end{tikzpicture}

\caption{A regular pentagon contact representation of the graph shown in black.}
\label{fig:regular_pcr}

\end{figure}

Triangle contact representations have been introduced by de Fraysseix
et al.~\cite{de1994triangle}.  They observed that Schnyder woods can
be considered as combinatorial encodings of triangle contact
representations of triangulations and essentially showed that any
Schnyder wood can be used to construct a corresponding triangle
contact system.  They also showed that the triangles can be requested
to be isosceles with a horizontal basis.  Representations with
homothetic triangles can degenerate in the presence of separating
triangles.  Gon\c{c}alves et al.~\cite{gonccalves2011triangle} showed
that 4-connected triangulations admit contact representations with
homothetic triangles.  The proof is an application of Schramm's
\emph{Convex Packing Theorem}, a strong theorem which is based on his
\emph{Monster Packing Theorem}.  A more combinatorial approach to
homothetic triangle contact representations which aims at computing
the representation as the solution of a system of linear equations
related to a Schnyder wood was described by
Felsner~\cite{felsner2009triangle}.  On the basis of this approach
Schrezenmaier reproved the existence of homothetic triangle
representations in his Master's thesis~\cite{schrezenmaier2016zur}.

Representations of graphs using squares or more precisely graphs as a
tool to model packings of squares already appear in classical work of
Brooks et al.~\cite{BSST-40} from 1940.
Schramm~\cite{schramm1993square} proved that every $5$-connected inner
triangulation of a $4$-gon admits a square contact representation.
Again there is a combinatorial approach to this result which aims at
computing the representation as the solution of a system of linear
equations, see Felsner~\cite{felsner2013rectangle}.  In this instance
the role of Schnyder woods is taken by \emph{transversal structures}.
As in the case of homothetic triangles this approach comes with an
algorithm which works well in practice, however, the proof that the
algorithm terminates with a solution is still missing.  On the basis
of the non-algorithmic aspects of this approach
Schrezenmaier~\cite{schrezenmaier2016zur} reproved Schramm's Squaring
Theorem.

In this paper we investigate representations of planar triangulations
as contact graphs of a set of internally disjoint homothetic
pentagons.  From Schramm's \emph{Convex Packing Theorem} it easily
follows that such a representation exists for every triangulation
whose outer face is a $5$-gon.  We relate such representations to
\emph{five color forests}.  The main part of the paper is devoted to
the study of this combinatorial structure. It will become 
clear that five color forests are close relatives of Schnyder woods
and transversal structures. We note in passing that Bernardi and
Fusy~\cite{bernardi-fusy2012} also studied some relatives of Schnyder
woods using 5 colors. Their ``five color trees'', however, only live on
duals of 5-regular planar graphs.  

At the end of the paper we propose an algorithm for computing
homothetic pentagon representations on the basis of systems of
equations and local changes in the corresponding five color
forests. We conjecture that the algorithm always terminates. A proof
of this conjecture would imply a proof for the existence of pentagon
contact representations which is independent of Schramm's Monster
Packing. The idea of looking for pentagon contact representations and
a substantial part of the work originate in the Bachelor's Thesis of
Steiner~\cite{steiner2016existenz}.

\subsection{The existence of pentagon contact representations}

The existence of regular pentagon contact representations for every
inner triangulation of a $5$-gon can be shown using the following
general result about contact representations by Schramm.

\begin{theorem}[Convex Packing Theorem \cite{schramm2007combinatorically}]
  Let~$G$ be an inner triangulation of the triangle~$abc$.  Further
  let~$C$ be a simple closed curve in the plane partitioned into three
  arcs~$\mathcal{P}_a,\mathcal{P}_b,\mathcal{P}_c$, and for each
  inner vertex~$v$ of~$G$ let~$\mathcal{P}_v$ be a convex set in
  the plane containing more than one point.  Then there exists a
  contact representation of a supergraph of~$G$ (on the same vertex
  set, but possibly with more edges) where each inner vertex~$v$ is
  represented by a single point or a homothetic copy of its
  prototype~$\mathcal{P}_v$ and each outer vertex~$w$ by the
  arc~$\mathcal{P}_w$.
\end{theorem}

\begin{theorem}
  Let~$G$ be an inner triangulation of the $5$-gon $a_1,\dotsc,a_5$.
  Then there exists a regular pentagon contact representation of~$G$.
\end{theorem}

\begin{proof}
  By adding the edges~$a_1a_3$ and~$a_1a_4$ in the outer face of~$G$,
  it becomes a triangulation~$G'$ with outer face~$a_1a_3a_4$.  We
  define the
  arcs~$\mathcal{P}_{a_1},\mathcal{P}_{a_3},\mathcal{P}_{a_4}$ to be
  extensions of the upper, lower left and lower right edge
  of a regular pentagon~$A$ with a horizontal edge at the top, respectively, such
  that~${\mathcal{P}_{a_1}\cup\mathcal{P}_{a_3}\cup\mathcal{P}_{a_4}}$
  forms a triangle and therefore a simple closed curve.  We define the
  convex sets~$\mathcal{P}_{a_2}$ and~$\mathcal{P}_{a_5}$ to be line
  segments parallel to the upper right and upper left edge of the
  pentagon~$A$ (see Fig.~\ref{fig:schramm_constr}).  Finally, for each
  inner vertex~$v$ of~$G$ let~$\mathcal{P}_v$ be a regular pentagon
  with a horizontal edge at the bottom.

\begin{figure}

\centering

\tikzstyle{edge}=[thick]

\tikzstyle{pcr}=[color=gray!70]

\begin{tikzpicture}[scale=2.5]

\coordinate (A) at (0,0);
\coordinate (B) at (144:2);
\coordinate (C) at (36:2);

\draw[edge] (A) -- (B) -- (C) -- (A);

\coordinate (L1) at ($(A)!0.35!(B)$);
\coordinate (L2) at ($(L1)+(72:3)$);
\coordinate (L3) at (intersection of B--C and L1--L2);
\coordinate (L4) at ($(L1)!0.1!(L3)$);
\coordinate (L5) at ($(L1)!0.9!(L3)$);

\coordinate (R1) at ($(A)!0.55!(C)$);
\coordinate (R2) at ($(R1)+(108:3)$);
\coordinate (R3) at (intersection of B--C and R1--R2);
\coordinate (R4) at ($(R1)!0.1!(R3)$);
\coordinate (R5) at ($(R1)!0.9!(R3)$);

\draw[edge] (L4) -- (L5);
\draw[edge] (R4) -- (R5);

\node at ($(B)!0.5!(C)$) [label={above}:\footnotesize $\mathcal{P}_{a_1}$] {};
\node at ($(A)!0.5!(C)$) [label={below}:\footnotesize $\mathcal{P}_{a_3}$] {};
\node at ($(A)!0.5!(B)$) [label={left}:\footnotesize $\mathcal{P}_{a_4}$] {};
\node at ($(R4)!0.5!(R5)$) [label={right}:\footnotesize $\mathcal{P}_{a_2}$] {};
\node at ($(L4)!0.5!(L5)$) [label={left}:\footnotesize $\mathcal{P}_{a_5}$] {};
        
\end{tikzpicture}

\caption{Prototypes for the five outer vertices of $G$.}

\label{fig:schramm_constr}

\end{figure}
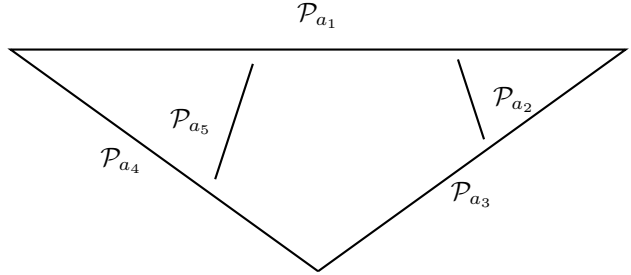

Now we can apply the Convex Packing Theorem. The result is a contact
representation of a supergraph of~$G'$ where~$a_1,a_3,a_4$ are
represented by~$\mathcal{P}_{a_1},\mathcal{P}_{a_3},\mathcal{P}_{a_4}$
and every inner vertex~$v$ by a homothetic copy of~$\mathcal{P}_v$ or
by a single point.

We claim that in this contact representation of~$G'$ none of the
homothetic copies of the prototypes is degenerate to a single point.
Assume there is a degenerate copy in the contact representation.
Let~$H$ be a maximal connected component of the subgraph of~$G'$
induced by the vertices whose pentagons are degenerate to a single
point. Since the line segments corresponding to the three outer
vertices are not degenerate,~$H$ has to be bounded by a cycle~$C$ of
vertices whose pentagons or line segments are not
degenerate.  In the contact representation all vertices of~$H$ are
represented by the same point and therefore all pentagons and
line segments representing the vertices of~$C$ have a contact with
this point, i.e., they meet at the point.
But for geometric reasons at most two of these can meet in a single
point.  Thus~$C$ is a $2$-cycle, in contradiction to our definition of
inner triangulations that does not allow multiple edges.

After cutting the segments~$\mathcal{P}_{a_1}$,~$\mathcal{P}_{a_3}$
and~$\mathcal{P}_{a_4}$, the vertices~$a_1,\dotsc,a_5$ are represented
by a pentagon of the required form and we obtain a regular pentagon
contact representation of~$G$.
\end{proof}

\section{Five Color Forests}

In this section $G$ will always be an inner triangulation with outer face
$a_1,\dotsc,a_5$ in clockwise order.  The set~$1,\dotsc,5$ of colors
is to be understood as representatives modulo~$5$, e.g., $-1$ and~$4$
denote the same color.

\begin{definition} A \emph{five color forest} of~$G$ is an orientation and
coloring of the inner edges of~$G$ in the colors~$1,\dotsc,5$ with the
following properties (see Fig.~\ref{fig:def_fcf} for an illustration):
\begin{figure}

\centering

\tikzstyle{vertex}=[circle,fill,scale=0.5]
\tikzstyle{out edge}=[-latex',thick]
\tikzstyle{in edge}=[latex'-]

\begin{tikzpicture}[scale=1.1]

\node[vertex] (v) at (0,0) {};

\node[color1] (v1) at (90:1.2) {$1$};
\node[color2] (v2) at (18:1.2) {$2$};
\node[color3] (v3) at (306:1.2) {$3$};
\node[color4] (v4) at (234:1.2) {$4$};
\node[color5] (v5) at (162:1.2) {$5$};

\node[color1light] at (270:.9) {\footnotesize $1$};
\node[color2light] at (198:.9) {\footnotesize $2$};
\node[color3light] at (126:.9) {\footnotesize $3$};
\node[color4light] at (54:.9) {\footnotesize $4$};
\node[color5light] at (342:.9) {\footnotesize $5$};

\draw[out edge,color1] (v) -- (v1);
\draw[out edge,color2] (v) -- (v2);
\draw[out edge,color3] (v) -- (v3);
\draw[out edge,color4] (v) -- (v4);
\draw[out edge,color5] (v) -- (v5);

\draw[in edge,color3light] (v) -- (114:.7);
\draw[in edge,color3light] (v) -- (138:.7);

\draw[in edge,color2light] (v) -- (186:.7);
\draw[in edge,color2light] (v) -- (210:.7);

\draw[in edge,color1light] (v) -- (258:.7);
\draw[in edge,color1light] (v) -- (282:.7);

\draw[in edge,color5light] (v) -- (330:.7);
\draw[in edge,color5light] (v) -- (354:.7);

\draw[in edge,color4light] (v) -- (42:.7);
\draw[in edge,color4light] (v) -- (66:.7);

\end{tikzpicture}
\qquad
\begin{tikzpicture}

\node[vertex,color1] (a1) at (90:1.2) [label={above}:$a_1$] {};
\node[vertex,color2] (a2) at (18:1.2) [label={right}:$a_2$] {};
\node[vertex,color3] (a3) at (306:1.2) [label={below right}:$a_3$] {};
\node[vertex,color4] (a4) at (234:1.2) [label={below left}:$a_4$] {};
\node[vertex,color5] (a5) at (162:1.2) [label={left}:$a_5$] {};

\draw (a1) -- (a2) -- (a3) -- (a4) -- (a5) -- (a1);

\draw[in edge,color1light] (a1) -- +(260:.5);
\draw[in edge,color1light] (a1) -- +(280:.5);

\draw[in edge,color2light] (a2) -- +(188:.5);
\draw[in edge,color2light] (a2) -- +(208:.5);

\draw[in edge,color3light] (a3) -- +(116:.5);
\draw[in edge,color3light] (a3) -- +(136:.5);

\draw[in edge,color4light] (a4) -- +(44:.5);
\draw[in edge,color4light] (a4) -- +(64:.5);

\draw[in edge,color5light] (a5) -- +(332:.5);
\draw[in edge,color5light] (a5) -- +(352:.5);

\node[color1light] at ($(a1)+(270:.7)$) {\footnotesize $1$};
\node[color2light] at ($(a2)+(198:.7)$) {\footnotesize $2$};
\node[color3light] at ($(a3)+(126:.7)$) {\footnotesize $3$};
\node[color4light] at ($(a4)+(54:.7)$) {\footnotesize $4$};
\node[color5light] at ($(a5)+(342:.7)$) {\footnotesize $5$};

\end{tikzpicture}

\caption{The local conditions of a five color forest}

\label{fig:def_fcf}

\end{figure}
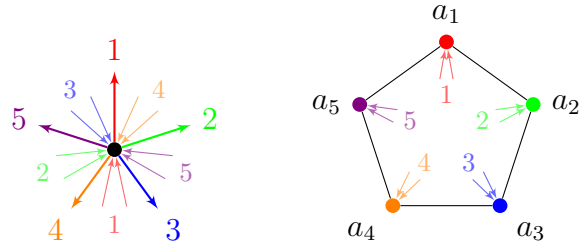
\begin{enumerate}[(F1)]
\item \label{item:outer_edges} All edges incident to~$a_i$ are
  oriented towards~$a_i$ and colored in the color~$i$.
\item \label{item:inner_vertex_blocks} For each inner vertex~$v$, the
  incoming edges build five (possibly empty)
  blocks $B_i$, ${i=1,\dotsc,5}$, of edges of color~$i$ and the
  clockwise order of these blocks is~$B_1,\dotsc,B_5$.  Moreover~$v$
  has at most one outgoing edge of color~$i$ and such an edge has to
  be located between the blocks~$B_{i+2}$ and~$B_{i-2}$.
\item \label{item:no_three_empty} For every inner vertex and
  for~${i=1,\dotsc,5}$ the block~$B_i$ is nonempty or one of the
  outgoing edges of colors~$i-2$ and~$i+2$ exists.
\end{enumerate}
\end{definition}

The following theorem shows the key correspondence between five color
forests and pentagon contact representations.

\begin{theorem} \label{thm:regular_induce_fcf} Every regular pentagon
contact representation induces a five color forest on its contact
graph.
\end{theorem}

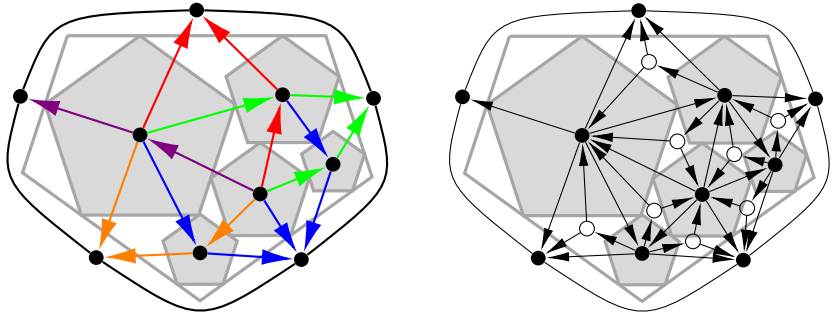
\begin{figure}

\centering

\tikzstyle{normal vertex}=[circle,fill,scale=0.5]
\tikzstyle{stack vertex}=[circle,draw,scale=0.5]
\tikzstyle{undirected edge}=[]
\tikzstyle{colored undirected edge}=[thick]
\tikzstyle{edge}=[-{Latex[scale length=1.8,scale width=1]}]
\tikzstyle{colored edge}=[thick,-{Latex[scale length=1.8,scale width=1.2]}]

\tikzstyle{pcr}=[very thick,color=gray!70]
\tikzstyle{pcr fill}=[fill=gray!30]

\begin{tikzpicture}[scale=3.4]

\def\A{0.4546294640647149}
\def\B{0.2707279827787506}
\def\C{0.1507367088399636}
\def\D{0.2356059251438577}
\def\E{0.1775675733633562}
\def\aa{0.2809764610791428}
\def\ab{0.4546294640647149}
\def\ba{0.2809764610791428}
\def\bb{0.4546294640647149}
\def\ca{0.1097427956383953}
\def\cb{0.1775675733633562}
\def\da{0.09316040941539433}
\def\db{0.1507367088399636}
\def\ea{0.167319095062964}
\def\eb{0.2707279827787506}
\def\fa{0.1136573660161787}
\def\fb{0.2707279827787506}
\def\fc{0.4546294640647149}
\def\fd{0.5517044438578933}
\def\ga{0.04853748989658916}
\def\gb{0.1570706167625718}
\def\gc{0.2356059251438577}
\def\gd{0.3026830864523393}
\def\ha{0.1775675733633562}
\def\hb{0.171233665440748}
\def\hc{0.4585440344424989}
\def\hd{0.4546294640647149}
\def\ia{0.006333907922608149}
\def\ib{0.1416978993119834}
\def\ic{0.1519463776123757}
\def\id{0.2356059251438577}
\def\ja{0.07415868564756978}
\def\jb{0.1507367088399636}
\def\jc{0.2707279827787506}
\def\jd{0.3180558039029276}
\def\ka{0.09828464856559038}
\def\kb{0.07657802319239386}
\def\kc{0.2356059251438577}
\def\kd{0.2221904928821613}
\def\la{0.1507367088399636}
\def\lb{0.05245206027437295}
\def\lc{0.296349178529731}
\def\ld{0.2356059251438577}
\def\ma{0.03586967405137269}
\def\mb{0.1775675733633562}
\def\mc{0.2356059251438577}
\def\md{0.3231800430531234}

\draw[pcr] (0,0) -- ++(0:   \aa+\fd+\ea)
            -- ++(-72: \ea+\jd+\da)
            -- ++(-144:\da+\lc+\md+\ca)
            -- ++(-216:\ca+\hc+\ba)
            -- ++(-288:\ba+\aa);
            
\coordinate (At) at (0: \aa);
            
\draw[pcr,pcr fill] (At)
        -- ++(-144:\A)
        -- ++(-72:\A)
        -- ++(0:\A)
        -- ++(72:\A)
        -- ++(144:\A);
        
\coordinate (Bt) at (0: \aa+\fd);       
        
\draw[pcr,pcr fill] (Bt)
        -- ++(-144:\B)
        -- ++(-72:\B)
        -- ++(0:\B)
        -- ++(72:\B)
        -- ++(144:\B);
        
\coordinate (Ctr) at ($(0: \aa+\fd+\ea)+(-72:\ea+\jd)$);
        
\draw[pcr,pcr fill] (Ctr)
        -- ++(-216:\C)
        -- ++(-144:\C)
        -- ++(-72:\C)
        -- ++(0:\C)
        -- ++(72:\C);
        
\coordinate (Dbr) at ($(0: \aa+\fd+\ea)+(-72:\ea+\jd+\da)+(-144:\da+\lc)$);
        
\draw[pcr,pcr fill] (Dbr)
        -- ++(-288:\D)
        -- ++(-216:\D)
        -- ++(-144:\D)
        -- ++(-72:\D)
        -- ++(0:\D);
        
\coordinate (Ebl) at ($(-108: \aa+\ba)+(-36:\ba+\hc)$);
        
\draw[pcr,pcr fill] (Ebl)
        -- ++(0:\E)
        -- ++(-288:\E)
        -- ++(-216:\E)
        -- ++(-144:\E)
        -- ++(-72:\E);
        
\node[normal vertex] (Ac) at ($(At)+(-90:0.8506*\A)$) {};
\node[normal vertex] (Bc) at ($(Bt)+(-90:0.8506*\B)$) {};
\node[normal vertex] (Cc) at ($(Ctr)+(-162:0.8506*\C)$) {};
\node[normal vertex] (Dc) at ($(Dbr)+(-234:0.8506*\D)$) {};
\node[normal vertex] (Ec) at ($(Ebl)+(-306:0.8506*\E)$) {};

\coordinate (c1) at (0,0);
\coordinate (c2) at (0:\aa+\fd+\ea);
\coordinate (c3) at ($(c2)+(-72:\ea+\jd+\da)$);
\coordinate (c4) at ($(c3)+(-144:\da+\lc+\md+\ca)$);
\coordinate (c5) at ($(c4)+(-216:\ca+\hc+\ba)$);

\node[normal vertex] (a1) at ($($(c1)!0.5!(c2)$)+(90:0.1)$) {};
\node[normal vertex] (a2) at ($($(c2)!0.5!(c3)$)+(18:0.1)$) {};
\node[normal vertex] (a3) at ($($(c3)!0.5!(c4)$)+(-54:0.1)$) {};
\node[normal vertex] (a4) at ($($(c4)!0.5!(c5)$)+(-126:0.1)$) {};
\node[normal vertex] (a5) at ($($(c5)!0.5!(c1)$)+(-198:0.1)$) {};

\draw[colored undirected edge] (a1) .. controls ($(c2)+(54:.1)$) .. (a2);
\draw[colored undirected edge] (a2) .. controls ($(c3)+(-18:.1)$) .. (a3);
\draw[colored undirected edge] (a3) .. controls ($(c4)+(-90:.1)$) .. (a4);
\draw[colored undirected edge] (a4) .. controls ($(c5)+(-162:.1)$) .. (a5);
\draw[colored undirected edge] (a5) .. controls ($(c1)+(-234:.1)$) .. (a1);

\draw[colored edge,color1] (Ac) -- (a1);
\draw[colored edge,color5] (Ac) -- (a5);
\draw[colored edge,color4] (Ac) -- (a4);
\draw[colored edge,color3] (Ac) -- (Ec);
\draw[colored edge,color2] (Ac) -- (Bc);
\draw[colored edge,color1] (Bc) -- (a1);
\draw[colored edge,color2] (Bc) -- (a2);
\draw[colored edge,color3] (Bc) -- (Cc);
\draw[colored edge,color2] (Cc) -- (a2);
\draw[colored edge,color3] (Cc) -- (a3);
\draw[colored edge,color3] (Dc) -- (a3);
\draw[colored edge,color5] (Dc) -- (Ac);
\draw[colored edge,color1] (Dc) -- (Bc);
\draw[colored edge,color2] (Dc) -- (Cc);
\draw[colored edge,color4] (Dc) -- (Ec);
\draw[colored edge,color3] (Ec) -- (a3);
\draw[colored edge,color4] (Ec) -- (a4);
        
\end{tikzpicture}
\quad
\begin{tikzpicture}[scale=3.4]

\def\A{0.4546294640647149}
\def\B{0.2707279827787506}
\def\C{0.1507367088399636}
\def\D{0.2356059251438577}
\def\E{0.1775675733633562}
\def\aa{0.2809764610791428}
\def\ab{0.4546294640647149}
\def\ba{0.2809764610791428}
\def\bb{0.4546294640647149}
\def\ca{0.1097427956383953}
\def\cb{0.1775675733633562}
\def\da{0.09316040941539433}
\def\db{0.1507367088399636}
\def\ea{0.167319095062964}
\def\eb{0.2707279827787506}
\def\fa{0.1136573660161787}
\def\fb{0.2707279827787506}
\def\fc{0.4546294640647149}
\def\fd{0.5517044438578933}
\def\ga{0.04853748989658916}
\def\gb{0.1570706167625718}
\def\gc{0.2356059251438577}
\def\gd{0.3026830864523393}
\def\ha{0.1775675733633562}
\def\hb{0.171233665440748}
\def\hc{0.4585440344424989}
\def\hd{0.4546294640647149}
\def\ia{0.006333907922608149}
\def\ib{0.1416978993119834}
\def\ic{0.1519463776123757}
\def\id{0.2356059251438577}
\def\ja{0.07415868564756978}
\def\jb{0.1507367088399636}
\def\jc{0.2707279827787506}
\def\jd{0.3180558039029276}
\def\ka{0.09828464856559038}
\def\kb{0.07657802319239386}
\def\kc{0.2356059251438577}
\def\kd{0.2221904928821613}
\def\la{0.1507367088399636}
\def\lb{0.05245206027437295}
\def\lc{0.296349178529731}
\def\ld{0.2356059251438577}
\def\ma{0.03586967405137269}
\def\mb{0.1775675733633562}
\def\mc{0.2356059251438577}
\def\md{0.3231800430531234}

\draw[pcr] (0,0) -- ++(0:   \aa+\fd+\ea)
            -- ++(-72: \ea+\jd+\da)
            -- ++(-144:\da+\lc+\md+\ca)
            -- ++(-216:\ca+\hc+\ba)
            -- ++(-288:\ba+\aa);
            
\coordinate (At) at (0: \aa);
            
\draw[pcr,pcr fill] (At)
        -- ++(-144:\A)
        -- ++(-72:\A)
        -- ++(0:\A)
        -- ++(72:\A)
        -- ++(144:\A);
        
\coordinate (Bt) at (0: \aa+\fd);       
        
\draw[pcr,pcr fill] (Bt)
        -- ++(-144:\B)
        -- ++(-72:\B)
        -- ++(0:\B)
        -- ++(72:\B)
        -- ++(144:\B);
        
\coordinate (Ctr) at ($(0: \aa+\fd+\ea)+(-72:\ea+\jd)$);
        
\draw[pcr,pcr fill] (Ctr)
        -- ++(-216:\C)
        -- ++(-144:\C)
        -- ++(-72:\C)
        -- ++(0:\C)
        -- ++(72:\C);
        
\coordinate (Dbr) at ($(0: \aa+\fd+\ea)+(-72:\ea+\jd+\da)+(-144:\da+\lc)$);
        
\draw[pcr,pcr fill] (Dbr)
        -- ++(-288:\D)
        -- ++(-216:\D)
        -- ++(-144:\D)
        -- ++(-72:\D)
        -- ++(0:\D);
        
\coordinate (Ebl) at ($(-108: \aa+\ba)+(-36:\ba+\hc)$);
        
\draw[pcr,pcr fill] (Ebl)
        -- ++(0:\E)
        -- ++(-288:\E)
        -- ++(-216:\E)
        -- ++(-144:\E)
        -- ++(-72:\E);
        
\node[normal vertex] (Ac) at ($(At)+(-90:0.8506*\A)$) {};
\node[normal vertex] (Bc) at ($(Bt)+(-90:0.8506*\B)$) {};
\node[normal vertex] (Cc) at ($(Ctr)+(-162:0.8506*\C)$) {};
\node[normal vertex] (Dc) at ($(Dbr)+(-234:0.8506*\D)$) {};
\node[normal vertex] (Ec) at ($(Ebl)+(-306:0.8506*\E)$) {};

\coordinate (c1) at (0,0);
\coordinate (c2) at (0:\aa+\fd+\ea);
\coordinate (c3) at ($(c2)+(-72:\ea+\jd+\da)$);
\coordinate (c4) at ($(c3)+(-144:\da+\lc+\md+\ca)$);
\coordinate (c5) at ($(c4)+(-216:\ca+\hc+\ba)$);

\node[normal vertex] (a1) at ($($(c1)!0.5!(c2)$)+(90:0.1)$) {};
\node[normal vertex] (a2) at ($($(c2)!0.5!(c3)$)+(18:0.1)$) {};
\node[normal vertex] (a3) at ($($(c3)!0.5!(c4)$)+(-54:0.1)$) {};
\node[normal vertex] (a4) at ($($(c4)!0.5!(c5)$)+(-126:0.1)$) {};
\node[normal vertex] (a5) at ($($(c5)!0.5!(c1)$)+(-198:0.1)$) {};

\node[stack vertex] (sf) at (.54,-.1) {};
\node[stack vertex] (sg) at (.65,-.41) {};
\node[stack vertex] (sh) at (.3,-.75) {};
\node[stack vertex] (si) at (.56,-.68) {};
\node[stack vertex] (sj) at (1.04,-.33) {};
\node[stack vertex] (sk) at (.87,-.46) {};
\node[stack vertex] (sl) at (.92,-.67) {};
\node[stack vertex] (sm) at (.71,-.8) {};

\draw[undirected edge] (a1) .. controls ($(c2)+(54:.1)$) .. (a2);
\draw[undirected edge] (a2) .. controls ($(c3)+(-18:.1)$) .. (a3);
\draw[undirected edge] (a3) .. controls ($(c4)+(-90:.1)$) .. (a4);
\draw[undirected edge] (a4) .. controls ($(c5)+(-162:.1)$) .. (a5);
\draw[undirected edge] (a5) .. controls ($(c1)+(-234:.1)$) .. (a1);

\draw[edge] (Ac) -- (a1);
\draw[edge] (Ac) -- (a5);
\draw[edge] (Ac) -- (a4);
\draw[edge] (Ac) -- (Ec);
\draw[edge] (Ac) -- (Bc);
\draw[edge] (Bc) -- (a1);
\draw[edge] (Bc) -- (a2);
\draw[edge] (Bc) -- (Cc);
\draw[edge] (Cc) -- (a2);
\draw[edge] (Cc) -- (a3);
\draw[edge] (Dc) -- (a3);
\draw[edge] (Dc) -- (Ac);
\draw[edge] (Dc) -- (Bc);
\draw[edge] (Dc) -- (Cc);
\draw[edge] (Dc) -- (Ec);
\draw[edge] (Ec) -- (a3);
\draw[edge] (Ec) -- (a4);

\draw[edge] (Bc) -- (sf);
\draw[edge] (sf) -- (Ac);
\draw[edge] (sf) -- (a1);

\draw[edge] (Bc) -- (sg);
\draw[edge] (sg) -- (Ac);
\draw[edge] (sg) -- (Dc);

\draw[edge] (Ec) -- (sh);
\draw[edge] (sh) -- (Ac);
\draw[edge] (sh) -- (a4);

\draw[edge] (Ec) -- (si);
\draw[edge] (si) -- (Dc);
\draw[edge] (si) -- (Ac);

\draw[edge] (Cc) -- (sj);
\draw[edge] (sj) -- (Bc);
\draw[edge] (sj) -- (a2);

\draw[edge] (Cc) -- (sk);
\draw[edge] (sk) -- (Bc);
\draw[edge] (sk) -- (Dc);

\draw[edge] (Cc) -- (sl);
\draw[edge] (sl) -- (Dc);
\draw[edge] (sl) -- (a3);

\draw[edge] (Ec) -- (sm);
\draw[edge] (sm) -- (Dc);
\draw[edge] (sm) -- (a3);
        
\end{tikzpicture}

\caption{The induced five color forest and $\alpha$-orientation of a pentagon contact representation.}

\label{fig:induced_fcf_ao}

\end{figure}

\begin{proof} 
  Let~$\mathcal{S}$ be a regular pentagon contact representation
  of~${G=G^\ast(\mathcal{S})}$.  We color the corners of all pentagons
  of~$\mathcal{S}$ with the colors~$1,\dotsc,5$ in clockwise order,
  starting with color~$1$ at the corner opposite to the horizontal
  segment.  Let~$e$ be an inner edge of~$G$.  If~$e$ corresponds to
  the contact of a corner of a pentagon~$A$ and a side of a
  pentagon~$B$ in~$\mathcal{S}$, then we orient the edge~$e$ from the
  vertex corresponding to~$A$ to the vertex corresponding to~$B$ and
  color it in the color of the corner of~$A$ involved in the contact (see Fig.~\ref{fig:induced_fcf_ao} (left)).
  A contact of two pentagon corners can be interpreted in two ways as
  a corner-side contact with infinitesimal distance to the other
  corner.  We choose one of these interpretations and proceed as
  before.  Hence, the five color forest induced by a degenerate
  pentagon contact representation is not unique.
  Figure~\ref{fig:induced_fcf_degen} shows an example.

\begin{figure}

\centering

\tikzstyle{normal vertex}=[circle,fill,scale=0.5]
\tikzstyle{stack vertex}=[circle,draw,scale=0.5]
\tikzstyle{undirected edge}=[]
\tikzstyle{colored undirected edge}=[thick]
\tikzstyle{edge}=[-{Latex[scale length=1.8,scale width=1]}]
\tikzstyle{colored edge}=[thick,-{Latex[scale length=1.8,scale width=1.2]}]

\tikzstyle{pcr}=[very thick,color=gray!70]
\tikzstyle{pcr fill}=[fill=gray!30]

\begin{tikzpicture}[scale=3.4]

\def\ax{0.1154318675}
\def\ay{0}
\def\bx{0}
\def\by{-0.35526275827}
\def\cx{0.6977954475}
\def\cy{-0.862240826728}
\def\dx{1}
\def\dy{-0.642676367179}
\def\ex{0.79118179}
\def\ey{0}

\def\Aax{0.6977954475}
\def\Aay{-0.287413608909}
\def\Abx{0.45330682875}
\def\Aby{-0.465044988044}
\def\Acx{0.54669317125}
\def\Acy{-0.752458596954}
\def\Adx{0.84889772375}
\def\Ady{-0.752458596954}
\def\Aex{0.94228406625}
\def\Aey{-0.465044988044}

\def\Bax{0.3022045525}
\def\Bay{0}
\def\Bbx{0.05771593375}
\def\Bby{-0.177631379135}
\def\Bcx{0.15110227625}
\def\Bcy{-0.465044988044}
\def\Bdx{0.45330682875}
\def\Bdy{-0.465044988044}
\def\Bex{0.54669317125}
\def\Bey{-0.177631379135}

\def\Cax{0.6757499225}
\def\Cay{0}
\def\Cbx{0.52464764625}
\def\Cby{-0.109782229774}
\def\Ccx{0.58236358}
\def\Ccy{-0.287413608909}
\def\Cdx{0.769136265}
\def\Cdy{-0.287413608909}
\def\Cex{0.82685219875}
\def\Cey{-0.109782229774}

\draw[pcr] (\ax,\ay) -- (\bx,\by) -- (\cx,\cy) -- (\dx,\dy) -- (\ex,\ey) -- cycle;

\draw[pcr, pcr fill] (\Aax,\Aay) -- (\Abx,\Aby) -- (\Acx,\Acy) -- (\Adx,\Ady) -- (\Aex,\Aey) -- cycle;
\draw[pcr, pcr fill] (\Bax,\Bay) -- (\Bbx,\Bby) -- (\Bcx,\Bcy) -- (\Bdx,\Bdy) -- (\Bex,\Bey) -- cycle;
\draw[pcr, pcr fill] (\Cax,\Cay) -- (\Cbx,\Cby) -- (\Ccx,\Ccy) -- (\Cdx,\Cdy) -- (\Cex,\Cey) -- cycle;
            
\node[normal vertex] (Ac) at ($.2*(\Aax+\Abx+\Acx+\Adx+\Aex, \Aay+\Aby+\Acy+\Ady+\Aey)$) {};
\node[normal vertex] (Bc) at ($.2*(\Bax+\Bbx+\Bcx+\Bdx+\Bex, \Bay+\Bby+\Bcy+\Bdy+\Bey)$) {};
\node[normal vertex] (Cc) at ($.2*(\Cax+\Cbx+\Ccx+\Cdx+\Cex, \Cay+\Cby+\Ccy+\Cdy+\Cey)$) {};

\coordinate (c1) at (\ax,\ay);
\coordinate (c2) at (\bx,\by);
\coordinate (c3) at (\cx,\cy);
\coordinate (c4) at (\dx,\dy);
\coordinate (c5) at (\ex,\ey);

\node[normal vertex] (a5) at ($($(c1)!0.5!(c2)$)+(-198:0.1)$) {};
\node[normal vertex] (a4) at ($($(c2)!0.5!(c3)$)+(-126:0.1)$) {};
\node[normal vertex] (a3) at ($($(c3)!0.5!(c4)$)+(-54:0.1)$) {};
\node[normal vertex] (a2) at ($($(c4)!0.5!(c5)$)+(18:0.1)$) {};
\node[normal vertex] (a1) at ($($(c5)!0.5!(c1)$)+(90:0.1)$) {};

\draw[colored undirected edge] (a1) .. controls ($(c5)+(54:.1)$) .. (a2);
\draw[colored undirected edge] (a2) .. controls ($(c4)+(-18:.1)$) .. (a3);
\draw[colored undirected edge] (a3) .. controls ($(c3)+(-90:.1)$) .. (a4);
\draw[colored undirected edge] (a4) .. controls ($(c2)+(-162:.1)$) .. (a5);
\draw[colored undirected edge] (a5) .. controls ($(c1)+(-234:.1)$) .. (a1);

\draw[colored edge,color1] (Bc) -- (a1);
\draw[colored edge,color1] (Cc) -- (a1);
\draw[colored edge,color1] (Ac) -- (Cc);
\draw[colored edge,color2] (Cc) -- (a2);
\draw[colored edge,color2] (Ac) -- (a2);
\draw[colored edge,color2] (Bc) -- (Cc);
\draw[colored edge,color3] (Ac) -- (a3);
\draw[colored edge,color4] (Ac) -- (a4);
\draw[colored edge,color4] (Bc) -- (a4);
\draw[colored edge,color5] (Bc) -- (a5);

\draw[colored edge,color3] (Bc) -- (Ac);
        
\end{tikzpicture}
\qquad
\begin{tikzpicture}[scale=3.4]

\def\ax{0.1154318675}
\def\ay{0}
\def\bx{0}
\def\by{-0.35526275827}
\def\cx{0.6977954475}
\def\cy{-0.862240826728}
\def\dx{1}
\def\dy{-0.642676367179}
\def\ex{0.79118179}
\def\ey{0}

\def\Aax{0.6977954475}
\def\Aay{-0.287413608909}
\def\Abx{0.45330682875}
\def\Aby{-0.465044988044}
\def\Acx{0.54669317125}
\def\Acy{-0.752458596954}
\def\Adx{0.84889772375}
\def\Ady{-0.752458596954}
\def\Aex{0.94228406625}
\def\Aey{-0.465044988044}

\def\Bax{0.3022045525}
\def\Bay{0}
\def\Bbx{0.05771593375}
\def\Bby{-0.177631379135}
\def\Bcx{0.15110227625}
\def\Bcy{-0.465044988044}
\def\Bdx{0.45330682875}
\def\Bdy{-0.465044988044}
\def\Bex{0.54669317125}
\def\Bey{-0.177631379135}

\def\Cax{0.6757499225}
\def\Cay{0}
\def\Cbx{0.52464764625}
\def\Cby{-0.109782229774}
\def\Ccx{0.58236358}
\def\Ccy{-0.287413608909}
\def\Cdx{0.769136265}
\def\Cdy{-0.287413608909}
\def\Cex{0.82685219875}
\def\Cey{-0.109782229774}

\draw[pcr] (\ax,\ay) -- (\bx,\by) -- (\cx,\cy) -- (\dx,\dy) -- (\ex,\ey) -- cycle;

\draw[pcr, pcr fill] (\Aax,\Aay) -- (\Abx,\Aby) -- (\Acx,\Acy) -- (\Adx,\Ady) -- (\Aex,\Aey) -- cycle;
\draw[pcr, pcr fill] (\Bax,\Bay) -- (\Bbx,\Bby) -- (\Bcx,\Bcy) -- (\Bdx,\Bdy) -- (\Bex,\Bey) -- cycle;
\draw[pcr, pcr fill] (\Cax,\Cay) -- (\Cbx,\Cby) -- (\Ccx,\Ccy) -- (\Cdx,\Cdy) -- (\Cex,\Cey) -- cycle;
            
\node[normal vertex] (Ac) at ($.2*(\Aax+\Abx+\Acx+\Adx+\Aex, \Aay+\Aby+\Acy+\Ady+\Aey)$) {};
\node[normal vertex] (Bc) at ($.2*(\Bax+\Bbx+\Bcx+\Bdx+\Bex, \Bay+\Bby+\Bcy+\Bdy+\Bey)$) {};
\node[normal vertex] (Cc) at ($.2*(\Cax+\Cbx+\Ccx+\Cdx+\Cex, \Cay+\Cby+\Ccy+\Cdy+\Cey)$) {};

\coordinate (c1) at (\ax,\ay);
\coordinate (c2) at (\bx,\by);
\coordinate (c3) at (\cx,\cy);
\coordinate (c4) at (\dx,\dy);
\coordinate (c5) at (\ex,\ey);

\node[normal vertex] (a5) at ($($(c1)!0.5!(c2)$)+(-198:0.1)$) {};
\node[normal vertex] (a4) at ($($(c2)!0.5!(c3)$)+(-126:0.1)$) {};
\node[normal vertex] (a3) at ($($(c3)!0.5!(c4)$)+(-54:0.1)$) {};
\node[normal vertex] (a2) at ($($(c4)!0.5!(c5)$)+(18:0.1)$) {};
\node[normal vertex] (a1) at ($($(c5)!0.5!(c1)$)+(90:0.1)$) {};

\draw[colored undirected edge] (a1) .. controls ($(c5)+(54:.1)$) .. (a2);
\draw[colored undirected edge] (a2) .. controls ($(c4)+(-18:.1)$) .. (a3);
\draw[colored undirected edge] (a3) .. controls ($(c3)+(-90:.1)$) .. (a4);
\draw[colored undirected edge] (a4) .. controls ($(c2)+(-162:.1)$) .. (a5);
\draw[colored undirected edge] (a5) .. controls ($(c1)+(-234:.1)$) .. (a1);

\draw[colored edge,color1] (Bc) -- (a1);
\draw[colored edge,color1] (Cc) -- (a1);
\draw[colored edge,color1] (Ac) -- (Cc);
\draw[colored edge,color2] (Cc) -- (a2);
\draw[colored edge,color2] (Ac) -- (a2);
\draw[colored edge,color2] (Bc) -- (Cc);
\draw[colored edge,color3] (Ac) -- (a3);
\draw[colored edge,color4] (Ac) -- (a4);
\draw[colored edge,color4] (Bc) -- (a4);
\draw[colored edge,color5] (Bc) -- (a5);

\draw[colored edge,color5] (Ac) -- (Bc);
        
\end{tikzpicture}

\caption{The two induced five color forests of a pentagon contact representation with an exceptional touching.}

\label{fig:induced_fcf_degen}

\end{figure}
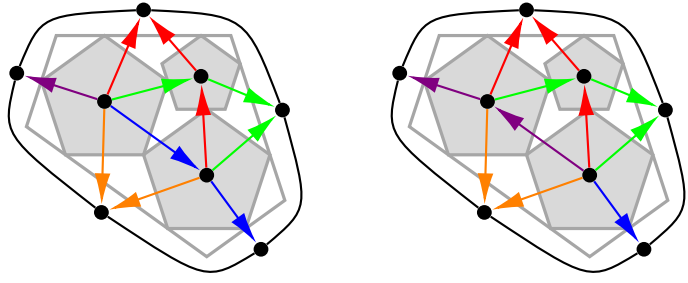

We claim that this coloring and orientation of~$G$ fulfills the
properties of a five color forest.  Property~(F\ref{item:outer_edges})
immediately follows from the construction.  Now consider
property~(F\ref{item:inner_vertex_blocks}).  It is clear that every
inner vertex has at most one outgoing edge of every color.  That the
incoming edges lie in the right interval, follows from the fact that a
corner-side contact between two homothetic regular pentagons is only
possible if a corner of the first pentagon touches the opposite side
of the second pentagon.

Finally we check property~(F\ref{item:no_three_empty}) for the
case~$i=1$.  The other cases are symmetric.  Let~$v$ be an inner
vertex of~$G$ and~$A$ the corresponding pentagon of~$\mathcal{S}$.
Since~$G$ is a triangulation, the pentagons corresponding to any two
consecutive neighbors of~$v$ have to touch. Therefore at least one of these pentagons,
we call it~$B$, has to intersect the area below the horizontal side
of~$A$, and that is only possible if the contact of~$A$ and~$B$
corresponds to an incoming edge of~$v$ of color~$1$ or an outgoing
edge of color~$3$ or~$4$. All other possibilities can be excluded in the following way:
If the contact of~$A$ and~$B$ corresponds to an incoming edge of color~$2$,
then the entire pentagon~$B$ lies on the left of the contact point of~$A$ and~$B$ and
therefore also left of the horizontal side of~$A$. If the contact corresponds to an
outgoing edge of color~$5$, each point of~$B$ lies above the contact point or
on the left of the contact point and therefore above the horizontal side of~$A$
or on its left. The other cases can be excluded with similar arguments.
\end{proof}

Schnyder \cite{schnyder1990embedding} introduced a similar structure
for inner triangulations of a triangle:

\begin{definition}
  A \emph{Schnyder wood} of an inner triangulation $T$ of the
  triangle~$b_1,b_2,b_3$ is an orientation and coloring of the inner
  edges of $T$ in the colors $1,2,3$ with the following properties:
\begin{enumerate}[(S1)]
\item All edges incident to $b_i$ are oriented towards $b_i$ and
  colored in the color $i$.
\item Each inner vertex has in clockwise order exactly one outgoing
  edge of color~$1$, one outgoing edge of color~$2$ and one outgoing
  edge of color~$3$, and in the interval between two outgoing edges
  there are only incoming edges in the third color.
\end{enumerate}
\end{definition}

Schnyder proved that every inner triangulation of a triangle admits a
Schnyder wood, and using this result, we will show that every inner
triangulation of a pentagon admits a five color forest.

\begin{theorem}[\cite{schnyder1990embedding}] \label{thm:schnyder_wood_existence}
Let~$T$ be an inner triangulation of a triangle.
Then there exists a Schnyder wood of~$T$.
\end{theorem}

\begin{theorem} \label{thm:fcf_existence}
Let~$G$ be an inner triangulation of the pentagon~$a_1,\dotsc,a_5$.
Then there exists a five color forest of~$G$.
\end{theorem}

\begin{figure}

\centering

\tikzstyle{vertex}=[circle,fill,scale=0.3]
\tikzstyle{out edge}=[-latex',thick]
\tikzstyle{in edge}=[latex'-]
\tikzstyle{edge}=[]
\tikzstyle{contract}=[color=red]

\begin{tikzpicture}[scale=1.75]

\node[vertex] (a1) at (90:1.2) [label={above}: \footnotesize $a_1$] {};
\node[vertex] (a2) at (18:1.2) [label={right}: \footnotesize $a_2$] {};
\node[vertex] (a3) at (306:1.2) [label={below right}: \footnotesize $a_3$] {};
\node[vertex] (a4) at (234:1.2) [label={below left}: \footnotesize $a_4$] {};
\node[vertex] (a5) at (162:1.2) [label={left}: \footnotesize $a_5$] {};

\node[vertex] (c5) at ($($(a2)!.5!(a3)$)+(162:.8)$) [label={below,xshift=-.1cm}: \footnotesize $c_5$] {};
\node[vertex] (c2) at ($($(a4)!.5!(a5)$)+(18:.8)$) [label={below,xshift=.1cm}: \footnotesize $c_2$] {};

\draw[edge] (a1) -- (a2) -- (a3) -- (a4) -- (a5) -- (a1);
\draw[edge] (a2) -- (c5) -- (a3);
\draw[edge] (a4) -- (c2) -- (a5);

\draw[rounded corners=5pt,contract] ($(a4)+(-117:.15)$) -- ($(a4)+(-27:.15)$) -- ($(a5)+(63:.15)$) -- ($(a5)+(153:.15)$) -- cycle;
\draw[rounded corners=5pt,contract] ($(a2)+(117:.15)$) -- ($(a2)+(27:.15)$) -- ($(a3)+(-63:.15)$) -- ($(a3)+(-153:.15)$) -- cycle;
\draw[contract] (a1) circle (.09);

\node at (a1) [label={right,contract}: \footnotesize $b_1$] {};
\node at ($(a2)!.5!(a3)$) [label={right,contract}: \footnotesize $b_3$] {};
\node at ($(a4)!.5!(a5)$) [label={left,contract}: \footnotesize $b_4$] {};

\end{tikzpicture}
\begin{tikzpicture}[scale=1.75]
\node[vertex] (b1) at (90:1.2) [label={above}: \footnotesize $b_1$] {};
\node[vertex] (b3) at (-30:1.2) [label={below right}: \footnotesize $b_3$] {};
\node[vertex] (b4) at (-150:1.2) [label={below left}: \footnotesize $b_4$] {};

\node[vertex] (c5) at ($(b3)+(150:.7)$) [label={left}: \footnotesize $c_5$] {};
\node[vertex] (c2) at ($(b4)+(30:.7)$) [label={right}: \footnotesize $c_2$] {};

\draw[edge] (b1) -- (b3) -- (b4) -- (b1);

\draw[color1,in edge] (b1) -- ($(b1)+(260:.4)$);
\draw[color1,in edge] (b1) -- ($(b1)+(280:.4)$);
\draw[color3,in edge] (b3) -- ($(b3)+(135:.4)$);
\draw[color3,in edge] (b3) -- (c5);
\draw[color3,in edge] (b3) -- ($(b3)+(165:.4)$);
\draw[color4,in edge] (b4) -- ($(b4)+(15:.4)$);
\draw[color4,in edge] (b4) -- (c2);
\draw[color4,in edge] (b4) -- ($(b4)+(45:.4)$);

\node[vertex] (v) at (90:.3) {};
\draw[color1,out edge] (v) -- +(90:.3);
\draw[color3,out edge] (v) -- +(-30:.3);
\draw[color4,out edge] (v) -- +(-150:.3);
\draw[color1,in edge] (v) -- +(-70:.2);
\draw[color1,in edge] (v) -- +(-110:.2);
\draw[color3,in edge] (v) -- +(-190:.2);
\draw[color3,in edge] (v) -- +(-230:.2);
\draw[color4,in edge] (v) -- +(50:.2);
\draw[color4,in edge] (v) -- +(10:.2);
\end{tikzpicture}
\begin{tikzpicture}[scale=1.75]

\node[vertex] (a1) at (90:1.2) [label={above}: \footnotesize $a_1$] {};
\node[vertex] (a2) at (18:1.2) [label={right}: \footnotesize $a_2$] {};
\node[vertex] (a3) at (306:1.2) [label={below right}: \footnotesize $a_3$] {};
\node[vertex] (a4) at (234:1.2) [label={below left}: \footnotesize $a_4$] {};
\node[vertex] (a5) at (162:1.2) [label={left}: \footnotesize $a_5$] {};

\node[vertex] (c5) at ($($(a2)!.5!(a3)$)+(162:.8)$) [label={below,xshift=-.1cm}: \footnotesize $c_5$] {};
\node[vertex] (c2) at ($($(a4)!.5!(a5)$)+(18:.8)$) [label={below,xshift=.1cm}: \footnotesize $c_2$] {};

\draw[edge] (a1) -- (a2) -- (a3) -- (a4) -- (a5) -- (a1);

\draw[color2,in edge] (a2) -- (c5);
\draw[color3,in edge] (a3) -- (c5);
\draw[color4,in edge] (a4) -- (c2);
\draw[color5,in edge] (a5) -- (c2);

\draw[color1,in edge] (a1) -- +(-110:.4);
\draw[color1,in edge] (a1) -- +(-70:.4);

\draw[color2,in edge] (a2) -- +(170:.4);
\draw[color5,in edge] (a5) -- +(10:.4);

\draw[color3,in edge] (a3) -- +(150:.4);
\draw[color4,in edge] (a4) -- +(30:.4);

\draw[color2,in edge] (a2) -- +(230:.4);
\draw[color5,in edge] (a5) -- +(310:.4);

\draw[color3,in edge] (a3) -- +(100:.4);
\draw[color4,in edge] (a4) -- +(80:.4);

\draw[color5,in edge] (c5) -- +(-5:.3);
\draw[color5,in edge] (c5) -- +(-25:.3);

\draw[color2,in edge] (c2) -- +(-175:.3);
\draw[color2,in edge] (c2) -- +(-155:.3);

\node[vertex] (v) at (90:.4) {};
\draw[color1,out edge] (v) -- +(90:.3);
\draw[color3,out edge] (v) -- +(-30:.3);
\draw[color4,out edge] (v) -- +(-150:.3);
\draw[color1,in edge] (v) -- +(-70:.2);
\draw[color1,in edge] (v) -- +(-110:.2);
\draw[color3,in edge] (v) -- +(-190:.2);
\draw[color3,in edge] (v) -- +(-230:.2);
\draw[color4,in edge] (v) -- +(50:.2);
\draw[color4,in edge] (v) -- +(10:.2);

\node[vertex] (vr) at ($(c5)+(-18:.5)$) {};
\draw[color2,out edge] (vr) -- +(90:.25);
\draw[color3,out edge] (vr) -- +(-30:.25);
\draw[color5,out edge] (vr) -- +(-150:.25);
\draw[color2,in edge] (vr) -- +(-70:.18);
\draw[color2,in edge] (vr) -- +(-110:.18);
\draw[color3,in edge] (vr) -- +(-190:.18);
\draw[color3,in edge] (vr) -- +(-230:.18);
\draw[color5,in edge] (vr) -- +(50:.18);
\draw[color5,in edge] (vr) -- +(10:.18);

\node[vertex] (vl) at ($(c2)+(-162:.5)$) {};
\draw[color5,out edge] (vl) -- +(90:.25);
\draw[color2,out edge] (vl) -- +(-30:.25);
\draw[color4,out edge] (vl) -- +(-150:.25);
\draw[color5,in edge] (vl) -- +(-70:.18);
\draw[color5,in edge] (vl) -- +(-110:.18);
\draw[color2,in edge] (vl) -- +(-190:.18);
\draw[color2,in edge] (vl) -- +(-230:.18);
\draw[color4,in edge] (vl) -- +(50:.18);
\draw[color4,in edge] (vl) -- +(10:.18);

\end{tikzpicture}

\caption{A construction of a five color forest of a given graph using Schnyder woods.}

\label{fig:fcf_from_sw}

\end{figure}
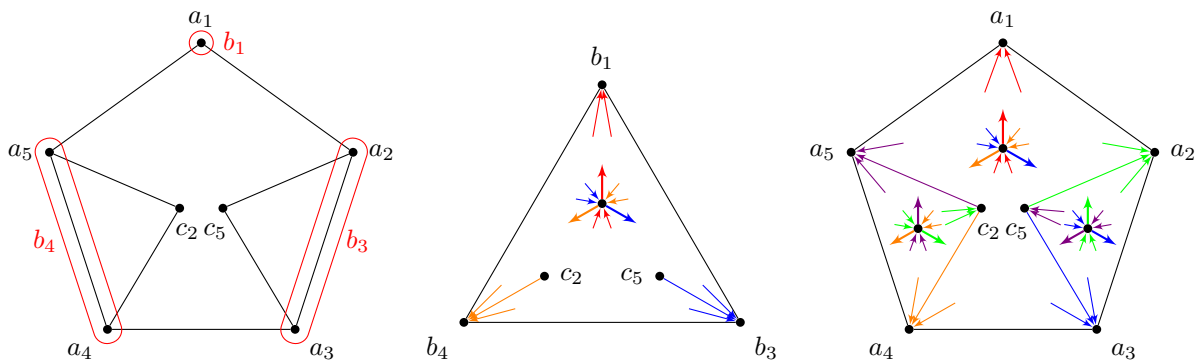

\begin{proof}
  The following construction is illustrated by Fig.~\ref{fig:fcf_from_sw}.
  Contract the edge~$a_2 a_3$ to a vertex~$b_3$ and the edge~$a_4 a_5$
  to a vertex~$b_4$.  In these contraction steps the maximal
  triangles~$c_5a_2a_3$ and~$c_2a_4a_5$ incident to~$a_2,a_3$ and
  to~$a_4,a_5$ are contracted to a single edge, in particular vertices
  inside these triangles are removed.  Further let~$b_1:=a_1$.
  This results in an inner triangulation~$T$ of the triangle~$b_1 ,
  b_3 , b_4$.  Due to Theorem~\ref{thm:schnyder_wood_existence} there exists
  a Schnyder Wood~$S$ of~$T$ (we use the colors~$1,3,4$ instead
  of~$1,2,3$).

  Take the colors and orientations of all inner edges not
  inside of~$c_5a_2a_3$ or~$c_2a_4a_5$ and not incident to~$a_2$ or~$a_5$
  from~$T$ to~$G$.  Now color all inner edges incident to~$a_2$
  and not inside~$c_5a_2a_3$ in color~$2$ and orient them
  towards~$a_2$, and color all inner edges incident to~$a_5$ and not
  inside~$c_2a_4a_5$ in color~$5$ and orient them towards~$a_5$.
  For the edges inside~$c_5a_2a_3$ construct another Schnyder
  wood where~$c_5$ has incoming edges in color~$5$,~$a_3$ has incoming
  edges in color~$3$, and~$a_2$ has incoming edges in color~$2$.
  Analogously, construct a Schnyder wood on the edges
  inside~$c_2a_4a_5$ in the colors~$2,4,5$.

  It can easily be verified that this coloring and orientation of the
  inner edges of~$G$ fulfills the properties~(F\ref{item:outer_edges}) and~(F\ref{item:inner_vertex_blocks}) of a five color forest.
  To see that property~(F\ref{item:no_three_empty}) is also fulfilled
  we distinguish several cases. If a vertex is not inside~$c_5a_2a_3$
  or~$c_2a_4a_5$ and not adjacent to~$a_2$ or~$a_5$, it has outgoing
  edges in colors~$1,3,4$. If a vertex is not inside~$c_5a_2a_3$
  or~$c_2a_4a_5$ and either adjacent to~$a_2$ or to~$a_5$, it has outgoing
  edges in colors~$1,2,4$ or~$1,3,5$, respectively. If a vertex is inside~$c_5a_2a_3$
  or~$c_2a_4a_5$, it has outgoing edges in colors~$2,3,5$ or~$2,4,5$, respectively.
  Therefore in all of these cases property~(F\ref{item:no_three_empty})
  is fulfilled. The only remaining case is that a vertex~$v$ is adjacent
  to~$a_2$ and~$a_5$. If~$v=c_2=c_5$, then~$v$ has outgoing edges
  in all five colors and fulfills property~(F\ref{item:no_three_empty}).
  Otherwise~$c_2$ and~$c_5$ lie inside the $5$-gon~$a_2 a_3 a_4 a_5 v$.
  Thus this $5$-gon is not empty and~$v$ has a neighbor~$w$ inside this~$5$-gon
  (since~$G$ has no chords,~$a_2 a_5 v$ cannot be a face of~$G$).
  Note that the edge between~$w$ and~$v$ is oriented from~$w$ to~$v$
  and has color~$1$. Therefore, ~$v$ has outgoing edges in colors~$1,2,5$ and
  at least one incoming edge in color~$1$. Hence,~$v$ fulfills property~(F\ref{item:no_three_empty}).
\end{proof}

\subsection{\texorpdfstring{$\alpha$}{alpha}-orientations}
Our goal is to connect the setting of five color forests with the well
studied orientations of planar graphs
with prescribed outdegrees.

\begin{definition}
  Let $H$ be an undirected graph and $\alpha:V(H)\rightarrow
  \mathbb{N}$.  Then an orientation $H'$ of $H$ is called an
  \emph{$\alpha$-orientation} if $\operatorname{outdeg}(v)=\alpha(v)$
  for all vertices $v\in V(H')$.
\end{definition}

In a five color forest every inner vertex has outdegree at most $5$.
The following lemma allows us to add vertices and edges so that the outdegree of
every inner vertex becomes exactly $5$.
The statement of the lemma corresponds to the geometric fact 
that in a regular pentagon contact representation of a triangulation $G$
the area between three pentagons corresponding to a face of~$G$
is a quadrilateral with exactly one concave corner.

\begin{lemma} \label{lem:missing_edge} 
  Let $G$ be endowed with a five color forest and let $f$ be a face of
  $G$ that is incident to at most one outer vertex.  Then in exactly one of the
  three inner angles of $f$ an outgoing edge is missing in the
  cyclic order of the respective vertex.
\end{lemma}

\begin{proof}
  Since the facial cycle of~$f$ has length three, it contains an
  oriented path of length two.  For symmetry reasons we can assume
  that this path is oriented clockwise and the first edge of this path
  has color~$1$.  Then because of
  property~(F\ref{item:no_three_empty}) the second edge of this path
  can only have the colors~$2$ and~$3$.  Figure~\ref{fig:face_cases} shows
  all possible cases for the orientation and coloring of the third
  edge of the cycle and verifies the statement for all these cases.
\end{proof}

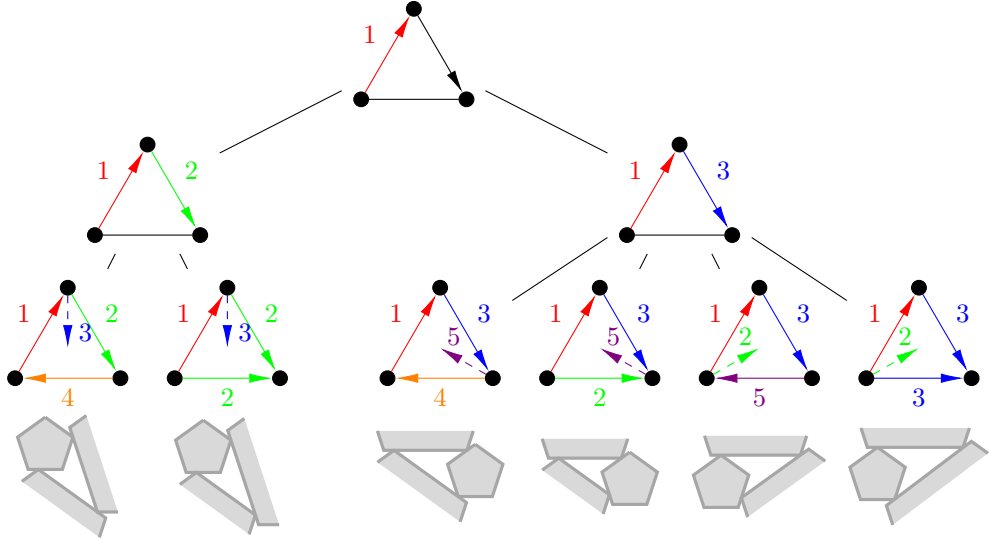
\begin{figure}
\centering
\begin{tikzpicture}[
  level 1/.style={sibling distance=7cm,level distance=1.8cm},
  level 2/.style={sibling distance=2.1cm,level distance=2.1cm},
  level 3/.style={level distance=1.7cm,edge from parent/.style={}}]
\tikzstyle{corner}=[circle,draw,fill,scale=0.5];
\tikzstyle{edge}=[];
\tikzstyle{arrow}=[-{Latex[scale length=1.8,scale width=1]},edge];
\tikzstyle{extra edge color}=[color=red];
\tikzstyle{extra edge}=[arrow,dashed];

\tikzstyle{triangle scale}=[scale=0.8];
\tikzstyle{visualization scale}=[scale=0.45];

\tikzstyle{pentagon}=[very thick,draw=gray!70,fill=gray!30];

\node {
  \begin{tikzpicture}[triangle scale]
    \node[corner] (l) at (210:1) {};
    \node[corner] (r) at (330:1) {};
    \node[corner] (o) at (90:1) {};
    \draw[arrow,color1] (l) -- (o) node[pos=0.5,above left] {\footnotesize $1$};
    \draw[arrow] (o) -- (r);
    \draw[edge] (l) -- (r);
  \end{tikzpicture}
 }
 child {
   node {
     \begin{tikzpicture}[triangle scale]
       \node[corner] (l) at (210:1) {};
       \node[corner] (r) at (330:1) {};
       \node[corner] (o) at (90:1) {};
       \draw[arrow,color1] (l) -- (o) node[pos=0.5,above left] {\footnotesize $1$};
       \draw[arrow,color2] (o) -- (r) node[pos=0.5,above right] {\footnotesize $2$};
       \draw[edge] (l) -- (r);
     \end{tikzpicture}
   }
   child {
     node {
       \begin{tikzpicture}[triangle scale]
         \node[corner] (l) at (210:1) {};
         \node[corner] (r) at (330:1) {};
         \node[corner] (o) at (90:1) {};
         \draw[arrow,color1] (l) -- (o) node[pos=0.5,above left] {\footnotesize $1$};
         \draw[arrow,color2] (o) -- (r) node[pos=0.5,above right] {\footnotesize $2$};
         \draw[arrow,color4] (r) -- (l) node[pos=0.5,below] {\footnotesize $4$};
         \draw[extra edge,color3] (o) -- (0,0) node[pos=0.7,right] {\footnotesize $3$};
       \end{tikzpicture}
     }
     child {
       node {
         \begin{tikzpicture}[visualization scale]
           \coordinate (A1) at (0,0);
           \coordinate (A2) at (324:1);
           \coordinate (A3) at ($(A2)+(252:1)$);
           \coordinate (A4) at ($(A3)+(180:1)$);
           \coordinate (A5) at ($(A4)+(108:1)$);
           \draw[pentagon] (A1) -- (A2) -- (A3) -- (A4) -- (A5) -- (A1);
           \coordinate (B1) at ($(A4)!0.3!(A3)$);
           \coordinate (Bh) at ($(B1)+(324:4)$);
           \coordinate (Ch) at ($(A2)+(288:4)$);
           \coordinate (C4) at (intersection of A2--Ch and B1--Bh);
           \coordinate (B2) at ($(C4)+(324:0.3)$);
           \coordinate (B3) at ($(B2)+(252:0.6)$);
           \coordinate (B5) at ($(B1)+(216:0.6)$);
           \coordinate (C5) at ($(A2)+(108:0.3)$);
           \coordinate (C3) at ($(C4)+(0:0.6)$);
           \coordinate (C1) at ($(C5)+(36:0.6)$);
           \draw[pentagon] (B5) -- (B1) -- (B2) -- (B3);
           \draw[pentagon] (C3) -- (C4) -- (C5) -- (C1);
         \end{tikzpicture}
       }
     }
   }
   child {
     node {
       \begin{tikzpicture}[triangle scale]
         \node[corner] (l) at (210:1) {};
         \node[corner] (r) at (330:1) {};
         \node[corner] (o) at (90:1) {};
         \draw[arrow,color1] (l) -- (o) node[pos=0.5,above left] {\footnotesize $1$};
         \draw[arrow,color2] (o) -- (r) node[pos=0.5,above right] {\footnotesize $2$};
         \draw[arrow,color2] (l) -- (r) node[pos=0.5,below] {\footnotesize $2$};
         \draw[extra edge,color3] (o) -- (0,0) node[pos=0.7,right] {\footnotesize $3$};
       \end{tikzpicture}
     }
     child {
       node {
         \begin{tikzpicture}[visualization scale]
           \coordinate (A1) at (0,0);
           \coordinate (A2) at (324:1);
           \coordinate (A3) at ($(A2)+(252:1)$);
           \coordinate (A4) at ($(A3)+(180:1)$);
           \coordinate (A5) at ($(A4)+(108:1)$);
           \draw[pentagon] (A1) -- (A2) -- (A3) -- (A4) -- (A5) -- (A1);
           \coordinate (B1) at ($(A4)!0.3!(A3)$);
           \coordinate (Bh) at ($(B1)+(324:4)$);
           \coordinate (Ch) at ($(A2)+(288:4)$);
           \coordinate (B2) at (intersection of A2--Ch and B1--Bh);
           \coordinate (C4) at ($(B2)+(288:0.3)$);
           \coordinate (B3) at ($(B2)+(252:0.6)$);
           \coordinate (B5) at ($(B1)+(216:0.6)$);
           \coordinate (C5) at ($(A2)+(108:0.3)$);
           \coordinate (C3) at ($(C4)+(0:0.6)$);
           \coordinate (C1) at ($(C5)+(36:0.6)$);
           \draw[pentagon] (B5) -- (B1) -- (B2) -- (B3);
           \draw[pentagon] (C3) -- (C4) -- (C5) -- (C1);
         \end{tikzpicture}
       }
     }
   }
 }
 child {
   node {
     \begin{tikzpicture}[triangle scale]
       \node[corner] (l) at (210:1) {};
       \node[corner] (r) at (330:1) {};
       \node[corner] (o) at (90:1) {};
       \draw[arrow,color1] (l) -- (o) node[pos=0.5,above left] {\footnotesize $1$};
       \draw[arrow,color3] (o) -- (r) node[pos=0.5,above right] {\footnotesize $3$};
       \draw[edge] (l) -- (r);
     \end{tikzpicture}
   }
   child {
     node {
       \begin{tikzpicture}[triangle scale]
         \node[corner] (l) at (210:1) {};
         \node[corner] (r) at (330:1) {};
         \node[corner] (o) at (90:1) {};
         \draw[arrow,color1] (l) -- (o) node[pos=0.5,above left] {\footnotesize $1$};
         \draw[arrow,color3] (o) -- (r) node[pos=0.5,above right] {\footnotesize $3$};
         \draw[arrow,color4] (r) -- (l) node[pos=0.5,below] {\footnotesize $4$};
         \draw[extra edge,color5] (r) -- (0,0) node[pos=0.7,above] {\footnotesize $5$};
       \end{tikzpicture}
     }
     child {
       node {
         \begin{tikzpicture}[visualization scale]
           \coordinate (A1) at (0,0);
           \coordinate (A2) at (324:1);
           \coordinate (A3) at ($(A2)+(252:1)$);
           \coordinate (A4) at ($(A3)+(180:1)$);
           \coordinate (A5) at ($(A4)+(108:1)$);
           \draw[pentagon] (A1) -- (A2) -- (A3) -- (A4) -- (A5) -- (A1);
           \coordinate (B3) at ($(A1)!0.3!(A5)$);
           \coordinate (Bh) at ($(B3)+(180:4)$);
           \coordinate (Ch) at ($(A4)+(144:4)$);
           \coordinate (C1) at (intersection of A4--Ch and B3--Bh);
           \coordinate (B4) at ($(C1)+(180:0.3)$);
           \coordinate (B2) at ($(B3)+(72:0.6)$);
           \coordinate (B5) at ($(B4)+(108:0.6)$);
           \coordinate (C2) at ($(A4)+(324:0.3)$);
           \coordinate (C3) at ($(C2)+(252:0.6)$);
           \coordinate (C5) at ($(C1)+(216:0.6)$);
           \draw[pentagon] (B2) -- (B3) -- (B4) -- (B5);
           \draw[pentagon] (C5) -- (C1) -- (C2) -- (C3);
         \end{tikzpicture}
       }
     }
   }
   child {
     node {
       \begin{tikzpicture}[triangle scale]
         \node[corner] (l) at (210:1) {};
         \node[corner] (r) at (330:1) {};
         \node[corner] (o) at (90:1) {};
         \draw[arrow,color1] (l) -- (o) node[pos=0.5,above left] {\footnotesize $1$};
         \draw[arrow,color3] (o) -- (r) node[pos=0.5,above right] {\footnotesize $3$};
         \draw[arrow,color2] (l) -- (r) node[pos=0.5,below] {\footnotesize $2$};
         \draw[extra edge,color5] (r) -- (0,0) node[pos=0.7,above] {\footnotesize $5$};
       \end{tikzpicture}
     }
     child {
       node {
         \begin{tikzpicture}[visualization scale]
           \coordinate (A1) at (0,0);
           \coordinate (A2) at (324:1);
           \coordinate (A3) at ($(A2)+(252:1)$);
           \coordinate (A4) at ($(A3)+(180:1)$);
           \coordinate (A5) at ($(A4)+(108:1)$);
           \draw[pentagon] (A1) -- (A2) -- (A3) -- (A4) -- (A5) -- (A1);
           \coordinate (B3) at ($(A1)!0.3!(A5)$);
           \coordinate (Bh) at ($(B3)+(180:4)$);
           \coordinate (C2) at ($(A4)!0.3!(A5)$);
           \coordinate (Ch) at ($(C2)+(144:4)$);
           \coordinate (C1) at (intersection of C2--Ch and B3--Bh);
           \coordinate (B4) at ($(C1)+(180:0.3)$);
           \coordinate (B2) at ($(B3)+(72:0.6)$);
           \coordinate (B5) at ($(B4)+(108:0.6)$);
           \coordinate (C3) at ($(C2)+(252:0.6)$);
           \coordinate (C5) at ($(C1)+(216:0.6)$);
           \draw[pentagon] (B2) -- (B3) -- (B4) -- (B5);
           \draw[pentagon] (C5) -- (C1) -- (C2) -- (C3);
         \end{tikzpicture}
       }
     }
   }
   child {
     node {
       \begin{tikzpicture}[triangle scale]
         \node[corner] (l) at (210:1) {};
         \node[corner] (r) at (330:1) {};
         \node[corner] (o) at (90:1) {};
         \draw[arrow,color1] (l) -- (o) node[pos=0.5,above left] {\footnotesize $1$};
         \draw[arrow,color3] (o) -- (r) node[pos=0.5,above right] {\footnotesize $3$};
         \draw[arrow,color5] (r) -- (l) node[pos=0.5,below] {\footnotesize $5$};
         \draw[extra edge,color2] (l) -- (0,0) node[pos=0.7,above] {\footnotesize $2$};
       \end{tikzpicture}
     }
     child {
       node {
         \begin{tikzpicture}[visualization scale]
           \coordinate (A1) at (0,0);
           \coordinate (A2) at (324:1);
           \coordinate (A3) at ($(A2)+(252:1)$);
           \coordinate (A4) at ($(A3)+(180:1)$);
           \coordinate (A5) at ($(A4)+(108:1)$);
           \draw[pentagon] (A1) -- (A2) -- (A3) -- (A4) -- (A5) -- (A1);
           \coordinate (Bh) at ($(A1)+(0:4)$);
           \coordinate (C5) at ($(A3)!0.3!(A2)$);
           \coordinate (Ch) at ($(C5)+(36:4)$);
           \coordinate (B3) at (intersection of A1--Bh and C5--Ch);
           \coordinate (B4) at ($(A1)+(180:0.3)$);
           \coordinate (B2) at ($(B3)+(72:0.6)$);
           \coordinate (B5) at ($(B4)+(108:0.6)$);
           \coordinate (C1) at ($(B3)+(36:0.3)$);
           \coordinate (C2) at ($(C1)+(324:0.6)$);
           \coordinate (C4) at ($(C5)+(288:0.6)$);
           \draw[pentagon] (B2) -- (B3) -- (B4) -- (B5);
           \draw[pentagon] (C4) -- (C5) -- (C1) -- (C2);
         \end{tikzpicture}
       }
     }
   }
   child {
     node {
       \begin{tikzpicture}[triangle scale]
         \node[corner] (l) at (210:1) {};
         \node[corner] (r) at (330:1) {};
         \node[corner] (o) at (90:1) {};
         \draw[arrow,color1] (l) -- (o) node[pos=0.5,above left] {\footnotesize $1$};
         \draw[arrow,color3] (o) -- (r) node[pos=0.5,above right] {\footnotesize $3$};
         \draw[arrow,color3] (l) -- (r) node[pos=0.5,below] {\footnotesize $3$};
         \draw[extra edge,color2] (l) -- (0,0) node[pos=0.7,above] {\footnotesize $2$};
       \end{tikzpicture}
     }
     child {
       node {
         \begin{tikzpicture}[visualization scale]
           \coordinate (A1) at (0,0);
           \coordinate (A2) at (324:1);
           \coordinate (A3) at ($(A2)+(252:1)$);
           \coordinate (A4) at ($(A3)+(180:1)$);
           \coordinate (A5) at ($(A4)+(108:1)$);
           \draw[pentagon] (A1) -- (A2) -- (A3) -- (A4) -- (A5) -- (A1);
           \coordinate (Bh) at ($(A1)+(0:4)$);
           \coordinate (Ch) at ($(A3)+(36:4)$);
           \coordinate (B3) at (intersection of A1--Bh and C5--Ch);
           \coordinate (B4) at ($(A1)+(180:0.3)$);
           \coordinate (B2) at ($(B3)+(72:0.6)$);
           \coordinate (B5) at ($(B4)+(108:0.6)$);
           \coordinate (C5) at ($(A3)+(216:0.3)$);
           \coordinate (C1) at ($(B3)+(36:0.3)$);
           \coordinate (C2) at ($(C1)+(324:0.6)$);
           \coordinate (C4) at ($(C5)+(288:0.6)$);
           \draw[pentagon] (B2) -- (B3) -- (B4) -- (B5);
           \draw[pentagon] (C4) -- (C5) -- (C1) -- (C2);
         \end{tikzpicture}
       }
     }
   }
 }
;
\end{tikzpicture}
\caption{Full case distinction for Lemma~\ref{lem:missing_edge}}
\label{fig:face_cases}
\end{figure}

Now we define an extension of $G$ and a function $\alpha_5$
such that every five color forest of $G$ can be extended to an
$\alpha_5$-orientation of this extension.

\begin{definition}
  The \emph{stack extension} $\sG$ of $G$ is the extension of $G$ that
  contains an extra vertex in every inner face that is incident to at
  most one of the outer vertices. These new vertices are connected to
  all three vertices of the respective face (see Fig.~\ref{fig:induced_fcf_ao} (right)).  We call the new vertices
  \emph{stack vertices} and the vertices of $G$ \emph{normal
    vertices}.
\end{definition}

\begin{definition}
An orientation of the inner edges of $\sG$ is a 
$\alpha_5$-orientation if the outdegrees of the vertices 
correspond to the following values:
\[
\alpha_5(v)= \begin{cases}
  2 & \text{if $v$ is a stack vertex,} \\
  5 & \text{if $v$ is an inner normal vertex,} \\
  0 & \text{if $v$ is an outer normal vertex.}
\end{cases}
\]
\end{definition}

A five color forest of $G$ induces an $\alpha_5$-orientation of $\sG$
in a canonical way by keeping the orientation of the edges of $G$ and
defining the missing edge of Lemma~\ref{lem:missing_edge} as the unique
incoming edge for every stack vertex.

\begin{observation} \label{obs:extend_coloring} The coloring of the inner
  edges of a five color forest can be extended to a coloring of the
  inner edges of the induced~$\alpha_5$-orientation that fulfills the
  properties of a five color forest at all normal vertices.
\end{observation}

\subsection{Bijection between five color forests and \texorpdfstring{$\alpha_5$}{alpha5}-orientations}

Now we want to prove that the canonical mapping from five color
forests to $\alpha_5$-orientations is a bijection.  For this purpose we
need to reconstruct the colors of the inner edges of~$G$ if we are
given an $\alpha_5$-orientation.  The idea of this construction will be
to start with an inner edge~$e$ of~$G$ and follow a properly defined
path until it reaches one of the five outer vertices.  Then
the color of this outer vertex will be the color of~$e$.  This
approach is similar to the proof of the bijection of Schnyder Woods
and $3$-orientations in \cite{de2001topological}.

\begin{lemma}\label{lem:number_edges_into_cycle}
  Let~$C$ be a simple cycle of length~$\ell$ in~$\sG$ and let all
  vertices of~$C$ be normal vertices, i.e., vertices of~$G$.  Then
  there are exactly~$2 \ell - 5$ edges pointing from~$C$ into the
  interior of~$C$.
\end{lemma}

\begin{proof}
  First we view~$C$ as a cycle in~$G$.  Let~$k$ be the number of
  vertices strictly inside~$C$.  Since~$G$ is an inner triangulation there are
  exactly~$2k+\ell-2$ faces and~$3k+\ell-3$ edges strictly inside~$C$ by Euler's formula.

  Now we view~$C$ as a cycle in~$\sG$.  In addition to the $3k+\ell-3$
  normal edges there are $3$ stack edges in each face, hence, the
  number of edges in $C$ is $9k+4\ell -9$. At each stack vertex we see
  2 starting edges and at every normal vertex 5. Therefore there are
  $2 (2k+\ell-2) + 5k = 9k +2\ell -4$ edges starting at a vertex
  inside~$C$. Taking the difference we find that there are $2\ell -5$ edges
  pointing from a vertex of~$C$ into the interior.
\end{proof}

Next we will show some properties of oriented cycles in five color
forests.  By~$T_i$ we denote the forest consisting of all edges of
color~$i$, and by~$T_i^{-1}$ we denote the forest~$T_i$ with all edges
reversed.

\begin{lemma} \label{lem:acyclic_orientation}
The orientation 
$T:=T_i + T_{i-1} + T_{i+1} + T_{i-2}^{-1} + T_{i+2}^{-1}$ 
of $G$ is acyclic.
\end{lemma}

\begin{proof}
  Assume there is a simple oriented cycle~$C$ of length~$\ell$ in~$T$.
  Because of the symmetry of the colors in the definition of a five
  color forest, it suffices to consider the case that~$C$ is oriented
  clockwise and~$i=1$.

\begin{figure}

\centering

\tikzstyle{vertex}=[circle,fill,scale=0.5]
\tikzstyle{out edge}=[-latex',very thick]
\tikzstyle{in edge}=[latex'-]
\tikzstyle{pi value}=[color=red]

\begin{tikzpicture}[scale=2]

\node[vertex] (v) at (0,0) {};

\node[] (v1) at (90:1.2) {$1$};
\node[] (v2) at (18:1.2) {$2$};
\node[] (v3) at (306:1.2) {$3$};
\node[] (v4) at (234:1.2) {$4$};
\node[] (v5) at (162:1.2) {$5$};

\node[] at (270:.9) {\footnotesize $1$};
\node[] at (198:.9) {\footnotesize $2$};
\node[] at (126:.9) {\footnotesize $3$};
\node[] at (54:.9) {\footnotesize $4$};
\node[] at (342:.9) {\footnotesize $5$};

\node[pi value] at (90:1.5) {$+1$};
\node[pi value] at (18:1.5) {$0$};
\node[pi value] at (306:1.5) {$0$};
\node[pi value] at (234:1.5) {$+1$};
\node[pi value] at (162:1.5) {$+2$};

\node[pi value] at (270:1.2) {$+1$};
\node[pi value] at (198:1.2) {$+2$};
\node[pi value] at (126:1.2) {$+2$};
\node[pi value] at (54:1.2) {$+1$};
\node[pi value] at (342:1.2) {$0$};

\draw[out edge] (v) -- (v1);
\draw[out edge] (v) -- (v2);
\draw[out edge] (v3) -- (v);
\draw[out edge] (v4) -- (v);
\draw[out edge] (v) -- (v5);

\draw[in edge] (114:.7) -- (v);
\draw[in edge] (138:.7) -- (v);

\draw[in edge] (v) -- (186:.7);
\draw[in edge] (v) -- (210:.7);

\draw[in edge] (v) -- (258:.7);
\draw[in edge] (v) -- (282:.7);

\draw[in edge] (v) -- (330:.7);
\draw[in edge] (v) -- (354:.7);

\draw[in edge] (42:.7) -- (v);
\draw[in edge] (66:.7) -- (v);

\draw[dotted] (180:1.5) -- (0:1.5);

\end{tikzpicture}

\caption{The possible incident edges of a vertex~$v$ in~$T$.
For an incoming edge of color~$c$ below the dotted line the red value is~$2-\pi(c)$ and counts the number of thick edges in the counterclockwise angle from this edge to the dotted line.
For an outgoing edge of color~$c'$ above the dotted line the red value is~$\pi(c')$ and counts the number of thick edges in the clockwise angle from this edge to the dotted line.
The thick edges are exactly those that are outgoing in~$v$ in the $\alpha$-orientation.}

\label{fig:reorientation_acyclic}

\end{figure}
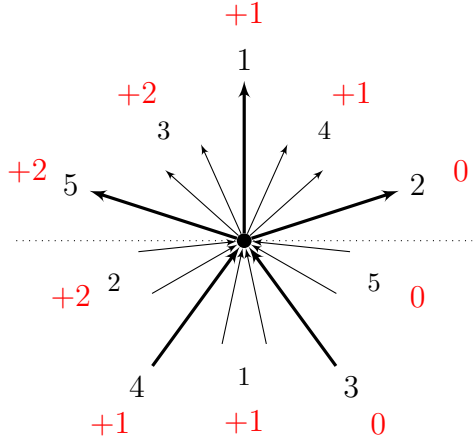

We define the following auxiliary function for the five colors:
\[ \pi(1) := 1 \enspace , \enspace \pi(2) := 0 \enspace , \enspace
\pi(3) :=2 \enspace , \enspace \pi(4) :=1 \enspace , \enspace \pi(5)
:=2 \enspace . \]

Let~$e$ be an edge of color~$c$ ending at vertex~$v$ and~$e'$ an edge
of color~$c'$ starting at~$v$ in the orientation~$T$.  From
Fig.~\ref{fig:reorientation_acyclic} we can read off the following:
In the counterclockwise angle of~$v$ between~$e$ and the dotted
line there are $2 - \pi(c)$ edges which are outgoing in the
$\alpha_5$-orientation of~$\sG$. In the clockwise angle of~$v$
between~$e'$ and the dotted line there are~$\pi(c')$ outgoing
edges. Hence, there are exactly~${\pi(c') + (2 - \pi(c))}$ edges
pointing away from~$v$ in the counterclockwise angle between~$e$ and~$e'$ (in the $\alpha_5$-orientation of~$\sG$).

Now let~$e_1,\dotsc,e_\ell$ be the edges of the cycle~$C$ of $T$ and
let~$c_i$ be the color of edge~$e_i$.  Then the number of edges
pointing from~$C$ into the interior (in the $\alpha_5$-orientation
of~$\sG$) is
\[ \sum_{i=1}^{\ell - 1} (\pi(c_{i+1}) - \pi(c_i) + 2) + (\pi(c_1) 
  - \pi(c_{\ell}) + 2) = 2 \ell \enspace . \]
This is in contradiction to Lemma~\ref{lem:number_edges_into_cycle}.
\end{proof}

\begin{proposition} \label{prop:cyle_properties_in_fcf}
Let $C$ be an oriented cycle in~$G$ where~$G$ is oriented with every~$T_i$. Then:
\begin{enumerate}[(i)]
\item $C$ uses at least $3$ different colors.
\item If $C$ uses exactly $3$ different colors, these colors are not
  consecutive in the cyclic order.
\item $C$ has two consecutive edges whose colors have distance at most
  one in the cyclic order.
\end{enumerate}
\end{proposition}

\begin{proof}
  For the first two statements we denote by~$J$ the set of colors used by~$C$.
  Assume that~${|J| \leq 2}$ or $|J|=\{ j,j+1,j+2 \}$ for a color~$j$.
  In both cases there is a color~$i$ such that~$J \subseteq \{ i, i-1, i+1 \}$ (in the second case we choose~$i=j+1$).
  Thus~$C$ is an oriented cycle in the orientation~$T_i + T_{i-1} + T_{i+1} + T_{i-2}^{-1} + T_{i+2}^{-1}$, in contradiction to Lemma~\ref{lem:acyclic_orientation}.

  For the third statement assume that there is an oriented simple
  cycle~$C$ such that the distance of the colors of any two
  consecutive edges on~$C$ is exactly~$2$.  Because of the symmetry it
  suffices to consider the case that~$C$ is oriented clockwise.
  Let~$e=uv$ and~$e'=vw$ be two consecutive edges on~$C$ and let~$i$
  be the color of~$e$.  If the color of~$e'$ is~$i+2$, there is no
  edge pointing from~$v$ into the interior of~$C$, i.e., there is no outgoing edge of~$v$ in the
  interval between the edges~$e'$ and~$e$ in the clockwise cyclic order of the incident edges of~$v$.  If the color of~$e'$
  is~$i-2$, there are exactly~$4$ edges pointing from~$v$ into the interior
  of~$C$.  Therefore the total number of edges pointing into
  the interior of~$C$ is even, in contradiction to
  Lemma~\ref{lem:number_edges_into_cycle}.
\end{proof}

Now we will define the paths starting with an inner edge $e$ and
ending at an outer vertex that allow us to define the color of $e$.
The idea is to always continue with the opposite outgoing edge, but if
we run into a stack vertex, we need to be careful.  The paths we will
define are not unique, but we will see that all paths starting with
the same edge $e$ end at the same outer vertex.

\begin{definition} \label{def:paths}
Let~$e=uv$ be an inner edge such that~$u$ is a normal vertex.  We will
recursively define a set~$\mathcal{P}(e)$ of walks starting with~$e$
by distinguishing several cases concerning~$v$.
\begin{itemize}
\item If~$v$ is an outer vertex, i.e.,~$v=a_i$ for some~$i$, the
  set~$\mathcal{P}(e)$ contains only one path, the path only
  consisting of the edge~$e$.
\item If~$v$ is an inner normal vertex, let~$e'$ be the opposite
  outgoing edge of~$e$ at~$v$, i.e., the third outgoing edge in
  clockwise or counterclockwise direction, and we
  define~$\mathcal{P}(e) := \{ e + P : P \in \mathcal{P}(e') \}$.
\item If~$v$ is a stack vertex, let~$e_1'=v v_1'$ and~$e_2'= v v_2'$
  be the left and right outgoing edge of~$v$.  Further let~$e_1''$ be
  the second outgoing edge of~$v_1'$ after~$e_1'$ in counterclockwise
  direction and~$e_2''$ the second outgoing edge of~$v_2'$
  after~$e_2'$ in clockwise direction.  Note that~$e_i''$ is well
  defined if~$v_i'$ is not an outer vertex, and that not both
  of~$v_1'$ and~$v_2'$ can be outer vertices (there are no stack vertices
  in faces of~$G$ that are incident to two outer vertices).  If both of~$e_1''$
  and~$e_2''$ are well defined, we define~$\mathcal{P}(e) := \{ e +
  e_1' + P : P \in \mathcal{P}(e_1'') \} \cup \{ e + e_2' + P : P \in
  \mathcal{P}(e_2'') \} $.  If only~$e_i''$ is well defined, we
  define~$\mathcal{P}(e) := \{ e + e_i' + P : P \in \mathcal{P}(e_i'')
  \}$. See Fig.~\ref{fig:path_def} (left) for an example.
\end{itemize}
\end{definition}

\begin{figure}

\centering

\tikzstyle{normal vertex}=[circle,fill,scale=0.5]
\tikzstyle{stack vertex}=[circle,draw,scale=0.5]
\tikzstyle{undirected edge}=[]
\tikzstyle{colored undirected edge}=[thick]
\tikzstyle{edge}=[-{Latex[scale length=1.8,scale width=1]}]
\tikzstyle{colored edge}=[thick,-{Latex[scale length=1.8,scale width=1.2]}]

\tikzstyle{pcr}=[very thick,color=gray!70]
\tikzstyle{pcr fill}=[fill=gray!30]

\begin{tikzpicture}[scale=1.8]

\node[stack vertex] (v) at (0,0) [label={[label distance=-0.1cm]95:\footnotesize $v$}] {};
\node[normal vertex] (u) at (180:1) [label={[label distance=-0.1cm]95:\footnotesize $u$}] {};
\node[normal vertex] (v1') at (60:1) {};
\node[normal vertex] (v2') at (-60:1) {};
\node (v1'') at ($(v1')+(1,0)$) {};
\node (v2'') at ($(v2')+(1,0)$) {};

\draw[edge] (u) -- (v) node[midway,above] {\footnotesize $e$};
\draw[edge] (v) -- (v1') node[midway,left] {\footnotesize $e_1'$};
\draw[edge] (v) -- (v2') node[midway,left] {\footnotesize $e_2'$};
\draw[edge,dashed] (v2') -- (v1');
\draw[edge,dashed] (v2') -- (u);
\draw[edge,dashed] (u) -- (v1');
\draw[edge] (v1') -- (v1'') node[midway,below] {\footnotesize $e_1''$};
\draw[edge] (v2') -- (v2'') node[midway,above] {\footnotesize $e_2''$};

\draw[edge,color=gray!20,dashed] (v1') --+(160:.4);
\draw[edge,color=gray!20,dashed] (v1') --+(85:.4);
\draw[edge,color=gray!20,dashed] (v1') --+(40:.4);
\draw[edge,color=gray!20,dashed] (v1') --+(-60:.4);
\draw[edge,color=gray!20,dashed] (v2') --+(-70:.4);
\draw[edge,color=gray!20,dashed] (v2') --+(-120:.4);

\node[normal vertex] at (60:1) [label={[label distance=-0.1cm]80:\footnotesize $v_1'$}] {};
\node[normal vertex] at (-60:1) [label={[label distance=-0.1cm]-80:\footnotesize $v_2'$}] {};
        
\end{tikzpicture}
\qquad
\begin{tikzpicture}[scale=1.8]

\node[stack vertex] (v) at (0,0) [label={[label distance=-0.1cm]95:\footnotesize $v$}] {};
\node[normal vertex] (u) at (180:1) [label={[label distance=-0.1cm]95:\footnotesize $u$}] {};
\node[normal vertex] (v1') at (60:1) {};
\node[normal vertex] (v2') at (-60:1) [label={[label distance=-0.1cm]-80:\footnotesize $w$}] {};
\node[normal vertex] (v2'') at ($(v2')+(1,0)$) [label={[label distance=-0.1cm]-80:\footnotesize $w'$}] {};
\node[normal vertex] (w) at ($(u)+(-1,0)$) [label={[label distance=-0.1cm]95:\footnotesize $u'$}] {};

\draw[edge, color=green] (w) -- (u);
\draw[edge, color=green, dashed] (u) -- (v);
\draw[edge] (v) -- (v1');
\draw[edge, color=green, dashed] (v) -- (v2');
\draw[edge] (v2') -- (v1');
\draw[edge, color=green] (v2') -- (u);
\draw[edge] (u) -- (v1');
\draw[edge, color=green] (v2') -- (v2'');
        
\end{tikzpicture}

\caption{Left: An example for the third case in the definition of the set $\mathcal{P}(e)$ of walks.
         Right: An example for the construction of the shortcut walk: The stack edges~$uv$ and~$vw$ of the walk~$(u',u,v,w,w')$
           are replaced by the edge~$uw$ in the shortcut walk~$(u',u,w,w')$. In the $\alpha_5$-orientation this edge is oriented
           from~$w$ to~$u$.}

\label{fig:path_def}

\end{figure}
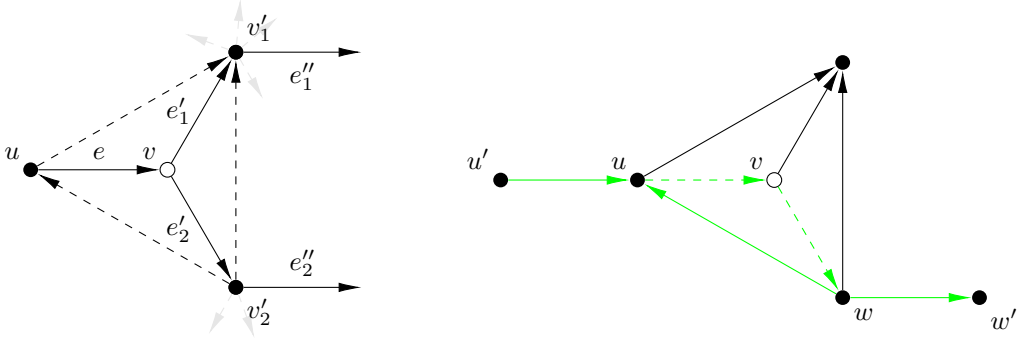

At the moment it is not clear that these walks are finite.  If they
are finite, they have to end in an outer vertex.  But we have to prove
that they do not cycle.

\begin{lemma} \label{lemma:path_properties}
\begin{enumerate}[(i)]
\item \label{item:no_cycle} The walks~$P\in \mathcal{P}(e)$ are paths,
  i.e., there are no vertex repetitions in~$P$.
\item \label{item:same_end_vertex} Let~$P_1,P_2 \in \mathcal{P}(e)$ be
  two paths starting with the same edge~$e$.  Then~$P_1$ and~$P_2$ end
  in the same outer vertex.
\item \label{item:no_intersection} Let~$v$ be a normal vertex and
  let~$e_1=v v_1, e_2 = v v_2$ be two different outgoing edges at~$v$.
  Further let~$P_1 \in \mathcal{P}(e_1)$ and~$P_2 \in
  \mathcal{P}(e_2)$ be two paths.  Then~$P_1$ and~$P_2$ do not cross
  and they end in different outer vertices.
\end{enumerate}
\end{lemma}

Before we can prove Lemma~\ref{lemma:path_properties}, we have to introduce
some notations. The general approach for the proof will be to produce
contradictions to Lemma~\ref{lem:number_edges_into_cycle}. Since Lemma~\ref{lem:number_edges_into_cycle}
is a statement about cycles only consisting of normal vertices and
the walks considered in Lemma~\ref{lemma:path_properties} consist of normal
and stack vertices, we consider abbreviations of these walks only consisting
of normal vertices.

\begin{definition}
  For a given edge~$e$ let~$P$ be a finite subwalk of a walk in~$\mathcal{P}(e)$
  that starts and ends with a normal vertex.  Then the
  \emph{shortcut}~$P'$ of~$P$ is obtained from~$P$ by replacing every
  consecutive pair~$uv,vw$ of edges, where~$v$ is a stack vertex, by
  the edge~$uw$.
  We call~$uw$ a \emph{shortcut edge}.
  Note that this edge might be oriented from~$w$ to~$u$ in the $\alpha_5$-orientation,
  and thus shortcut walks in general are not oriented walks in the $\alpha_5$-orientation.
  See Fig.~\ref{fig:path_def} (right) for an example.
\end{definition}

Let~$P$ be a finite subwalk of a walk defined in Definition~\ref{def:paths}.
For our later argumentation we need to be able to extend such a walk~$P$
by a normal edge before its first
edge or after its last edge such that the extended walks remains a subwalk
of a walk defined in Definition~\ref{def:paths}. For sure,~$P$ has to start or to end with a
normal vertex, respectively, to allow such an extension.
Since this condition is not sufficient,
we have to change the graph~$G$ in some cases. Thus we can find the
contradictions to Lemma~\ref{lem:number_edges_into_cycle} in the new graph.

\begin{definition}
Let~$X$ be a fixed $\alpha_5$-orientation of~$\sG$. Let~$w$ be a stack
vertex and let~$v_1,v_2,v_3$ be its neighbors in~$\sG$ such that in~$X$
the edge~$v_1w$ is incoming at~$w$ and the edges~$wv_2,wv_3$
are outgoing at~$w$. Then we call the following change of~$G$,~$\sG$ and~$X$
\emph{stacking a normal vertex} at~$w$:
The vertex~$w$ becomes a normal vertex and we add stack vertices~$u_1,u_2,u_3$
in the three incident faces of~$w$. All edges keep their orientations,
the edges~$wu_i$ are incoming at the~$u_i$ and the edges~$u_iv_j$ are
outgoing at the~$u_i$. See Fig.~\ref{fig:stacking_normal} for an example.
\end{definition}

\begin{figure}

\centering

\tikzstyle{normal vertex}=[circle,fill,scale=0.5]
\tikzstyle{stack vertex}=[circle,draw,scale=0.5]
\tikzstyle{undirected edge}=[]
\tikzstyle{colored undirected edge}=[thick]
\tikzstyle{edge}=[-{Latex[scale length=1.8,scale width=1]}]
\tikzstyle{colored edge}=[thick,-{Latex[scale length=1.8,scale width=1.2]}]

\tikzstyle{pcr}=[very thick,color=gray!70]
\tikzstyle{pcr fill}=[fill=gray!30]

\begin{tikzpicture}[scale=3]

\node[stack vertex] (w) at (0,0) [label={[label distance=-0.cm]95:\footnotesize $w$}] {};
\node[normal vertex] (v1) at (90:1) [label={180:\footnotesize $v_1$}] {};
\node[normal vertex] (v3) at (210:1) [label={180:\footnotesize $v_3$}] {};
\node[normal vertex] (v2) at (330:1) [label={0:\footnotesize $v_2$}] {};

\draw[edge] (v1) -- (w);
\draw[edge] (w) -- (v2);
\draw[edge] (w) -- (v3);
\draw[edge] (v2) -- (v1);
\draw[edge] (v2) -- (v3);
\draw[edge] (v3) -- (v1);
        
\end{tikzpicture}
\qquad
\begin{tikzpicture}[scale=3]

\node[normal vertex] (w) at (0,0) [label={[label distance=-0.cm]95:\footnotesize $w$}] {};
\node[normal vertex] (v1) at (90:1) [label={180:\footnotesize $v_1$}] {};
\node[normal vertex] (v3) at (210:1) [label={180:\footnotesize $v_3$}] {};
\node[normal vertex] (v2) at (330:1) [label={0:\footnotesize $v_2$}] {};

\node[stack vertex] (u1) at ($1/3*(v1)+1/3*(v2)+1/3*(w)$) [label={[label distance=-0.06cm]0:\footnotesize $u_1$}] {};
\node[stack vertex] (u2) at ($1/3*(v2)+1/3*(v3)+1/3*(w)$) [label={[label distance=-0.06cm]270:\footnotesize $u_2$}] {};
\node[stack vertex] (u3) at ($1/3*(v3)+1/3*(v1)+1/3*(w)$) [label={[label distance=-0.13cm]180:\footnotesize $u_3$}] {};

\draw[edge] (v1) -- (w);
\draw[edge] (w) -- (v2);
\draw[edge] (w) -- (v3);
\draw[edge] (v2) -- (v1);
\draw[edge] (v2) -- (v3);
\draw[edge] (v3) -- (v1);

\draw[edge] (w) -- (u1);
\draw[edge] (w) -- (u2);
\draw[edge] (w) -- (u3);
\draw[edge] (u1) -- (v1);
\draw[edge] (u1) -- (v2);
\draw[edge] (u2) -- (v2);
\draw[edge] (u2) -- (v3);
\draw[edge] (u3) -- (v3);
\draw[edge] (u3) -- (v1);
        
\end{tikzpicture}

\caption{An example for stacking a normal vertex into a face of~$G$.}

\label{fig:stacking_normal}

\end{figure}
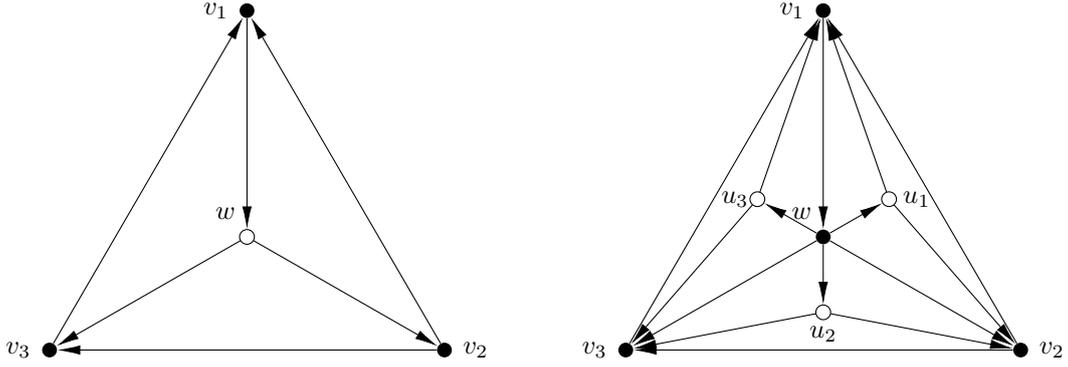

Note that after stacking a normal vertex at~$w$ the walks defined in Definition~\ref{def:paths}
change in the following way: After each occurrence of~$w$ the vertex~$u_2$ is
inserted. If a walk does not contain the vertex~$w$, it remains the same.

\begin{lemma}\label{lemma:extending_by_normal_edge}
Let~$P$ be a finite subwalk of a walk defined in Definition~\ref{def:paths} that starts (ends)
with an inner normal vertex. Then after stacking at most two normal vertices,~$P$ can be extended by a normal edge
before its first edge (after its last edge) in such a way that it remains a subwalk
of a walk defined in Definition~\ref{def:paths}.
\end{lemma}

\begin{proof}
Let~$P=v_1,v_2,\dotsc,v_n$. Let us first consider the case that~$v_1$
is a normal vertex and that we want to extend~$P$ by a normal edge~$v_0v_1$.
We distinguish three cases. In the first case~$v_1$ has an
incoming normal edge~$v_0 v_1$ such that~$v_1v_2$ is the third outgoing
edge of~$v_1$ in clockwise (and counterclockwise) order after~$v_0v_1$.
Then we are done because due to Definition~\ref{def:paths} (second case)~$v_1v_2$
is the unique successor of~$v_0 v_1$. In the second case~$v_1$ has an
incoming stack edge~$v_0 v_1$ such that~$v_1v_2$ is the third outgoing
edge of~$v_1$ in clockwise (and counterclockwise) order after~$v_0v_1$.
Then we can stack a normal vertex at~$v_0$ and we are in the first case, again.
In the third case~$v_1$ has no incoming edge~$v_0 v_1$ such that~$v_1v_2$ is the third outgoing
edge of~$v_1$ in clockwise (and counterclockwise) order after~$v_0v_1$.
Then~$v_1$ has an outgoing stack edge~$v_1 w$ and an outgoing normal edge~$v_1 w'$
such that~$v_1 v_2$ is the second outgoing edge of~$v_1$
either in clockwise or in counterclockwise order after~$v_1 w$ and after~$v_1 w'$ since the neighbors of~$v_1$
alternate between normal vertices and stack vertices. If we stack a normal vertex at~$w$,
a stack vertex~$u$ is stacked into the face~$v_1 w w'$ of~$\sG$.
Since the edge~$u v_1$ is incoming at~$v_1$, we are in the second case, again.

Now let us consider the case that~$v_n$ is an inner normal vertex and that we
want to extend~$P$ by a normal edge~$v_n v_{n+1}$. Due to Definition~\ref{def:paths} (second
and third case) there is exactly one outgoing edge~$v_n v_{n+1}$ of~$v_n$ such that
the extended walk~$v_1,\dotsc,v_{n+1}$ is a subwalk of a walk defined in Definition~\ref{def:paths}.
If~$v_{n+1}$ is a normal vertex, we are done. Otherwise we stack a normal vertex at~$v_{n+1}$
and are also done.
\end{proof}

If~$P$ is a path with designated start vertex~$s$, then at an inner
vertex~$v$ we can distinguish edges on the right side of~$P$ and edges
on the left side of~$P$.  We define~$\ro_{P}(v)$ and $\lo_{P}(v)$ to
be the number of right and left outgoing edges from path~$P$
at vertex~$v$, respectively.

\begin{lemma} \label{lem:shortcut_outgoing_edges} Let $P'$ be a
  shortcut walk of length~$\ell$ that does not start and does not end
  with a shortcut edge.  Then the number of edges pointing from the
  inner vertices of~$P'$ to the right (left) of~$P'$
  is~$2(\ell-1)$.
\end{lemma}

\begin{proof}
  Since right and left are symmetric, it is enough to prove the lemma
  for right outgoing edges.  The proof is by induction on the
  length~$\ell$ of the walk.
  If~$\ell=2$, the shortcut walk is equal to the original walk and
  consists of two normal edges since by assumption the shortcut walk does not start and does not
  end with a shortcut edge. Due to Definition~\ref{def:paths}~(second case)
  this path has~$2 = 2 (\ell - 1)$ outgoing edges on either side.
  Now let~$\ell \geq 3$ and let~$P'=v_0,e_1,v_1,\dotsc,e_{\ell},v_{\ell}$.  If~$e_{\ell-1}$ is
  not a shortcut edge, the statement follows by induction because the
  vertex between two consecutive normal edges of a walk
  in~$\mathcal{P}(e)$ has~$2$ outgoing edges on either side
  by Definition~\ref{def:paths}~(second case).

  Now assume that~$e_{\ell-1}$ is a shortcut edge and
  let~$e'=v_{\ell-2}w,e''=wv_{\ell-1}$ be the corresponding edges of
  the original path~$P$.  Let~$Q'$ be the subwalk of~$P'$ starting
  at~$v_0$ and ending at~$v_{\ell-2}$ extended by the edge
  $v_{\ell-2}w$.  Note that we can apply the induction hypothesis
  to~$Q'$ by stacking a normal vertex at~$w$.  Hence, we know
  that the number of edges pointing from an inner vertex of~$Q'$ to the right is~$2
  (\ell-2)$.

  The edges~$v_{\ell-2}w$ and~$v_{\ell-2}v_{\ell-1}$ are consecutive
  in the cyclic order of incident edges of~$v_{\ell-2}$.  Define a
  sign~$\sigma$ to be~$-1$ if~$v_{\ell-2}v_{\ell-1}$ is right
  of~$v_{\ell-2}w$ and~$+1$ otherwise.  Let~$\delta = 0$
  if~$\sigma=-1$ and~$v_{\ell-2}v_{\ell-1}$ is outgoing
  at~$v_{\ell-1}$, otherwise~$\delta=1$.  When comparing outgoing
  edges at~$v_{\ell-2}$ in~$P'$ and~$Q'$ the contribution of~$\delta$
  will account for the edge~$v_{\ell-2}v_{\ell-1}$ if~$\sigma=-1$
  (if~$\delta=1$ and~$\sigma=-1$, the edges pointing from~$v_{\ell-2}$ to the right of~$P'$
  are exactly the edges pointing from~$v_{\ell-2}$ to the right of~$Q'$ minus the edge~$v_{\ell-2}v_{\ell-1}$) and
  for the edge~$v_{\ell-2}w$ if~$\sigma=+1$ (if~$\delta=1$ and~$\sigma=+1$, the edges pointing from~$v_{\ell-2}$ to the right of~$P'$
  are exactly the edges pointing from~$v_{\ell-2}$ to the right of~$Q'$ plus the edge~$v_{\ell-2}w$).  It follows
  that~$\ro_{P'}(v_{\ell-2}) = \ro_{Q'}(v_{\ell-2}) + \sigma \delta$,
  see Fig.~\ref{fig:shortcut}.
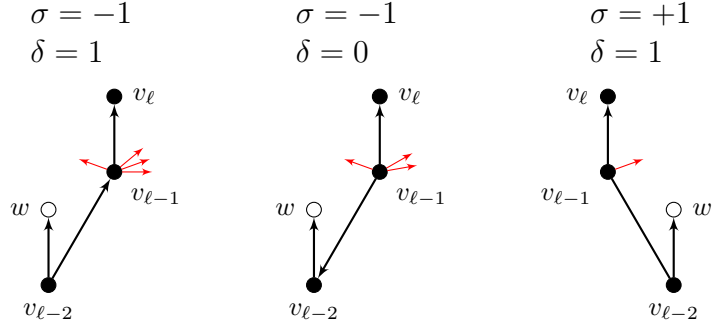
\begin{figure}

\centering

\tikzstyle{vertex}=[circle,scale=0.5]
\tikzstyle{normal vertex}=[vertex,draw,fill]
\tikzstyle{stack vertex}=[vertex,draw]
\tikzstyle{edge}=[-latex',thick]
\tikzstyle{undirected edge}=[thick]
\tikzstyle{other edges}=[-latex',color=red]

\begin{tikzpicture}[scale=1]

\node[normal vertex] (v2) at (0,0) [label={below}: \footnotesize $v_{\ell-2}$] {};
\node[stack vertex] (w) at (90:1) [label={left}: \footnotesize $w$] {};
\node[normal vertex] (v1) at ($(w)+(30:1)$) [label={below right}: \footnotesize $v_{\ell-1}$] {};
\node[normal vertex] (v0) at ($(v1)+(90:1)$) [label={right}: \footnotesize $v_{\ell}$] {};

\draw[edge] (v2) -- (w);
\draw[edge] (v2) -- (v1);
\draw[edge] (v1) -- (v0);

\node[align=left] at ($(90:3.1)+(30:.5)$) {$\sigma=-1$\\$\delta=1$};

\draw[other edges] (v1) -- +(160:.5);
\draw[other edges] (v1) -- +(0:.5);
\draw[other edges] (v1) -- +(20:.5);
\draw[other edges] (v1) -- +(40:.5);

\end{tikzpicture}
\qquad
\begin{tikzpicture}[scale=1]

\node[normal vertex] (v2) at (0,0) [label={below}: \footnotesize $v_{\ell-2}$] {};
\node[stack vertex] (w) at (90:1) [label={left}: \footnotesize $w$] {};
\node[normal vertex] (v1) at ($(w)+(30:1)$) [label={below right}: \footnotesize $v_{\ell-1}$] {};
\node[normal vertex] (v0) at ($(v1)+(90:1)$) [label={right}: \footnotesize $v_{\ell}$] {};

\draw[edge] (v2) -- (w);
\draw[edge] (v1) -- (v2);
\draw[edge] (v1) -- (v0);

\node[align=left] at ($(90:3.1)+(30:.5)$) {$\sigma=-1$\\$\delta=0$};

\draw[other edges] (v1) -- +(160:.5);
\draw[other edges] (v1) -- +(10:.5);
\draw[other edges] (v1) -- +(30:.5);

\end{tikzpicture}
\qquad
\begin{tikzpicture}[scale=1]

\node[normal vertex] (v2) at (0,0) [label={below}: \footnotesize $v_{\ell-2}$] {};
\node[stack vertex] (w) at (90:1) [label={right}: \footnotesize $w$] {};
\node[normal vertex] (v1) at ($(w)+(150:1)$) [label={below left}: \footnotesize $v_{\ell-1}$] {};
\node[normal vertex] (v0) at ($(v1)+(90:1)$) [label={left}: \footnotesize $v_{\ell}$] {};

\draw[edge] (v2) -- (w);
\draw[undirected edge] (v2) -- (v1);
\draw[edge] (v1) -- (v0);

\node[align=left] at ($(90:3.1)+(150:.5)$) {$\sigma=+1$\\$\delta=1$};

\draw[other edges] (v1) -- +(20:.5);

\end{tikzpicture}

\caption{Right outgoing edges at~$v_{\ell-2}$ and~$v_{\ell-1}$.}

\label{fig:shortcut}

\end{figure}

\begin{myclaim} \label{claim:right_out}
  $\ro_{P'}(v_{\ell-1})= 2 - \sigma \delta$. 
\end{myclaim}

\begin{claimproof}
  If~$\sigma=+1$, then we are in the~$v_1'$ case of the stack vertex part
  in Definition~\ref{def:paths}, i.e., between~$v_{\ell-2}v_{\ell-1}$
  and~$v_{\ell-1}v_\ell$ there is one outgoing edge on the right
  at~$v_{\ell-1}$.  The claim follows because in this case~$\delta=1$.

  If~$\sigma=-1$, then we are in the~$v_2'$ case of the stack vertex
  part in Definition~\ref{def:paths}, i.e., between~$v_{\ell-2}v_{\ell-1}$
  and~$v_{\ell-1}v_\ell$ there is one outgoing edge on the left
  at~$v_{\ell-1}$.  Now it depends on whether the
  edge~$v_{\ell-2}v_{\ell-1}$ is outgoing at~$v_{\ell-2}$ or
  at~$v_{\ell-1}$.  In the first case we have~$\delta=1$
  and~${\ro_{P'}(v_{\ell-1}) = 3}$, in the second case~$\delta=0$
  and~${\ro_{P'}(v_{\ell-1}) = 2}$.  In either case this is what the
  claim says.
\end{claimproof}

The number of right outgoing edges of~$P'$ is obtained as the sum of
those of~$Q'$, which is~$2 (\ell-2)$, with~$\ro_{P'}(v_{\ell-2}) -
\ro_{Q'}(v_{\ell-2})$, which is~$\sigma \delta$,
and~$\ro_{P'}(v_{\ell-1})$, which is~$2 - \sigma \delta$ by Claim~\ref{claim:right_out}.
Hence, the number of right outgoing edges of~$P'$ is~$2 (\ell-1)$.
\end{proof}

\begin{lemma} \label{lem:general_shortcut_outgoing_edges} Let $P'$ be a
  shortcut walk of length~$\ell$ (that might start or end
  with a shortcut edge).  Then the number of edges pointing from the
  inner vertices of~$P'$ to the right (left) of~$P'$
  is~$2(\ell-1)+\mu$ with~$-2 \leq \mu \leq 2$.
\end{lemma}

\begin{proof}
  Because of symmetry it is enough to proof the lemma for the right
  outgoing edges of~$P'$.
  Let~$P$ be the original walk. Further let~$Q$ be the walk~$P$
  extended by a normal edge at both ends (possibly after stacking normal vertices). We denote the shortcut of~$Q$
  by~$Q'=v_{-1},e_0,v_0,e_1,v_1,\dotsc,e_{\ell},v_{\ell},e_{\ell+1},v_{\ell+1}$.
  Due to Lemma~\ref{lem:shortcut_outgoing_edges} there are exactly~$2 (\ell + 1)$ edges
  pointing from the inner vertices of~$Q'$ to the right of~$Q'$.
  From Claim~\ref{claim:right_out} in the proof of Lemma~\ref{lem:shortcut_outgoing_edges}
  (see also Fig.~\ref{fig:shortcut}) it follows that there are~$\mu_1 \in \{ 1,2,3 \}$
  edges pointing from~$v_{\ell}$ to the right of~$Q'$. With a similar case distinction
  (see Fig.~\ref{fig:shortcut_start}) we can see that there are~$\mu_2 \in \{ 1,2,3 \}$ edges pointing from~$v_0$
  to the right of~$Q'$. Therefore there are exactly~$2 (\ell + 1) - \mu_1 - \mu_2 = 2 (\ell - 1) + (4-\mu_1-\mu_2)$
  edges pointing from the inner vertices of~$P'$ to the right of~$P'$. Since~$-2 \leq 4-\mu_1-\mu_2 \leq 2$,
  this completes the proof.
\end{proof}

\begin{figure}

\centering

\tikzstyle{vertex}=[circle,scale=0.5]
\tikzstyle{normal vertex}=[vertex,draw,fill]
\tikzstyle{stack vertex}=[vertex,draw]
\tikzstyle{edge}=[-latex',thick]
\tikzstyle{undirected edge}=[thick]
\tikzstyle{other edges}=[-latex',color=red]

\begin{tikzpicture}[scale=1]

\node[normal vertex] (v2) at (0,0) [label={below right}: \footnotesize $v_{0}$] {};
\node[stack vertex] (w) at (90:1) [label={left}: \footnotesize $w$] {};
\node[normal vertex] (v1) at ($(w)+(30:1)$) [label={below right}: \footnotesize $v_{1}$] {};
\node[normal vertex] (v0) at ($(v2)+(270:1)$) [label={right}: \footnotesize $v_{-1}$] {};

\draw[edge] (v2) -- (w);
\draw[edge] (w) -- (v1);
\draw[edge] (v2) -- (v1);
\draw[edge] (v0) -- (v2);

\draw[other edges] (v2) -- +(10:.5);

\end{tikzpicture}
\qquad
\begin{tikzpicture}[scale=1]

\node[normal vertex] (v2) at (0,0) [label={below right}: \footnotesize $v_{0}$] {};
\node[stack vertex] (w) at (90:1) [label={left}: \footnotesize $w$] {};
\node[normal vertex] (v1) at ($(w)+(30:1)$) [label={below right}: \footnotesize $v_{1}$] {};
\node[normal vertex] (v0) at ($(v2)+(270:1)$) [label={right}: \footnotesize $v_{-1}$] {};

\draw[edge] (v2) -- (w);
\draw[edge] (w) -- (v1);
\draw[edge] (v1) -- (v2);
\draw[edge] (v0) -- (v2);

\draw[other edges] (v2) -- +(0:.5);
\draw[other edges] (v2) -- +(20:.5);

\end{tikzpicture}
\qquad
\begin{tikzpicture}[scale=1]

\node[normal vertex] (v2) at (0,0) [label={below right}: \footnotesize $v_{0}$] {};
\node[stack vertex] (w) at (90:1) [label={right}: \footnotesize $w$] {};
\node[normal vertex] (v1) at ($(w)+(150:1)$) [label={below left}: \footnotesize $v_{1}$] {};
\node[normal vertex] (v0) at ($(v2)+(270:1)$) [label={left}: \footnotesize $v_{-1}$] {};

\draw[edge] (v2) -- (w);
\draw[edge] (w) -- (v1);
\draw[undirected edge] (v2) -- (v1);
\draw[edge] (v0) -- (v2);

\draw[other edges] (v2) -- +(10:.5);
\draw[other edges] (v2) -- +(30:.5);

\end{tikzpicture}
\qquad
\begin{tikzpicture}[scale=1]

\node[normal vertex] (v2) at (0,0) [label={below right}: \footnotesize $v_{0}$] {};
\node[normal vertex] (v1) at ($(v2)+(90:1)$) [label={below left}: \footnotesize $v_{1}$] {};
\node[normal vertex] (v0) at ($(v2)+(270:1)$) [label={left}: \footnotesize $v_{-1}$] {};

\draw[edge] (v2) -- (v1);
\draw[edge] (v0) -- (v2);

\draw[other edges] (v2) -- +(10:.5);
\draw[other edges] (v2) -- +(30:.5);

\end{tikzpicture}

\caption{Right outgoing edges at~$v_{0}$.}

\label{fig:shortcut_start}

\end{figure}
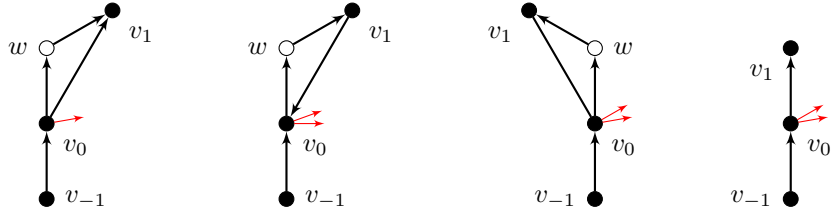

Now we are ready to prove Lemma~\ref{lemma:path_properties}.

\begin{proof}[Proof of Lemma~\ref{lemma:path_properties}]
  For~(\ref{item:no_cycle}) assume that~$P$ cycles. Let~$C$ be a
  simple cycle which appears as a subwalk of the shortcut of~$P$, and
  let~$\ell$ be its length.  According to
  Lemma~\ref{lem:general_shortcut_outgoing_edges}, there are at least~$2(\ell-1) - 2$
  edges pointing into the interior of~$C$.  This is in contradiction
  to Lemma~\ref{lem:number_edges_into_cycle} which states that there are
  only~$2\ell-5$ edges pointing into the interior of~$C$.

  \medskip For~(\ref{item:same_end_vertex}) assume that~$P_1$
  and~$P_2$ coincide up to a vertex~$v^*$, then~$P_1$ goes to the left
  and~$P_2$ to the right.  Note that~$v^*$ has to be a stack vertex
  and let~$v$ and~$v'$ be its predecessors in~$P_1$ and~$P_2$, i.e.,
  $v',v,v^*$ appear in this order on both paths.  If~$P_1$ and~$P_2$ start
  in~$v$, we can use a dummy edge~$v'v$ (possibly after stacking normal vertices) which makes sure
  that~$\ro_{P_1}(v)=\lo_{P_2}(v)= 2$.  For~$i=1,2$ let~$P_i'$ be the
  shortcut of~$P_i$.

  \begin{myclaim} \label{claim:count_at_v} $\ro_{P'_1}(v)+\lo_{P'_2}(v)
    = 6$.
\end{myclaim}

\begin{claimproof}
  Every outgoing edge of~$v$ except possibly the edge~$vv'$ is a right
  edge with respect to~$P_1'$ or a left edge with respect to~$P_2'$.
  The edge from~$v$ to~$v^*$ is both and therefore
  counted twice, see Fig.~\ref{fig:overcount}
  (left).
\end{claimproof}
\begin{figure}
\centering
\includegraphics[]{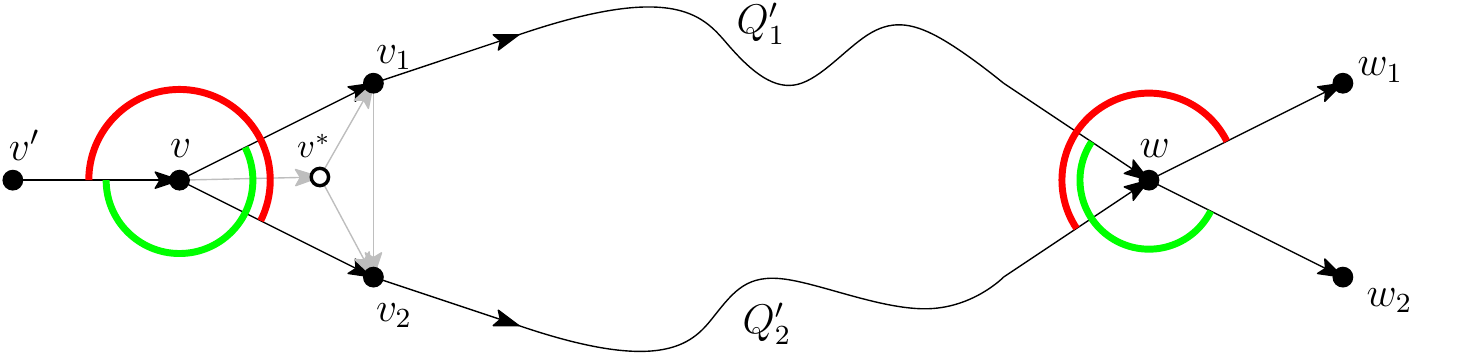}
\caption{Overcounts at the two ends of the paths~$Q_1'$ and~$Q_2'$.
Edges in the green area are counted as right edges of~$Q_1'$, 
those in the red area as left edges of~$Q_2'$.}
\label{fig:overcount}
\end{figure}

\begin{myclaim}
  If after splitting at~$v$ the two shortcut paths~$P'_1$ and~$P'_2$
  meet at some vertex~$w$ which is not an outer vertex, then the first
  vertex~$w'$ after~$w$ is the same on~$P_1$ and~$P_2$.
\end{myclaim}

\begin{claimproof}
For~$i=1,2$ let~$Q_i'$ be the subpath of~$P'_i$ starting
at~$v$ and ending at~$w$.  At the beginning these paths are extended
by a normal edge~$v'v$ (possibly after stacking normal vertices). At the end the paths are
extended by the successor~$w_i$ of~$w$ on~$P_i$.
If~$w_i$ is a stack vertex, we can pretend that it is a normal vertex by stacking a normal vertex at~$w_i$.

Let~$\ell_i$ be the length of~$Q_i'$.  From
Lemma~\ref{lem:shortcut_outgoing_edges} we know that there are
exactly~$2(\ell_1-1)$ edges pointing from~$Q_1'$ to the right and
exactly~$2(\ell_2-1)$ edges pointing from~$Q_2'$ to the left.  Let~$C$
be the cycle formed by the two shortcuts~$Q_1'$ and~$Q_2'$ between~$v$
and~$w$.  The length of~$C$ is~$\ell_1+\ell_2-4$ ($C$ consists of the edges of~$Q_1'$ and~$Q_2'$
except for their first and last edges) and therefore, by
Lemma~\ref{lem:number_edges_into_cycle}, there are
exactly~$2(\ell_1+\ell_2-4)-5$ edges pointing into the interior
of~$C$.

The sum of the right edges of~$Q_1'$ and the left edges of~$Q_2'$
correctly accounts for the edges pointing into the interior of~$C$ at
all vertices except at~$v$ and~$w$.  Claim~\ref{claim:count_at_v} implies that at~$v$ we overcount by exactly~$5$,
i.e., the number of edges pointing from~$Q_1'$ to the right
plus the number of edges pointing from~$Q_2'$ to the left is~$6$,
but at~$v$ only~$1$ edge is pointing into the interior of~$C$.
It follows that the overcount at~$w$ is
\[ \left( 2(\ell_1-1) + 2(\ell_2-1) \right) - \left( 2
(\ell_1+\ell_2-4) - 5 \right) - 5 = 4 \enspace ,
\] where~$2(\ell_1-1) + 2(\ell_2-1)$ is the number
of edges pointing from~$Q_1'$ to the right plus the number of edges pointing
from~$Q_2'$ to the left,~$2(\ell_1+\ell_2-4) - 5$ is the number of edges
pointing into the interior of~$C$ and~$5$ is the overcount at~$v$.
This means that each edge except one contributes to the overcount
at~$w$, see Fig.~\ref{fig:overcount} (right).  Hence,~$ww_1$ and~$ww_2$
have to be identical and~$w'=w_1=w_2$.
\end{claimproof}

Now suppose that~$P_1$ and~$P_2$ end at different outer vertices~$a_i$
and~$a_j$.  Let~$v^*$ be the last common vertex of the two
paths,~$v^*$ has to be a stack vertex.  Assume that~$P_1$ goes to the
left at~$v^*$ and~$P_2$ goes to the right.  Let~$v$ be the predecessor
of~$v^*$ in~$P_1$ and~$P_2$.  Further let~$Q_1$ and~$Q_2$ be the
subpaths of~$P_1$ and~$P_2$ starting at~$v$ and ending at~$a_i$
and~$a_j$, extended by a normal edge~$v'v$ (possibly after stacking normal vertices).
For~$i=1,2$ let~$Q_i'$ be the shortcut of~$Q_i$ and
let~$\ell_i+1$ be the length of~$Q_i'$.  Let~$\ell_3 \geq 2$ be the
length of the path~$P_3$ between~$a_i$ and~$a_j$ that alternates
between outer and inner normal vertices.  Note that~$P_3$
has~$\frac{\ell_3}{2}$ inner vertices and~$\frac{\ell_3}{2}+1$ outer
vertices.  Let~$C$ be the simple cycle in the union of~$Q_1'$,~$Q_2'$
and~$P_3$.  Note that~$P_3$ can have at most one common edge with
each~$Q'_i$.  Let~$\xi\in\{0,1,2\}$ be the number of common edges
of~$P_3$ with~$Q_1'\cup Q_2'$.  The length of~$C$
is~$\ell_1+\ell_2+\ell_3 -2\,\xi$.  If~$\xi=0$ the edges pointing
into~$C$ can be obtained by adding the right edges of~$Q'_1$ and the
left edges of~$Q'_2$ subtracting~$5$ for the overcount at~$v$ and
adding~$3$ for each normal vertex of~$P_3$.  When~$Q'_1$ shares the
last edge with~$P_3$, then we have to disregard the contribution of
the first normal vertex~$w$ of~$P_3$ and subtract one for the
edge~$wa_{i+1}$ which belongs to~$C$ and is counted as a right edge
of~$Q'_1$.  Hence, we obtain the following estimate for the number of
edges pointing into~$C$:
\[
  2 \ell_1 + 2 \ell_2 - 5 -\xi + 3(\frac{\ell_3}{2} -\xi)  = 
2 (\ell_1+\ell_2+\ell_3-2\;\xi) - 5 - \frac{\ell_3}{2} \enspace . 
\]
This is less than the number of edges pointing into~$C$ which is~$2
(\ell_1+\ell_2+\ell_3-2\xi) - 5$ by
Lemma~\ref{lem:number_edges_into_cycle}.
See Fig.~\ref{fig:paths_same_outer} for an illustration.

\begin{figure}
\centering
\includegraphics[]{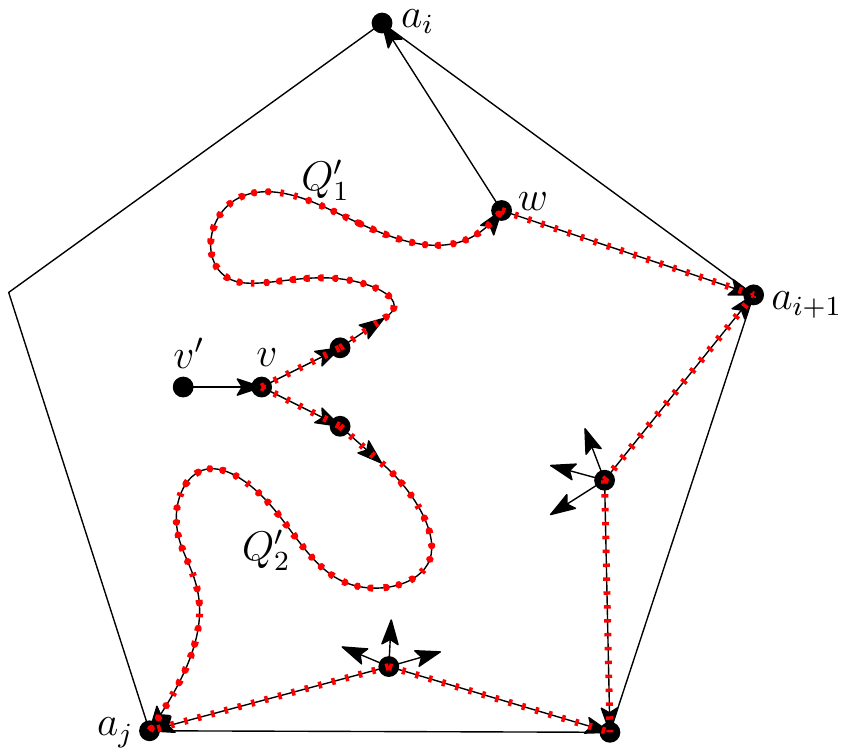}
\caption{Illustration of the proof of Lemma~\ref{lemma:path_properties}~(\ref{item:same_end_vertex}).
  The cycle~$C$ is drawn in red and dotted.
  The edge~$wa_i$ is contained in~$Q_1'$ and in~$P_3$.
  Therefore this is the case~$\xi = 1$.}
\label{fig:paths_same_outer}
\end{figure}

\medskip For~(\ref{item:no_intersection}) assume that~$P_1$ and~$P_2$
have a common vertex different from~$v$ and let~$w$ be the first
normal vertex of this kind (note that~$P_1$ and~$P_2$ might already
meet at a stack vertex immediately before~$w$).  Let~$Q_1$ and~$Q_2$
be the subpaths of~$P_1$ and~$P_2$ that begin with the
predecessors~$v_1'$ and~$v_2'$ of~$v$, respectively, and end with the
successors~$w_1$ and~$w_2$ of~$w$, respectively. After possibly stacking
some normal vertices, we can assume that~$v_1',v_2'$ exist and that~$v_1',v_2',w_1,w_2$
are normal vertices.
Let~$Q_1'$ and~$Q_2'$ be the corresponding shortcut paths and
let~$\ell_1$ and~$\ell_2$ be their lengths.  Note that the successor
edge of~$v$ in~$Q_i$ and~$Q'_i$ is the same.  Let~$C$ be the cycle we
get by gluing the parts of~$Q_1'$ and~$Q_2'$ between~$v$ and~$w$
together.  Let~$s$ be the number of outgoing edges of~$v$
between~$Q_1'$ and~$Q_2'$ inside~$C$.  We assume that looking
from~$v$ into the interior of~$C$, the left edge of~$C$ belongs
to~$Q_1'$ and the right one to~$Q_2'$.  The number of outgoing edges
between~$v_2'v$ and~$v_1'v$ is~$5 - (\ro_{Q_1}(v) +\lo_{Q_2}(v) - s) =
s+1$.  The length of~$C$ is~$\ell_1+\ell_2-4$ and therefore, due to
Lemma~\ref{lem:number_edges_into_cycle}, exactly~$2 (\ell_1+\ell_2-4) - 5$
edges are pointing into the interior of~$C$.  If we add the number of
edges pointing from~$Q_1'$ to the right and the number of edges
pointing from~$Q_2'$ to the left, we overcount by~$5-(s+1)$ at~$v$.
Hence, the overcount at~$w$ must be
\[
  \left( 2(\ell_1-1) + 2(\ell_2-1) \right) - 
(5-(s+1)) - \left( 2 (\ell_1+\ell_2-4) - 5 \right)= 5+s \enspace .
\] 
To get an overcount~$\geq 5$ at~$w$ we need to have every edge in the
union of the right edges of~$Q'_1$ and left edges of~$Q'_2$.  In
particular~$ww_1$ is to the left of~$Q'_2$.  This implies that at
least one of the edges of~$Q_1'$ and~$Q_2'$ ending in~$w$ must be a
shortcut edge.  It follows that~$w$ is not an outer vertex.  Further
there are exactly~$s$ outgoing edges of~$w$ between~$ww_1$ and~$ww_2$.
Therefore we can inductively repeat the argument for the subpaths
of~$P_1$ and~$P_2$ starting at~$w$.
See Fig.~\ref{fig:paths_no_crossing} for an illustration.
\end{proof}

\begin{figure}
\centering
\includegraphics[]{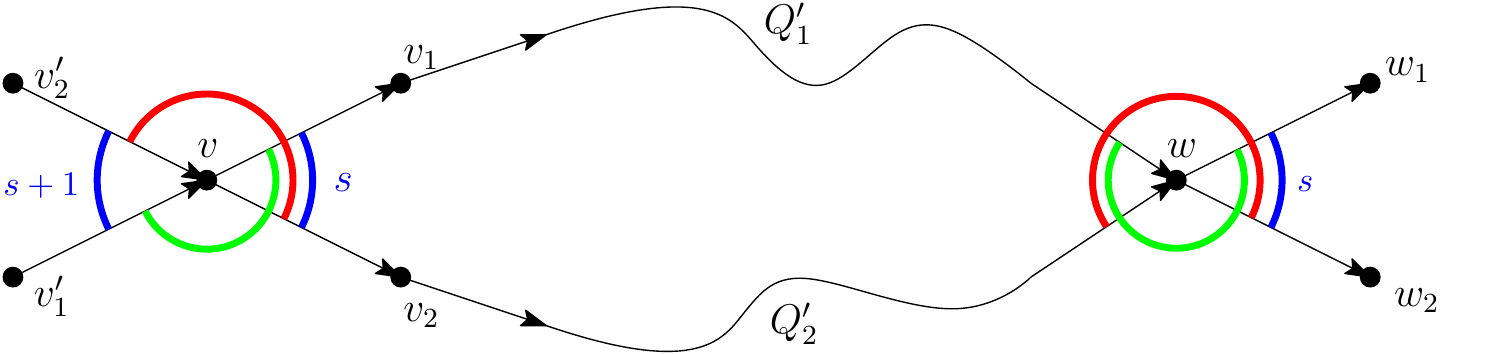}
\caption{Illustration of the proof of Lemma~\ref{lemma:path_properties}~(\ref{item:no_intersection}).
  The outgoing edges in the green (red) angles are the edges pointing from $Q_1'$ to the right (from $Q_2'$ to the left).
  In the blue angles there are~$s$ or~$s+1$ outgoing edges.}
\label{fig:paths_no_crossing}
\end{figure}

Now we are able to prove the main result of this subsection.

\begin{theorem} \label{thm:bij_fcf_alphaorien} 
  The canonical map from
  five color forests to $\alpha_5$-orientations is a bijection.
\end{theorem}

\begin{proof}
  Let~$\mathcal{F}$ be the set of five color forests of~$G$
  and~$\mathcal{A}$ the set of $\alpha_5$-orientations of~$\sG$.
  Further let~$\chi:\mathcal{F}\rightarrow\mathcal{A}$ be the
  canonical map and~$\psi:\mathcal{A}\rightarrow\mathcal{F}$ the map
  that keeps the orientation of the edges as in the
  $\alpha_5$-orientation and colors every edge~$e$ in the color of the
  end vertex of the paths in~$\mathcal{P}(e)$.

  \begin{myclaim} \label{claim:psi_well_defined} The
    map~$\psi:\mathcal{A}\rightarrow\mathcal{F}$ is well-defined.
\end{myclaim}

\begin{claimproof}
  Lemma~\ref{lemma:path_properties} (\ref{item:no_cycle}) and
  (\ref{item:same_end_vertex}) show that~$\psi$ is a well-defined
  coloring of the edges of~$G$.  It remains to show that this coloring
  fulfills the properties of a five color forest.

  Property~(F\ref{item:outer_edges}) is clear from the construction.

  Now consider property~(F\ref{item:inner_vertex_blocks}).  Because of
  Lemma~\ref{lemma:path_properties} (\ref{item:no_intersection}) the
  circular order of the colors of the outgoing edges of an inner
  vertex has to coincide with the order of the colors of the outer
  vertices.  That the incoming edges of color~$i$ are opposite of the
  outgoing edge of the same color~$i$, follows from the construction
  of the paths.

  For showing property~(F\ref{item:no_three_empty}) assume that there
  is an inner vertex~$v$ with no outgoing edges of colors~$i-2$
  and~$i+2$ for some~$i$.  In the $\alpha_5$-orientation these missing
  outgoing edges correspond to edges ending in stack vertices.  In the
  interval between these two outgoing edges there has to be at least
  one edge~$e=vw$ with a normal vertex~$w$.  And because of
  property~(F\ref{item:inner_vertex_blocks}) this edge can only be an
  incoming edge and the color is~$i$.
\end{claimproof}

\begin{myclaim} \label{claim:psi_injective}
The function~$\psi:\mathcal{A}\rightarrow\mathcal{F}$ is injective.
\end{myclaim}

\begin{claimproof}
  We show that we can recover the edges of $\sG$ from the five
  color forest on $G$.

  The orientation at normal vertices can directly be read of from the
  five color forest. Lemma~\ref{lem:missing_edge} implies that the orientation at
  stack edges is also prescribed by  the five color forest.
\end{claimproof}

\begin{myclaim} \label{claim:chi_injective}
The function~$\chi:\mathcal{F}\rightarrow\mathcal{A}$ is injective.
\end{myclaim}

\begin{claimproof}
  We show that we can recover the coloring of the edges of a five
  color forest from the orientations of the edges, i.e., from the
  $\alpha_5$-orientation in the image of $\chi$.  

  Clearly, the colors of the edges incident to the outer vertices are
  known. Moreover, because of
  property~(F\ref{item:inner_vertex_blocks}) the knowledge of the
  color of a normal edge incident to an inner normal vertex~$v$
  implies the knowledge of the colors of all edges incident to~$v$.
  Since~$G$ is connected, this implies that the colors of all edges
  are unique and known.
\end{claimproof}

Since~$\mathcal{A}$ and~$\mathcal{F}$ are finite sets, and
$\chi\circ\psi$ is the identity map on $\alpha_5$-orientations, we
obtain from Claims~\ref{claim:psi_injective}~and~\ref{claim:chi_injective} that~$\psi$
and~$\chi$ are inverse bijections.
\end{proof}

\subsection{The distributive lattice of five color forests}

It has been shown in~\cite{felsner2004lattice} that the set of all
$\alpha$-orientations of a planar graph carries the structure of a
distributive lattice.  We need some definitions to be able to describe
the cover relation of this lattice.

\begin{definition}
  A \emph{chordal path} of a simple cycle $C$ is a directed path
  consisting of edges inside $C$ whose first and last vertex are
  vertices of $C$.  These two vertices are allowed to coincide.
\end{definition}

\begin{definition}
  A simple cycle $C$ is an \emph{essential cycle} if there is an
  $\alpha$-orientation $X$ such that $C$ is a directed cycle in $X$
  and has no chordal path in $X$.
\end{definition}

\begin{theorem}[\cite{felsner2004lattice}]
  The following relation on the set of all $\alpha$-orientations of a
  planar graph is the cover relation of a distributive lattice: An
  $\alpha$-orientation $X$ covers an $\alpha$-orientation~$Y$ if and
  only if $X$ can be obtained from $Y$ by the reorientation of a
  counterclockwise oriented essential cycle in $Y$.
\end{theorem}

The reorientation of a counterclockwise (clockwise) oriented essential
cycle is called a \emph{flip (flop)}.  The following theorem gives a
full characterization of the flip operation in the lattice of five
color forests.

\begin{theorem}
  The set of all $\alpha_5$-orientations on $\sG$ carries the structure of a
  distributive lattice.  The flip operation in this lattice is the
  reorientation of a counterclockwise oriented facial cycle.
\end{theorem}

\begin{proof}
Let~$C$ be an essential cycle in~$\sG$.
Then there exists an $\alpha_5$-orientation~$X$ such that~$C$ is a directed cycle in~$X$
and has no chordal path in~$X$.
Suppose that~$C$ is not facial.

\begin{myclaim} \label{claim:no_edge_pointing_into}
There is no edge pointing into the interior of $C$.
\end{myclaim}

\begin{claimproof}
  Assume that there is an edge~$e$ pointing from a normal vertex into
  the interior of~$C$.  Let~$P \in \mathcal{P}(e)$ be a directed path
  starting with the edge~$e$ and ending in an outer vertex of~$\sG$.
  Then~$P$ has to cross~$C$ at some point and the subpath of~$P$ that
  ends at the first crossing vertex with~$C$ is a chordal path of~$C$,
  contradicting that~$C$ is essential.

  If~$e=vw$ is a stack edge, then~$v$ is a stack vertex and~$w$ a
  normal vertex.
  If~$w$ is on~$C$, the edge~$e$ is a chord of~$C$.  
  Otherwise take any outgoing edge~$e'$ of~$w$, then a
  path~$P\in\mathcal{P}(e')$ has to cross~$C$. Together with~$e$ this
  yields a chordal path.
\end{claimproof}

\begin{myclaim} \label{claim:contains_stack_vertex} 
  The cycle $C$ contains at least one stack vertex.
\end{myclaim}

\begin{claimproof}
  Assume that $C$ contains only normal vertices.  Then according to
  Lemma~\ref{lem:number_edges_into_cycle} there are exactly $2\ell(C)-5
  \neq 0$ edges pointing into the interior of $C$, in contradiction to
  Claim~\ref{claim:no_edge_pointing_into}.
\end{claimproof}

Now let~$v$ be a stack vertex on~$C$.  Let~$w_1$ be the predecessor
and~$w_2$ be the successor of~$v$ on~$C$. We know that the other outgoing 
edge of~$v$ has to point to the outside of~$C$. Now, unless~$C$ is a facial cycle the edge~$w_1w_2$ is an inner chord of~$C$. 
In either orientation the edge forms a chordal path, hence,~$C$ is
not essential. 
\end{proof}

Figure~\ref{fig:flip} shows the effect of a flip in terms of 
contacts of pentagons, the effect on the five color forest can 
be read from the figure.

\calc_figscale{30}%
\begin{figure}[t]
    \centerline{\input{\fpath/flip.pstex_t}}
    \caption{\label{fig:flip}}
    \end{figure}%
VC
{A flip of a facial cycle in an $\alpha_5$-orientation and its 
effect on contacts of pentagons. The red segment contributes to
another pentagon.}%

\section{The Algorithm}

In this section we will propose an algorithm to compute a regular
pentagon contact representation of a given graph~$G$.

We will propose a system of linear equations
related to a given five color forest~$F$ of~$G$. 
If the five color forest is induced by a regular pentagon
representation, the solution of the system allows to compute
coordinates for the corners of the pentagons in this representation.
Otherwise the solution of the system will have negative variables.

We start by describing how to obtain the skeleton graph~$\skel$
of the contact representation from the given five color forest~$F$.
We start with a crossing-free straight-line
drawing of~$G$. Add a subdivision vertex on each edge of~$G$.
Moreover, for each inner vertex~$v$ draw an edge ending at a
new vertex inside each face with a missing outgoing edge
of~$v$. Then connect all the new adjacent vertices of~$v$ in the
cyclic order given by the drawing. We call the resulting polygon
the \emph{abstract pentagon} of~$v$ (note that this polygon can have more
than five corners). Since, due to Lemma~\ref{lem:missing_edge}, in each face of~$G$
which is incident to at most one outer vertex there is exactly
one missing outgoing edge, these faces are represented by quadrilaterals in~$\skel$.
We call these quadrilaterals \emph{abstract facial quadrilaterals}.

We color the edges of~$\skel$ according to the following
rules: If the edge is part of the abstract pentagon of the inner vertex~$v$ and lies
in the interval between the outgoing edges of~$v$ of colors~$c$
and~$c+1$, it gets the color~$c-2$. The edges being part of the
abstract pentagon of the outer vertex~$a_i$ get color~$i$.
See Fig.~\ref{fig:skeleton} (left) for an example. The colors of the edges of~$\skel$
correspond to the required slopes of these edges in the following way:
We take a regular pentagon~$B$ with horizontal side at the top and
color its sides in the colors~$1,\dotsc,5$ in clockwise order, starting
with color~$1$ at the top side. Then a crossing-free straight line drawing
of~$\skel$ is a regular pentagon contact representation of~$G$ with
induced five color forest~$F$ if and only if each edge~$e$ has the same
slope as the side of~$B$ that has the same color and all abstract pentagons
are regular pentagons, i.e., have five equal side lengths. See Fig.~\ref{fig:skeleton} (right).

\begin{figure}[t]

\centering

\tikzstyle{vertex}=[circle,fill,scale=0.5]
\tikzstyle{out edge}=[-latex',thick]
\tikzstyle{in edge}=[latex'-]

\tikzstyle{normal vertex}=[circle,fill,scale=0.5,color=gray!30]
\tikzstyle{stack vertex}=[circle,draw,scale=0.5,color=gray!30]
\tikzstyle{subdivision vertex}=[rectangle,fill,scale=.3]
\tikzstyle{undirected edge}=[color=gray!30,thin]
\tikzstyle{colored undirected edge}=[thick]
\tikzstyle{edge}=[-{Latex[scale length=1.8,scale width=1]},thin]
\tikzstyle{colored edge}=[thick,-{Latex[scale length=1.8,scale width=1.2]}]
\tikzstyle{skeleton edge}=[thick]

\tikzstyle{pcr}=[very thick,color=gray!70]
\tikzstyle{pcr fill}=[fill=gray!30]

\tikzstyle{pos edge}=[thick,color=green!80]
\tikzstyle{neg edge}=[thick,color=red!80]
\tikzstyle{pos vertex}=[circle,fill,scale=0.4,color=green!80]
\tikzstyle{sign separating edge}=[-{Latex[scale length=1.8,scale width=1]}, color=blue]

\begin{tikzpicture}[scale=4]

\def\A{0.4546294640647149}
\def\B{0.2707279827787506}
\def\C{0.1507367088399636}
\def\D{0.2356059251438577}
\def\E{0.1775675733633562}
\def\aa{0.2809764610791428}
\def\ab{0.4546294640647149}
\def\ba{0.2809764610791428}
\def\bb{0.4546294640647149}
\def\ca{0.1097427956383953}
\def\cb{0.1775675733633562}
\def\da{0.09316040941539433}
\def\db{0.1507367088399636}
\def\ea{0.167319095062964}
\def\eb{0.2707279827787506}
\def\fa{0.1136573660161787}
\def\fb{0.2707279827787506}
\def\fc{0.4546294640647149}
\def\fd{0.5517044438578933}
\def\ga{0.04853748989658916}
\def\gb{0.1570706167625718}
\def\gc{0.2356059251438577}
\def\gd{0.3026830864523393}
\def\ha{0.1775675733633562}
\def\hb{0.171233665440748}
\def\hc{0.4585440344424989}
\def\hd{0.4546294640647149}
\def\ia{0.006333907922608149}
\def\ib{0.1416978993119834}
\def\ic{0.1519463776123757}
\def\id{0.2356059251438577}
\def\ja{0.07415868564756978}
\def\jb{0.1507367088399636}
\def\jc{0.2707279827787506}
\def\jd{0.3180558039029276}
\def\ka{0.09828464856559038}
\def\kb{0.07657802319239386}
\def\kc{0.2356059251438577}
\def\kd{0.2221904928821613}
\def\la{0.1507367088399636}
\def\lb{0.05245206027437295}
\def\lc{0.296349178529731}
\def\ld{0.2356059251438577}
\def\ma{0.03586967405137269}
\def\mb{0.1775675733633562}
\def\mc{0.2356059251438577}
\def\md{0.3231800430531234}

            
\coordinate (At) at (0: \aa);
            
        
\coordinate (Bt) at (0: \aa+\fd);       
        
        
\coordinate (Ctr) at ($(0: \aa+\fd+\ea)+(-72:\ea+\jd)$);
        
        
\coordinate (Dbr) at ($(0: \aa+\fd+\ea)+(-72:\ea+\jd+\da)+(-144:\da+\lc)$);
        
        
\coordinate (Ebl) at ($(-108: \aa+\ba)+(-36:\ba+\hc)$);
        
        
\node[normal vertex] (Ac) at ($(At)+(-90:0.8506*\A)$) {};
\node[normal vertex] (Bc) at ($(Bt)+(-90:0.8506*\B)$) {};
\node[normal vertex] (Cc) at ($(Ctr)+(-162:0.8506*\C)$) {};
\node[normal vertex] (Dc) at ($(Dbr)+(-234:0.8506*\D)$) {};
\node[normal vertex] (Ec) at ($(Ebl)+(-306:0.8506*\E)$) {};

\coordinate (c1) at (0,0);
\coordinate (c2) at (0:\aa+\fd+\ea);
\coordinate (c3) at ($(c2)+(-72:\ea+\jd+\da)$);
\coordinate (c4) at ($(c3)+(-144:\da+\lc+\md+\ca)$);
\coordinate (c5) at ($(c4)+(-216:\ca+\hc+\ba)$);

\node[normal vertex] (a1) at ($($(c1)!0.5!(c2)$)+(90:0.1)$) {};
\node[normal vertex] (a2) at ($($(c2)!0.5!(c3)$)+(18:0.1)$) {};
\node[normal vertex] (a3) at ($($(c3)!0.5!(c4)$)+(-54:0.1)$) {};
\node[normal vertex] (a4) at ($($(c4)!0.5!(c5)$)+(-126:0.1)$) {};
\node[normal vertex] (a5) at ($($(c5)!0.5!(c1)$)+(-198:0.1)$) {};

\node[stack vertex] (sf) at (.54,-.1) {};
\node[stack vertex] (sg) at (.65,-.41) {};
\node[stack vertex] (sh) at (.3,-.75) {};
\node[stack vertex] (si) at (.56,-.68) {};
\node[stack vertex] (sj) at (1.04,-.33) {};
\node[stack vertex] (sk) at (.87,-.46) {};
\node[stack vertex] (sl) at (.92,-.67) {};
\node[stack vertex] (sm) at (.71,-.8) {};

\node[subdivision vertex] (sd_c1) at  ($(c1)+(-234:.05)$) {};
\node[subdivision vertex] (sd_c2) at ($(c2)+(54:.05)$)  {};
\node[subdivision vertex] (sd_c3) at  ($(c3)+(-18:.05)$) {};
\node[subdivision vertex] (sd_c4) at  ($(c4)+(-90:.05)$) {};
\node[subdivision vertex] (sd_c5) at  ($(c5)+(-162:.05)$) {};

\draw[undirected edge] (a1) -- (sd_c2) -- (a2);
\draw[undirected edge] (a2) -- (sd_c3) -- (a3);
\draw[undirected edge] (a3) -- (sd_c4) -- (a4);
\draw[undirected edge] (a4) -- (sd_c5) -- (a5);
\draw[undirected edge] (a5) -- (sd_c1) -- (a1);

\draw[edge,color1verylight] (Ac) -- (a1);
\draw[edge,color5verylight] (Ac) -- (a5);
\draw[edge,color4verylight] (Ac) -- (a4);
\draw[edge,color3verylight] (Ac) -- (Ec);
\draw[edge,color2verylight] (Ac) -- (Bc);
\draw[edge,color1verylight] (Bc) -- (a1);
\draw[edge,color2verylight] (Bc) -- (a2);
\draw[edge,color3verylight] (Bc) -- (Cc);
\draw[edge,color2verylight] (Cc) -- (a2);
\draw[edge,color3verylight] (Cc) -- (a3);
\draw[edge,color3verylight] (Dc) -- (a3);
\draw[edge,color5verylight] (Dc) -- (Ac);
\draw[edge,color1verylight] (Dc) -- (Bc);
\draw[edge,color2verylight] (Dc) -- (Cc);
\draw[edge,color4verylight] (Dc) -- (Ec);
\draw[edge,color3verylight] (Ec) -- (a3);
\draw[edge,color4verylight] (Ec) -- (a4);

\foreach \p/\q in {Ac/a1, Ac/a5, Ac/a4, Ac/Ec, Ac/Bc,
    Bc/a1, Bc/a2, Bc/Cc,
    Cc/a2, Cc/a3,
    Dc/a3, Dc/Ac, Dc/Bc, Dc/Cc, Dc/Ec,
    Ec/a3, Ec/a4}
  \node[subdivision vertex] (sd_\p_\q) at  ($(\p)!.5!(\q)$) {};

\draw[edge,color5verylight] (Bc) -- (sf);
\draw[edge,color4verylight] (Bc) -- (sg);
\draw[edge,color5verylight] (Ec) -- (sh);
\draw[edge,color1verylight] (Ec) -- (si);
\draw[edge,color1verylight] (Cc) -- (sj);
\draw[edge,color5verylight] (Cc) -- (sk);
\draw[edge,color4verylight] (Cc) -- (sl);
\draw[edge,color2verylight] (Ec) -- (sm);

\foreach \p/\q in {Bc/sf, Bc/sg, Ec/sh, Ec/si, Cc/sj, Cc/sk, Cc/sl, Ec/sm}
  \node[subdivision vertex] (sd_\p_\q) at  ($(\p)!.5!(\q)$) {};

\foreach \p/\q/\r/\c in {Ac/a1/a5/3, Ac/a5/a4/2, Ac/a4/Ec/1, Ac/Bc/a1/4,
    Bc/a1/a2/4, Bc/a2/Cc/5, Bc/sf/a1/3,
    Cc/a2/a3/5, Cc/a3/sl/1, Cc/sj/a2/4,
    Dc/a3/Ec/1, Dc/Ec/Ac/2, Dc/Ac/Bc/3, Dc/Bc/Cc/4, Dc/Cc/a3/5,
    Ec/a3/a4/1, Ec/a4/sh/2, Ec/sm/a3/5}
  \draw[skeleton edge, color\c] (sd_\p_\q) -- (sd_\p_\r);

\foreach \p/\q/\r/\c in {Ac/Ec/Dc/5, Ac/Bc/Dc/5,
    Bc/Cc/Dc/1, Bc/sg/Dc/1, Bc/sg/Ac/2, Bc/sf/Ac/2,
    Cc/sl/Dc/2, Cc/sk/Dc/2, Cc/sk/Bc/3, Cc/sj/Bc/3,
    Ec/sh/Ac/3, Ec/si/Ac/3, Ec/si/Dc/4, Ec/sm/Dc/4}
  \draw[skeleton edge, color\c] (sd_\p_\q) -- (sd_\r_\p);

\draw[skeleton edge, color1] (sd_c1) -- (sd_Ac_a1);
\draw[skeleton edge, color1] (sd_Ac_a1) -- (sd_Bc_a1);
\draw[skeleton edge, color1] (sd_Bc_a1) -- (sd_c2);

\draw[skeleton edge, color2] (sd_c2) -- (sd_Bc_a2);
\draw[skeleton edge, color2] (sd_Bc_a2) -- (sd_Cc_a2);
\draw[skeleton edge, color2] (sd_Cc_a2) -- (sd_c3);

\draw[skeleton edge, color3] (sd_c3) -- (sd_Cc_a3);
\draw[skeleton edge, color3] (sd_Cc_a3) -- (sd_Dc_a3);
\draw[skeleton edge, color3] (sd_Dc_a3) -- (sd_Ec_a3);
\draw[skeleton edge, color3] (sd_Ec_a3) -- (sd_c4);

\draw[skeleton edge, color4] (sd_c4) -- (sd_Ec_a4);
\draw[skeleton edge, color4] (sd_Ec_a4) -- (sd_Ac_a4);
\draw[skeleton edge, color4] (sd_Ac_a4) -- (sd_c5);

\draw[skeleton edge, color5] (sd_c5) -- (sd_Ac_a5);
\draw[skeleton edge, color5] (sd_Ac_a5) -- (sd_c1);
        
\end{tikzpicture}
\qquad
\begin{tikzpicture}[scale=4]

\def\A{0.4546294640647149}
\def\B{0.2707279827787506}
\def\C{0.1507367088399636}
\def\D{0.2356059251438577}
\def\E{0.1775675733633562}
\def\aa{0.2809764610791428}
\def\ab{0.4546294640647149}
\def\ba{0.2809764610791428}
\def\bb{0.4546294640647149}
\def\ca{0.1097427956383953}
\def\cb{0.1775675733633562}
\def\da{0.09316040941539433}
\def\db{0.1507367088399636}
\def\ea{0.167319095062964}
\def\eb{0.2707279827787506}
\def\fa{0.1136573660161787}
\def\fb{0.2707279827787506}
\def\fc{0.4546294640647149}
\def\fd{0.5517044438578933}
\def\ga{0.04853748989658916}
\def\gb{0.1570706167625718}
\def\gc{0.2356059251438577}
\def\gd{0.3026830864523393}
\def\ha{0.1775675733633562}
\def\hb{0.171233665440748}
\def\hc{0.4585440344424989}
\def\hd{0.4546294640647149}
\def\ia{0.006333907922608149}
\def\ib{0.1416978993119834}
\def\ic{0.1519463776123757}
\def\id{0.2356059251438577}
\def\ja{0.07415868564756978}
\def\jb{0.1507367088399636}
\def\jc{0.2707279827787506}
\def\jd{0.3180558039029276}
\def\ka{0.09828464856559038}
\def\kb{0.07657802319239386}
\def\kc{0.2356059251438577}
\def\kd{0.2221904928821613}
\def\la{0.1507367088399636}
\def\lb{0.05245206027437295}
\def\lc{0.296349178529731}
\def\ld{0.2356059251438577}
\def\ma{0.03586967405137269}
\def\mb{0.1775675733633562}
\def\mc{0.2356059251438577}
\def\md{0.3231800430531234}

\coordinate (o1) at (0,0);
\coordinate (o2) at ($(o1)+(0: \aa+\fd+\ea)$);
\coordinate (o3) at ($(o2)+(-72: \ea+\jd+\da)$);
\coordinate (o4) at ($(o3)+(-144:\da+\lc+\md+\ca)$);
\coordinate (o5) at ($(o4)+(-216:\ca+\hc+\ba)$);

\draw[pcr,color1] (o1) -- (o2);
\draw[pcr,color2] (o2) -- (o3);
\draw[pcr,color3] (o3) -- (o4);
\draw[pcr,color4] (o4) -- (o5);
\draw[pcr,color5] (o5) -- (o1);
            
\coordinate (At) at (0: \aa);
\coordinate (A1) at ($(At)+(-144:\A)$);
\coordinate (A2) at ($(A1)+(-72:\A)$);
\coordinate (A3) at ($(A2)+(0:\A)$);
\coordinate (A4) at ($(A3)+(72:\A)$);
            
\draw[pcr fill] (At)
        -- ++(-144:\A)
        -- ++(-72:\A)
        -- ++(0:\A)
        -- ++(72:\A)
        -- ++(144:\A);
        
\draw[pcr, color3] (At) -- (A1);
\draw[pcr, color2] (A1) -- (A2);
\draw[pcr, color1] (A2) -- (A3);
\draw[pcr, color5] (A3) -- (A4);
\draw[pcr, color4] (A4) -- (At);
        
\coordinate (Bt) at (0: \aa+\fd);
\coordinate (B1) at ($(Bt)+(-144:\B)$);
\coordinate (B2) at ($(B1)+(-72:\B)$);
\coordinate (B3) at ($(B2)+(0:\B)$);
\coordinate (B4) at ($(B3)+(72:\B)$);       
        
\draw[pcr fill] (Bt)
        -- ++(-144:\B)
        -- ++(-72:\B)
        -- ++(0:\B)
        -- ++(72:\B)
        -- ++(144:\B);
        
\draw[pcr, color3] (Bt) -- (B1);
\draw[pcr, color2] (B1) -- (B2);
\draw[pcr, color1] (B2) -- (B3);
\draw[pcr, color5] (B3) -- (B4);
\draw[pcr, color4] (B4) -- (Bt);
        
\coordinate (Ctr) at ($(0: \aa+\fd+\ea)+(-72:\ea+\jd)$);
\coordinate (C1) at ($(Ctr)+(-216:\C)$);
\coordinate (C2) at ($(C1)+(-144:\C)$);
\coordinate (C3) at ($(C2)+(-72:\C)$);
\coordinate (C4) at ($(C3)+(0:\C)$);   
        
\draw[pcr fill] (Ctr)
        -- ++(-216:\C)
        -- ++(-144:\C)
        -- ++(-72:\C)
        -- ++(0:\C)
        -- ++(72:\C);
        
\draw[pcr, color4] (Ctr) -- (C1);
\draw[pcr, color3] (C1) -- (C2);
\draw[pcr, color2] (C2) -- (C3);
\draw[pcr, color1] (C3) -- (C4);
\draw[pcr, color5] (C4) -- (Ctr);
        
\coordinate (Dbr) at ($(0: \aa+\fd+\ea)+(-72:\ea+\jd+\da)+(-144:\da+\lc)$);
\coordinate (D1) at ($(Dbr)+(-288:\D)$);
\coordinate (D2) at ($(D1)+(-216:\D)$);
\coordinate (D3) at ($(D2)+(-144:\D)$);
\coordinate (D4) at ($(D3)+(-72:\D)$);  
        
\draw[pcr fill] (Dbr)
        -- ++(-288:\D)
        -- ++(-216:\D)
        -- ++(-144:\D)
        -- ++(-72:\D)
        -- ++(0:\D);
        
\draw[pcr, color5] (Dbr) -- (D1);
\draw[pcr, color4] (D1) -- (D2);
\draw[pcr, color3] (D2) -- (D3);
\draw[pcr, color2] (D3) -- (D4);
\draw[pcr, color1] (D4) -- (Dbr);
        
\coordinate (Ebl) at ($(-108: \aa+\ba)+(-36:\ba+\hc)$);
\coordinate (E1) at ($(Ebl)+(0:\E)$);
\coordinate (E2) at ($(E1)+(-288:\E)$);
\coordinate (E3) at ($(E2)+(-216:\E)$);
\coordinate (E4) at ($(E3)+(-144:\E)$);  
        
\draw[pcr fill] (Ebl)
        -- ++(0:\E)
        -- ++(-288:\E)
        -- ++(-216:\E)
        -- ++(-144:\E)
        -- ++(-72:\E);
        
\draw[pcr, color1] (Ebl) -- (E1);
\draw[pcr, color5] (E1) -- (E2);
\draw[pcr, color4] (E2) -- (E3);
\draw[pcr, color3] (E3) -- (E4);
\draw[pcr, color2] (E4) -- (Ebl);
        
\end{tikzpicture}

\caption{Left: The skeleton graph corresponding to the five color forest in the background.
         Right: A realization of this skeleton graph as a regular pentagon contact representation.}

\label{fig:skeleton}

\end{figure}
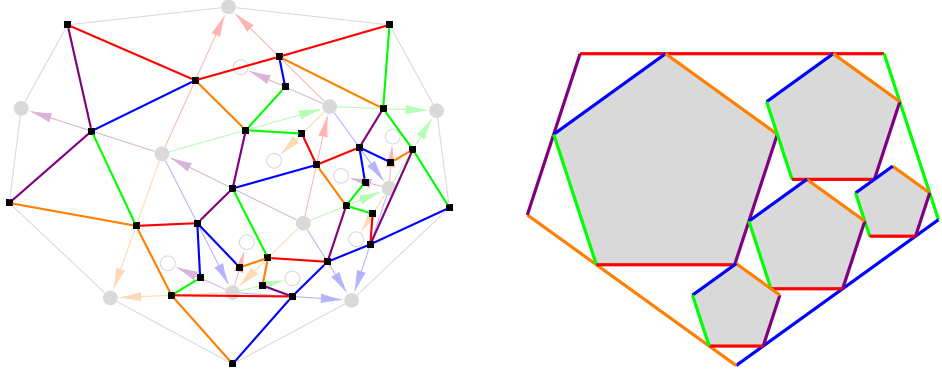

The purpose of the system of linear equations is to find edge lengths
for the edges of~$\skel$. Therefore we have a variable~$x_v$ for each
inner vertex~$v$ of~$G$ representing the side length of the corresponding
pentagon and a variable for edge of~$\skel$ representing its length.
The second type of variables can also be defined in the following way:
Every inner face~$f$ of~$G$
gets four variables~$x_f^{(1)},\dotsc,x_f^{(4)}$ representing the
segment lengths of the corresponding quadrilateral in clockwise order
where the concave corner is located between the edges corresponding
to~$x_f^{(1)}$ and~$x_f^{(2)}$ (see Fig.~\ref{fig:var_faces} (left)).  For the
five inner faces which are incident to two outer vertices of $G$ we
add the equation~$x_f^{(1)}=0$ since these faces are represented by
triangles, not by quadrilaterals.

With every inner vertex~$v$ we associate five equations, one for
each side. Each of these equations states that the side length $x_v$
is equal to the sum of the lengths of the boundary segments of faces
incident to the side. More formally, for~${i=1,\dotsc,5}$,
let~$\delta_i(v)$ denote the set of faces of~$G$ incident to~$v$ in
the interval between the outgoing edges of colors~$i+2$ and~$i-2$.
Then we can write these five equations as~${x_v=\sum_{f\in
    \delta_i(v)}x_f^{(j_{v,f,i})}}$ with
the~$j_{v,f,i}\in\{1,\dotsc,4\}$ appropriately chosen.
The following lemma gives two more equations for every inner face.

\begin{lemma} \label{lemma:face_equations}
Let~$f$ be an inner face of~$G$. If the variables~${ x_f^{(1)}, \dotsc, x_f^{(4)} }$
come from a regular pentagon contact representation of~$G$, they fulfill the
following equations:
\begin{equation*}
x_f^{(3)} = x_f^{(1)} + \phi x_f^{(2)} \enspace ,\qquad \enspace 
x_f^{(4)} = \phi x_f^{(1)} + x_f^{(2)} \enspace.
\end{equation*}
Here~$\phi =\frac{1 + \sqrt{5}}{2}$ denotes the golden ratio.
\end{lemma}

\begin{proof}
For geometric reasons the three convex corners of the facial quadrilateral
corresponding to~$f$ are exactly~$\frac{\pi}{5}$.
Now we cut the quadrilateral along an extension of the edge corresponding to~$x_f^{(1)}$
into two triangles. We denote the length of the cut by~$c$.
The edge corresponding to~$x_f^{(3)}$ is cut into two parts.
We denote the lengths of these parts by~$a$ and~$b$ in clockwise order (see Fig.~\ref{fig:var_faces} (middle)).
The two resulting triangles have constant inner angles.
Thus there are constants~$\alpha,\beta,\gamma \in \mathbb{R}$ such that
\[ c = \gamma x_f^{(2)} \enspace , \enspace b = \beta x_f^{(2)} \enspace , \enspace a = \alpha (x_f^{(1)} + c) \enspace . \]
Hence, we have~$x_f^{(3)} = a+b = \alpha x_f^{(1)} + (\alpha\gamma + \beta) x_f^{(2)}$.
To figure out the constants~$\alpha$ and~$\alpha\gamma + \beta$ let us consider the
special cases that~$x_f^{(1)} = 0$ or~$x_f^{(2)} = 0$ (see Fig.~\ref{fig:var_faces} (right)). In the first case we have~$x_f^{(3)} = 2 \cos (\pi/5) x_f^{(2)} = \phi x_f^{(2)}$,
in the second case~$x_f^{(3)} = x_f^{(1)}$ since the corresponding edges are the legs of an isosceles triangle.
Therefore we have~$x_f^{(3)} = x_f^{(1)} + \phi x_f^{(2)}$.
The second equation can be obtained symmetrically.
\end{proof}

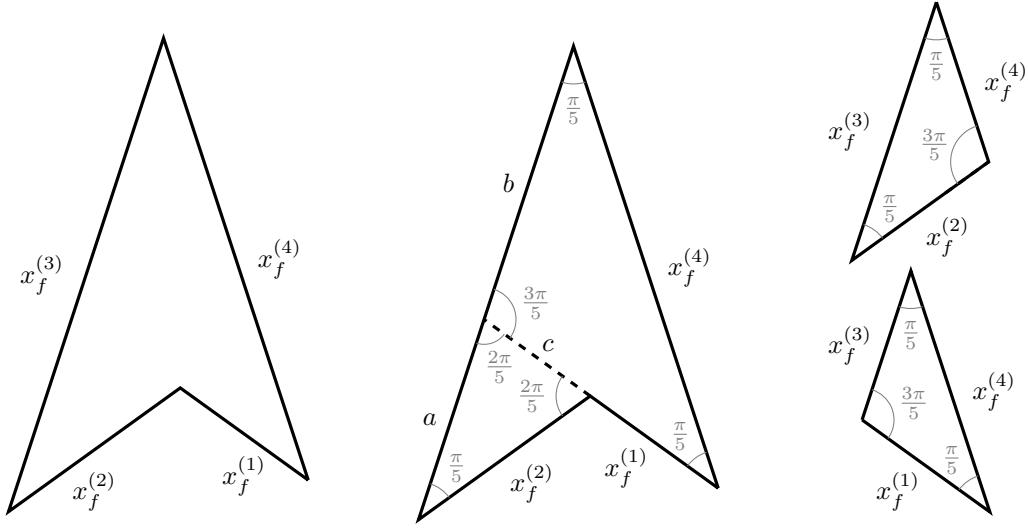
\begin{figure}

\centering

\tikzstyle{vertex}=[circle,fill,scale=0.5]
\tikzstyle{out edge}=[-latex',thick]
\tikzstyle{in edge}=[latex'-]
\tikzstyle{angle color}=[color=gray]

\hfill
\begin{minipage}[b]{.25\linewidth}
\begin{tikzpicture}[scale=2.2]

\coordinate (A) at (0,0);
\coordinate (B) at (-108:3);
\coordinate (C) at (-72:2.8);
\coordinate (hB) at ($(B)+(36:3)$);
\coordinate (hC) at ($(C)+(144:3)$);
\coordinate (D) at (intersection of B--hB and C--hC);

\draw[very thick] (C) -- (A) -- (B) -- (D) -- (C);

\node at ($(C)!0.5!(D)$) [label={below}:\footnotesize $x_f^{(1)}$] {};
\node at ($(B)!0.5!(D)$) [label={below}:\footnotesize $x_f^{(2)}$] {};
\node at ($(B)!0.5!(A)$) [label={left}:\footnotesize $x_f^{(3)}$] {};
\node at ($(C)!0.5!(A)$) [label={right}:\footnotesize $x_f^{(4)}$] {};


\end{tikzpicture}
\end{minipage}
\hfill
\begin{minipage}[b]{.25\linewidth}
\begin{tikzpicture}[scale=2.2]

\coordinate (A) at (0,0);
\coordinate (B) at (-108:3);
\coordinate (C) at (-72:2.8);
\coordinate (hB) at ($(B)+(36:3)$);
\coordinate (hC) at ($(C)+(144:3)$);
\coordinate (D) at (intersection of B--hB and C--hC);

\coordinate (E) at (intersection of C--hC and A--B);

\draw[very thick] (C) -- (A) -- (B) -- (D) -- (C);
\draw[very thick, dashed] (D) -- (E);

\node at ($(C)!0.5!(D)$) [label={below left, label distance=-.3cm}:\footnotesize $x_f^{(1)}$] {};
\node at ($(B)!0.5!(D)$) [label={below right, label distance=-.3cm}:\footnotesize $x_f^{(2)}$] {};
\node at ($(C)!0.5!(A)$) [label={right}:\footnotesize $x_f^{(4)}$] {};

\node at ($(B)!0.5!(E)$) [label={left, label distance=-.1cm}:\footnotesize $a$] {};
\node at ($(E)!0.5!(A)$) [label={left, label distance=-.1cm}:\footnotesize $b$] {};
\node at ($(D)!0.4!(E)$) [label={above, label distance=-.1cm}:\footnotesize $c$] {};

\pic [draw, angle color, "\footnotesize $\frac{\pi}{5}$", angle eccentricity=1.7, angle radius=.5cm] {angle = B--A--C};
\pic [draw, angle color, "\footnotesize $\frac{\pi}{5}$", angle eccentricity=1.7, angle radius=.5cm] {angle = D--B--A};
\pic [draw, angle color, "\footnotesize $\frac{\pi}{5}$", angle eccentricity=1.7, angle radius=.5cm] {angle = A--C--D};
\pic [draw, angle color, "\footnotesize $\frac{3\pi}{5}$", angle eccentricity=1.7, angle radius=.42cm] {angle = D--E--A};
\pic [draw, angle color, "\footnotesize $\frac{2\pi}{5}$", angle eccentricity=2, angle radius=.33cm] {angle = B--E--D};
\pic [draw, angle color, "\footnotesize $\frac{2\pi}{5}$", angle eccentricity=1.7, angle radius=.45cm] {angle = E--D--B};

\end{tikzpicture}
\end{minipage}
\hfill
\begin{minipage}[b]{.25\linewidth}
\begin{tikzpicture}[scale=1.2]

\coordinate (A) at (0,0);
\coordinate (B) at (-108:3);
\coordinate (C) at (-72:2.8);
\coordinate (hB) at ($(B)+(36:3)$);
\coordinate (hC) at ($(C)+(144:3)$);
\coordinate (D) at (intersection of B--hB and C--hC);

\coordinate (E) at (intersection of B--hB and A--C);

\draw[very thick] (A) -- (B) -- (E) -- (A);
\draw[very thick, dashed] (D) -- (E);

\node at ($(B)!0.5!(E)$) [label={below right, label distance=-.3cm}:\footnotesize $x_f^{(2)}$] {};
\node at ($(B)!0.5!(A)$) [label={left}:\footnotesize $x_f^{(3)}$] {};
\node at ($(E)!0.5!(A)$) [label={right}:\footnotesize $x_f^{(4)}$] {};

\pic [draw, angle color, "\footnotesize $\frac{\pi}{5}$", angle eccentricity=1.7, angle radius=.5cm] {angle = B--A--C};
\pic [draw, angle color, "\footnotesize $\frac{\pi}{5}$", angle eccentricity=1.7, angle radius=.5cm] {angle = D--B--A};
\pic [draw, angle color, "\footnotesize $\frac{3\pi}{5}$", angle eccentricity=1.5, angle radius=.5cm] {angle = A--E--B};

\end{tikzpicture}

\begin{tikzpicture}[scale=1.2]

\coordinate (A) at (0,0);
\coordinate (B) at (-108:3);
\coordinate (C) at (-72:2.8);
\coordinate (hB) at ($(B)+(36:3)$);
\coordinate (hC) at ($(C)+(144:3)$);
\coordinate (D) at (intersection of B--hB and C--hC);

\coordinate (E) at (intersection of C--hC and A--B);

\draw[very thick] (E) -- (C) -- (A) -- (E);

\node at ($(C)!0.5!(E)$) [label={below left, label distance=-.3cm}:\footnotesize $x_f^{(1)}$] {};
\node at ($(C)!0.5!(A)$) [label={right}:\footnotesize $x_f^{(4)}$] {};

\node at ($(E)!0.5!(A)$) [label={left, label distance=-.1cm}:\footnotesize $x_f^{(3)}$] {};

\pic [draw, angle color, "\footnotesize $\frac{\pi}{5}$", angle eccentricity=1.7, angle radius=.5cm] {angle = B--A--C};
\pic [draw, angle color, "\footnotesize $\frac{\pi}{5}$", angle eccentricity=1.7, angle radius=.5cm] {angle = A--C--D};
\pic [draw, angle color, "\footnotesize $\frac{3\pi}{5}$", angle eccentricity=1.7, angle radius=.42cm] {angle = D--E--A};

\end{tikzpicture}
\end{minipage}
\hfill

\caption{Left: The variables for an inner face~$f$.
  Middle: The cut as described in the proof of Lemma~\ref{lemma:face_equations}.
  Right: The two special cases with~$x_f^{(1)}=0$ and~$x_f^{(2)}=0$.}

\label{fig:var_faces}

\end{figure}

Finally, we add one more
equation to the system which implies that the sum of the
lengths of the face edges building the line segment corresponding to
the outer vertex~$a_1$ of~$G$ is exactly~$1$, i.e.,~$\sum_{f\in
  \delta_1(a_1)}x_f^{(j_{a_1,f,1})}=1$
with~${j_{a_1,f,1}\in\{1,\dotsc,4\}}$ appropriately chosen.
See Fig.~\ref{fig:equations} for an illustration of the different types of
equations.

\begin{figure}[t]

\centering

\tikzstyle{vertex}=[circle,fill,scale=0.5]
\tikzstyle{out edge}=[-latex',thick]
\tikzstyle{in edge}=[latex'-]

\tikzstyle{normal vertex}=[circle,fill,scale=0.5,color=gray!30]
\tikzstyle{stack vertex}=[circle,draw,scale=0.5,color=gray!30]
\tikzstyle{subdivision vertex}=[rectangle,fill,scale=.3]
\tikzstyle{undirected edge}=[color=gray!30,thin]
\tikzstyle{colored undirected edge}=[thick]
\tikzstyle{edge}=[-{Latex[scale length=1.8,scale width=1]},thin]
\tikzstyle{colored edge}=[thick,-{Latex[scale length=1.8,scale width=1.2]}]
\tikzstyle{skeleton edge}=[thick]

\tikzstyle{pcr}=[very thick,color=gray!70]
\tikzstyle{pcr fill}=[fill=gray!30]

\tikzstyle{pos edge}=[thick,color=green!80]
\tikzstyle{neg edge}=[thick,color=red!80]
\tikzstyle{pos vertex}=[circle,fill,scale=0.4,color=green!80]
\tikzstyle{sign separating edge}=[-{Latex[scale length=1.8,scale width=1]}, color=blue]

\begin{tikzpicture}[scale=6]

\def\A{0.4546294640647149}
\def\B{0.2707279827787506}
\def\C{0.1507367088399636}
\def\D{0.2356059251438577}
\def\E{0.1775675733633562}
\def\aa{0.2809764610791428}
\def\ab{0.4546294640647149}
\def\ba{0.2809764610791428}
\def\bb{0.4546294640647149}
\def\ca{0.1097427956383953}
\def\cb{0.1775675733633562}
\def\da{0.09316040941539433}
\def\db{0.1507367088399636}
\def\ea{0.167319095062964}
\def\eb{0.2707279827787506}
\def\fa{0.1136573660161787}
\def\fb{0.2707279827787506}
\def\fc{0.4546294640647149}
\def\fd{0.5517044438578933}
\def\ga{0.04853748989658916}
\def\gb{0.1570706167625718}
\def\gc{0.2356059251438577}
\def\gd{0.3026830864523393}
\def\ha{0.1775675733633562}
\def\hb{0.171233665440748}
\def\hc{0.4585440344424989}
\def\hd{0.4546294640647149}
\def\ia{0.006333907922608149}
\def\ib{0.1416978993119834}
\def\ic{0.1519463776123757}
\def\id{0.2356059251438577}
\def\ja{0.07415868564756978}
\def\jb{0.1507367088399636}
\def\jc{0.2707279827787506}
\def\jd{0.3180558039029276}
\def\ka{0.09828464856559038}
\def\kb{0.07657802319239386}
\def\kc{0.2356059251438577}
\def\kd{0.2221904928821613}
\def\la{0.1507367088399636}
\def\lb{0.05245206027437295}
\def\lc{0.296349178529731}
\def\ld{0.2356059251438577}
\def\ma{0.03586967405137269}
\def\mb{0.1775675733633562}
\def\mc{0.2356059251438577}
\def\md{0.3231800430531234}

\coordinate (o1) at (0,0);
\coordinate (o2) at ($(o1)+(0: \aa+\fd+\ea)$);
\coordinate (o3) at ($(o2)+(-72: \ea+\jd+\da)$);
\coordinate (o4) at ($(o3)+(-144:\da+\lc+\md+\ca)$);
\coordinate (o5) at ($(o4)+(-216:\ca+\hc+\ba)$);

\draw[pcr] (o1) -- (o2);
\draw[pcr] (o2) -- (o3);
\draw[pcr] (o3) -- (o4);
\draw[pcr] (o4) -- (o5);
\draw[pcr] (o5) -- (o1);
            
\coordinate (At) at (0: \aa);
\coordinate (A1) at ($(At)+(-144:\A)$);
\coordinate (A2) at ($(A1)+(-72:\A)$);
\coordinate (A3) at ($(A2)+(0:\A)$);
\coordinate (A4) at ($(A3)+(72:\A)$);
            
\draw[pcr fill] (At)
        -- ++(-144:\A)
        -- ++(-72:\A)
        -- ++(0:\A)
        -- ++(72:\A)
        -- ++(144:\A);
        
\draw[pcr] (At) -- (A1);
\draw[pcr] (A1) -- (A2);
\draw[pcr] (A2) -- (A3);
\draw[pcr] (A3) -- (A4);
\draw[pcr] (A4) -- (At);
        
\coordinate (Bt) at (0: \aa+\fd);
\coordinate (B1) at ($(Bt)+(-144:\B)$);
\coordinate (B2) at ($(B1)+(-72:\B)$);
\coordinate (B3) at ($(B2)+(0:\B)$);
\coordinate (B4) at ($(B3)+(72:\B)$);       
        
\draw[pcr fill] (Bt)
        -- ++(-144:\B)
        -- ++(-72:\B)
        -- ++(0:\B)
        -- ++(72:\B)
        -- ++(144:\B);
        
\draw[pcr] (Bt) -- (B1);
\draw[pcr] (B1) -- (B2);
\draw[pcr] (B2) -- (B3);
\draw[pcr] (B3) -- (B4);
\draw[pcr] (B4) -- (Bt);
        
\coordinate (Ctr) at ($(0: \aa+\fd+\ea)+(-72:\ea+\jd)$);
\coordinate (C1) at ($(Ctr)+(-216:\C)$);
\coordinate (C2) at ($(C1)+(-144:\C)$);
\coordinate (C3) at ($(C2)+(-72:\C)$);
\coordinate (C4) at ($(C3)+(0:\C)$);   
        
\draw[pcr fill] (Ctr)
        -- ++(-216:\C)
        -- ++(-144:\C)
        -- ++(-72:\C)
        -- ++(0:\C)
        -- ++(72:\C);
        
\draw[pcr] (Ctr) -- (C1);
\draw[pcr] (C1) -- (C2);
\draw[pcr] (C2) -- (C3);
\draw[pcr] (C3) -- (C4);
\draw[pcr] (C4) -- (Ctr);
        
\coordinate (Dbr) at ($(0: \aa+\fd+\ea)+(-72:\ea+\jd+\da)+(-144:\da+\lc)$);
\coordinate (D1) at ($(Dbr)+(-288:\D)$);
\coordinate (D2) at ($(D1)+(-216:\D)$);
\coordinate (D3) at ($(D2)+(-144:\D)$);
\coordinate (D4) at ($(D3)+(-72:\D)$);  
        
\draw[pcr fill] (Dbr)
        -- ++(-288:\D)
        -- ++(-216:\D)
        -- ++(-144:\D)
        -- ++(-72:\D)
        -- ++(0:\D);
        
\draw[pcr] (Dbr) -- (D1);
\draw[pcr] (D1) -- (D2);
\draw[pcr] (D2) -- (D3);
\draw[pcr] (D3) -- (D4);
\draw[pcr] (D4) -- (Dbr);
        
\coordinate (Ebl) at ($(-108: \aa+\ba)+(-36:\ba+\hc)$);
\coordinate (E1) at ($(Ebl)+(0:\E)$);
\coordinate (E2) at ($(E1)+(-288:\E)$);
\coordinate (E3) at ($(E2)+(-216:\E)$);
\coordinate (E4) at ($(E3)+(-144:\E)$);  
        
\draw[pcr fill] (Ebl)
        -- ++(0:\E)
        -- ++(-288:\E)
        -- ++(-216:\E)
        -- ++(-144:\E)
        -- ++(-72:\E);
        
\draw[pcr] (Ebl) -- (E1);
\draw[pcr] (E1) -- (E2);
\draw[pcr] (E2) -- (E3);
\draw[pcr] (E3) -- (E4);
\draw[pcr] (E4) -- (Ebl);

\node at ($1/3*(At)+1/3*(o1)+1/3*(A1)$) {\footnotesize $f_1$};
\node at ($1/3*(At)+1/3*(Bt)+1/3*(A4)$) {\footnotesize $f_2$};
\node at ($1/3*(Bt)+1/3*(o2)+1/3*(B4)$) {\footnotesize $f_3$};
\node at ($1/3*(A4)+1/3*(D2)+1/3*(D3)$) {\footnotesize $f_4$};

\draw[pcr,color=red] (o1) -- (At);
\draw[pcr,color=green] (At) -- (Bt);
\draw[pcr,color=blue] (Bt) -- (o2);
\draw[pcr,color=violet] (At) -- (A4);
\draw[pcr,color=orange] (A4) -- (B1);
\draw[pcr,color=magenta] (B1) -- (Bt);
\draw[pcr,color=olive] (A4) -- (B2);

\node at ($($(o1)!0.5!(o2)$)+(90:.1)$) {\footnotesize $x_{f_1}^{(\textcolor{red}{2})} + x_{f_2}^{(\textcolor{green}{4})} + x_{f_3}^{(\textcolor{blue}{4})} = 1$};

\node at ($1/5*(Bt)+1/5*(B1)+1/5*(B2)+1/5*(B3)+1/5*(B4)$) {\footnotesize $v$};

\node at ($(o2)+(0.15,0.05)$) {\footnotesize $x_{f_3}^{(1)} = 0$};

\node (face_eq1) at ($($(o2)!0.5!(o3)$)+(0.4,0)$) {\footnotesize $x_{f_2}^{(\textcolor{violet}{3})} = x_{f_2}^{(\textcolor{magenta}{1})} + \phi x_{f_2}^{(\textcolor{orange}{2})} $ };
\node (face_eq2) at ($(face_eq1)+(0,-0.13)$) {\footnotesize $x_{f_2}^{(\textcolor{green}{4})} = \phi x_{f_2}^{(\textcolor{magenta}{1})} + x_{f_2}^{(\textcolor{orange}{2})} $ };

\node (node_eq1) at ($(face_eq2)+(0,-0.2)$) {\footnotesize $x_{v} = x_{f_2}^{(\textcolor{magenta}{1})} $ };
\node (node_eq2) at ($(node_eq1)+(0,-0.13)$) {\footnotesize $x_{v} = x_{f_2}^{(\textcolor{orange}{2})} + x_{f_4}^{(\textcolor{olive}{1})} $ };
        
\end{tikzpicture}

\caption{Examples for the different types of equations before substituting the variables~$x_{f_i}^{(3)}$ and~$x_{f_i}^{(4)}$.}

\label{fig:equations}

\end{figure}
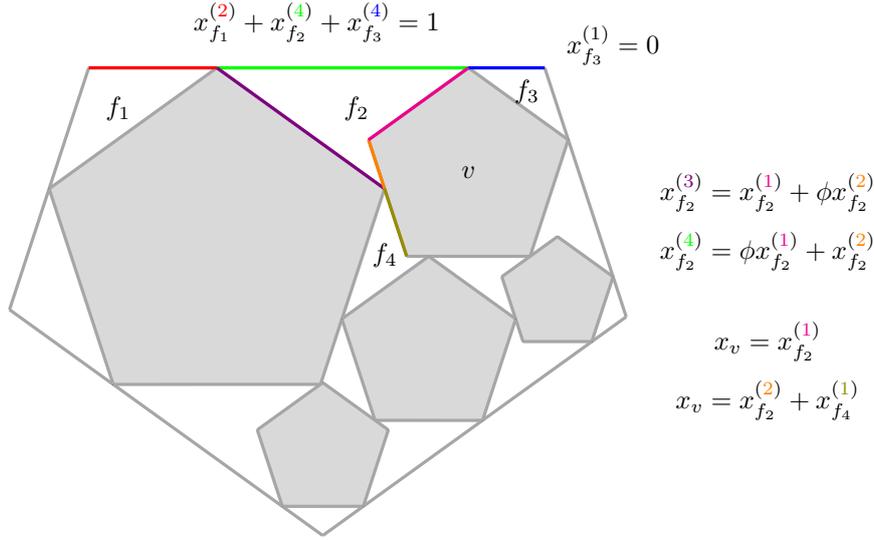

In the equations
\begin{equation}\label{eq:vertex_equations}
 \sum_{f\in \delta_1(a_1)}x_f^{(j_{a_1,f,1})}=1 \quad \text{and}\quad \sum_{f\in \delta_i(v)}x_f^{(j_{v,f,i})}-x_v=0
\end{equation} 
we eliminate the variables~$x_f^{(3)},x_f^{(4)}$ using substitutions
according to the equations of Lemma~\ref{lemma:face_equations}.  The resulting system
of linear equations is denoted~$A_F \mathbf{x} = \mathbf{e_1}$, here~$A_F$ is
the coefficient matrix depending on the five color forest~$F$
and~$\mathbf{e_1}$ is the first standard unit vector.

We next show that the system of linear equations is uniquely solvable.
For this purpose we need a lemma about perfect matchings in
plane bipartite graphs.  The lemma is well known from the context of
Pfaffian orientations, see e.g., \cite{Thom-06}, we include a proof
for completeness.  Let~$H$ be a bipartite graph with vertex classes~${
  \{ v_1,\dotsc,v_k \} }$ and~${ \{ w_1,\dotsc,w_k \} }$.  Then a
perfect matching of~$H$ induces a permutation~$\sigma \in
\mathcal{S}_k$ by~${ \sigma(i) = j :\Leftrightarrow \{ v_i,w_j \} \in
  M }$.  We define the \emph{sign} of a perfect matching~$M$, denoted by~$\sgn(M)$,
as the sign of the corresponding permutation.

\begin{lemma} \label{lemma:perfect_matchings_same_sign} Let~$H$ be a
  bipartite graph and let~$M,M'$ be two perfect matchings of~$H$.  If
  the symmetric difference of~$M$ and~$M'$ is the disjoint union of
  simple cycles~$C_1,\dotsc,C_m$ such that, for~$i=1,\dotsc,m$, the
  length~$\ell_i$ of~$C_i$ fulfills~$\ell_i \equiv 2 \mod 4$,
  then~$\sgn(M)=\sgn(M')$.

  If~$H$ is a plane graph such that each inner face~$f$ of~$H$ is
  bounded by a simple cycle of length~$\ell_f \equiv 2 \mod 4$, this
  property is fulfilled for any two perfect matchings of~$H$.
  Therefore we have~$\sgn(M)=\sgn(M')$ for any two perfect matchings~$M,M'$ of~$H$
  in this case.
\end{lemma}

\begin{proof}
  For~$i=1,\dotsc,m$, there is an~$n_i\in\mathbb{N}$ with~$\ell_i=4n_i
  + 2$.  Then on the vertices of~$C_i$ the permutation~$\sigma$
  corresponding to~$M$ and the permutation~$\sigma'$ corresponding
  to~$M'$ differ in a cyclic permutation~$\tau_i$ of length~$2n_i+1$.
  See Fig.~\ref{fig:permutations}.
  Hence, we have~$\sigma'=\sigma\circ\tau_1\circ\cdots\circ\tau_m$ and
  therefore
\begin{align*}
\sgn(\sigma') &= \sgn(\sigma) \cdot \sgn(\tau_1) \cdots \sgn(\tau_m) \\
              &= \sgn(\sigma) \cdot (-1)^{2n_1} \cdots (-1)^{2n_m} = \sgn(\sigma) \enspace .
\end{align*}

In the case that~$H$ is a plane graph such that each inner face~$f$
of~$H$ is bounded by a simple cycle of length~$\ell_f \equiv 2 \mod
4$, for each cycle of length~$\ell$ with~$k'$ vertices in its interior
the formula~$\ell + 2k' \equiv 2 \mod 4$ is valid.  This can be shown
by induction on the number of faces enclosed by the cycle.  Since each
of the cycles~$C_1,\dotsc,C_m$ contains an even number of vertices in
its interior, this implies~$\ell_i \equiv 2 \mod 4$
for~$i=1,\dotsc,m$.
\end{proof}

\begin{figure}

\centering

\begin{tikzpicture}[scale=2]

\foreach \j in {1,...,5} {
  \node[circle,draw,scale=0.5] (v\j) at (\j * 72:1) {};
  \node[circle,draw,fill,scale=0.5] (w\j) at (\j * 72 + 36:1) {};
}

\foreach \i/\j in {1/5,2/1,3/2,4/3,5/4} {
  \draw[color=red,very thick] (v\i) -- (w\i);
  \draw[color=red,dotted,-latex',thick] (v\i) to [in=310-72 + \i * 72, out=230-72 + \i * 72] (w\i);
  \draw[color=blue,very thick] (v\i) -- (w\j);
  \draw[color=blue,dotted,-latex',thick] (v\i) to [in=230-72-36 + \i * 72, out=310-72-36 + \i * 72] (w\j);
  \draw[color=green,dotted,latex'-,thick] (w\j) to [in=50+72-72-72 + \i * 72, out=130+36-72-72 + \i * 72] (w\i);
}

\node at (38:1.7) {\footnotesize $\textcolor{red}{M},\textcolor{blue}{M'}$};
\node at (90:.65) {\footnotesize $\textcolor{red}{\sigma}$};
\node at (90-36:.65) {\footnotesize $\textcolor{blue}{\sigma'}$};
\node at (90-18:1.3) {\footnotesize $\textcolor{green}{\tau_i}$};
        
\end{tikzpicture}

\caption{Illustration of~$\sigma'=\sigma\circ\tau_i$ in the proof of Lemma~\ref{lemma:perfect_matchings_same_sign}.}

\label{fig:permutations}

\end{figure}
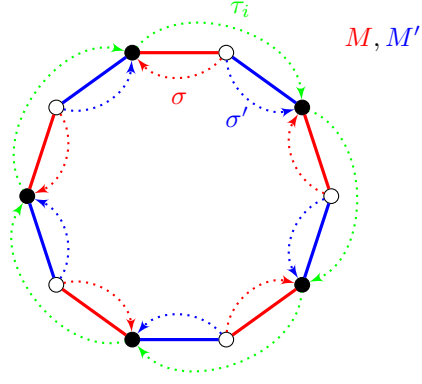

\begin{theorem} \label{thm:uniquely_solvable}
The system~$A_F \mathbf{x} = \mathbf{e_1}$ is uniquely solvable.
\end{theorem}

\begin{proof}
  We show that~$\det(A_F)\neq 0$.  Let~$\hat{A}_F$ be the matrix
  obtained from~$A_F$ by multiplying all columns corresponding to inner
  vertices of~$G$ with~$-1$.  Since in~$A_F$ all entries in these
  columns are non-positive (all vertex-variables have negative coefficients
  in~(\ref{eq:vertex_equations})) and the entries in all other columns are
  non-negative, all entries of~$\hat{A}_F$ are non-negative.  Further
  we have~$\det(A_F)=(-1)^{n} \det(\hat{A}_F)$ where~$n$ is the number of inner vertices of~$G$.

  Now we want to interpret the Leibniz formula of~$\det(\hat{A}_F)$ as
  the sum over the perfect matchings of a plane auxiliary graph~$H_F$.
  Let~$H_F$ be the bipartite graph whose first vertex
  class~$v_1,\dotsc,v_k$ consists of the variables of the equation
  system and whose second vertex class~$w_1,\dotsc,w_k$ consists of
  the equations of the equation system.  There is an edge~$v_i w_j$
  in~$H_F$ if and only if~$(\hat{A}_F)_{ij}>0$.  Then we have
\[
  \det(\hat{A}_F) = \sum_{\sigma} \sgn(\sigma) \prod_{i} (\hat{A}_F)_{i \sigma(i)}
                  = \sum_{M} \sgn(M) P_M \enspace ,
\]
where the second sum goes over all perfect matchings of~$H_F$.
The idea of the second equality is to ignore all permutations~$\sigma$
with~$\prod_{i} (\hat{A}_F)_{i \sigma(i)} = 0$ and we have
for each perfect matching~${M=\{ v_1 w_{\sigma(1)}, \dotsc, v_k w_{\sigma(k)} \}}$
a product~$P_M = \prod_{i} (\hat{A}_F)_{i \sigma(i)} > 0$ in the final sum.

Next we will define an embedding of~$H_F$ into the plane.
See Fig.~\ref{fig:graph_embedding} for an illustration. We start
with a crossing-free straight-line drawing of~$G$.
Then we draw the missing outgoing edges (see Lemma~\ref{lem:missing_edge})
as segments starting at a vertex and ending inside a face of the drawing.
After that we put pairwise disjoint disks around the inner vertices and cut the
bordering circles at the intersections with the five outgoing edges
of the vertex (including the edges we added in the last step) into five arcs.  These five arcs are the drawings of
the five equation-vertices incident to the respective vertex and each
of these arcs is contained in exactly those faces of the embedding
of~$G$ which are involved in the corresponding equation.  Then every
inner face~$f$ of~$G$ is intersected by exactly four of these arcs,
two from the incident vertices with the missing outgoing edge and one from the other two
vertices.  We denote them by~$A_1,\dotsc,A_4$ in clockwise order
where~$A_1$ and~$A_2$ come from the same vertex.  We place the
vertices corresponding to~$x_f^{(1)}$ and~$x_f^{(2)}$ inside~$f$, but
outside of the disks of the three incident vertices of~$f$.  We
connect~$x_f^{(1)}$ to~$A_1$ and~$A_4$, and~$x_f^{(2)}$ to~$A_2$
and~$A_3$.  Up to this point the drawing is crossing-free.  Finally we
add the two intersecting edges~$x_f^{(1)} A_3$ and~$x_f^{(2)} A_4$
inside~$f$.
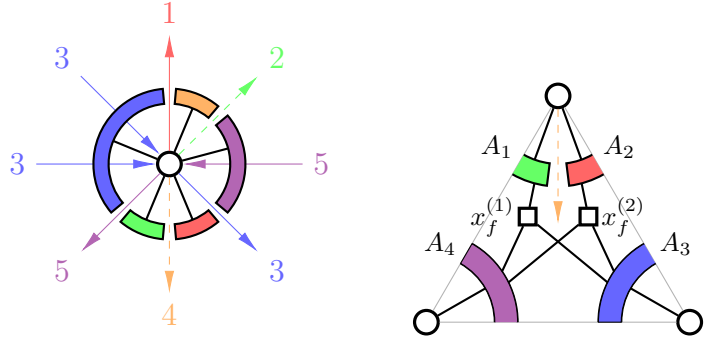
\begin{figure}

\centering

\tikzstyle{vertex}=[draw,very thick,circle,scale=0.8]
\tikzstyle{edge}=[-{Latex[scale length=1.8,scale width=1]}]
\tikzstyle{missing edge}=[edge,dashed]
\tikzstyle{eq vertex}=[thick]
\tikzstyle{H edge}=[thick]
\tikzstyle{face vertex}=[draw,very thick,rectangle,scale=0.8]

\begin{tikzpicture}

\pgfmathsetmacro{\eps}{5}
\pgfmathsetmacro{\rad}{.8}
\pgfmathsetmacro{\radout}{1}

\node[vertex] (v) at (0,0) {};

\node[color3light] (v1) at (135:2) {$3$};
\node[color3light] (v2) at (180:2) {$3$};
\node[color5light] (v3) at (225:2) {$5$};
\node[color4light] (v4) at (270:2) {$4$};
\node[color3light] (v5) at (315:2) {$3$};
\node[color5light] (v6) at (0:2) {$5$};
\node[color2light] (v7) at (45:2) {$2$};
\node[color1light] (v8) at (90:2) {$1$};

\draw[edge,color3light] (v1) -- (v);
\draw[edge,color3light] (v2) -- (v);
\draw[edge,color5light] (v) -- (v3);
\draw[missing edge,color4light] (v) -- (v4);
\draw[edge,color3light] (v) -- (v5);
\draw[edge,color5light] (v6) -- (v);
\draw[missing edge,color2light] (v) -- (v7);
\draw[edge,color1light] (v) -- (v8);

\foreach \alpha/\beta/\color in {45/90/fill4light, 90/225/fill3light, 225/270/fill2light, 270/315/fill1light, 315/405/fill5light}
    \draw[eq vertex, \color] (v) ++(\alpha+\eps:\rad) arc (\alpha+\eps:\beta-\eps:\rad) -- (\beta-\eps:\radout) arc (\beta-\eps:\alpha+\eps:\radout) -- cycle;

\draw[H edge] (v) -- (67.5:\rad);
\draw[H edge] (v) -- (157.5:\rad);
\draw[H edge] (v) -- (247.5:\rad);
\draw[H edge] (v) -- (292.5:\rad);
\draw[H edge] (v) -- (15:\rad);

\end{tikzpicture}
\qquad
\begin{tikzpicture}

\pgfmathsetmacro{\eps}{7}
\pgfmathsetmacro{\rad}{.9}
\pgfmathsetmacro{\radout}{1.2}

\node[vertex] (v1) at (90:2) {};
\node[vertex] (v2) at (-30:2) {};
\node[vertex] (v3) at (-150:2) {};

\draw[color=gray!60] (v1) -- (v2) -- (v3) -- (v1);
\draw[missing edge,color4light] (v1) -- (90:0.3);

\draw[eq vertex, fill2light] (v1) ++(-120:\rad) arc (-120:-90-\eps:\rad) -- ++(-90-\eps:\radout-\rad) arc (-90-\eps:-120:\radout);
\draw[eq vertex, fill1light] (v1) ++(-60:\rad) arc (-60:-90+\eps:\rad) -- ++(-90+\eps:\radout-\rad) arc (-90+\eps:-60:\radout);
\draw[eq vertex] (v2) ++(120:\rad) arc (120:180:\rad) +(180:\radout-\rad) arc (180:120:\radout);
\draw[ fill3light] (v2) ++(120:\rad) arc (120:180:\rad) -- +(180:\radout-\rad) arc (180:120:\radout);
\draw[eq vertex] (v3) ++(0:\rad) arc (0:60:\rad) +(60:\radout-\rad) arc (60:0:\radout);
\draw[fill5light] (v3) ++(0:\rad) arc (0:60:\rad) -- +(60:\radout-\rad) arc (60:0:\radout);

\node[face vertex] (f1) at (-0.4,0.4) [label={left,label distance=-.1cm}: \footnotesize $x_f^{(1)}$] {};
\node[face vertex] (f2) at (0.4,0.4) [label={right,label distance=-.1cm}: \footnotesize $x_f^{(2)}$] {};

\draw[H edge] (f1) -- ($(v1)+(-105:\radout)$);
\draw[H edge] (f2) -- ($(v1)+(-75:\radout)$);
\draw[H edge] (f1) -- ($(v2)+(154:\radout)$);
\draw[H edge] (f2) -- ($(v2)+(146:\radout)$);
\draw[H edge] (f1) -- ($(v3)+(34:\radout)$);
\draw[H edge] (f2) -- ($(v3)+(26:\radout)$);

\draw[H edge] (v2) -- ($(v2)+(150:\rad)$);
\draw[H edge] (v3) -- ($(v3)+(30:\rad)$);
\draw[H edge] (v1) -- ($(v1)+(-110:\rad)$);
\draw[H edge] (v1) -- ($(v1)+(-70:\rad)$);

\node at ($(v1)+(-140:.5*\rad+.5*\radout)$) {\footnotesize $A_1$};
\node at ($(v1)+(-40:.5*\rad+.5*\radout)$) {\footnotesize $A_2$};
\node at ($(v2)+(100:.5*\rad+.5*\radout)$) {\footnotesize $A_3$};
\node at ($(v3)+(80:.5*\rad+.5*\radout)$) {\footnotesize $A_4$};

\end{tikzpicture}

\caption{Embedding of~$H_F$ into the plane.}
\label{fig:graph_embedding}
\end{figure}
\medskip

\begin{myclaim}
The graph~$H_F$ has a perfect matching.
\end{myclaim}

\begin{claimproof}
  We describe an explicit construction of a perfect matching of~$H_F$.
  The five equation-vertices adjacent to a vertex~$v$ of~$G$ are
  corresponding to the five colors of the five color forest.  We
  always match the vertex~$v$ with the equation-vertex of color~$4$.
  The equation-vertices of colors~$2$ and~$3$ are matched with one of
  the two variable-vertices of the last incident face in clockwise
  order, and the equation-vertices of colors~$5$ and~$1$ are matched
  with one of the two variable-vertices of the last incident face in
  counterclockwise order (see Fig.~\ref{fig:perf-match} (left)).
\calc_figscale{55}%
\begin{figure}[t]
    \centerline{\input{\fpath/perf-match.pstex_t}}
    \caption{\label{fig:perf-match}}
    \end{figure}%
VC
{Left: Cases for the construction of the perfect matching.
Right: The complete perfect matching in a small instance.}%

Now we will show that each pair of face vertices (except for the five corner-faces,
i.e.\ the five faces incident to two outer vertices) is
matched exactly twice.  We call a segment of a facial quadrilateral~$B$ a
\emph{short segment} if it is incident to the concave corner, and we call
it a \emph{long segment} otherwise.  For two adjacent segments of~$B$ we
call the segment, whose containing pentagon side ends in the contact point of the two
segments, the \emph{cut segment}.  Above we mentioned that the five
equation-vertices of a vertex of~$G$ correspond to the five colors.
Since each segment of~$B$ is involved in exactly one of these equations,
we can also associate a color with each of the segments of~$B$.  We
distinguish three cases concerning segments of color~$4$
(Fig.~\ref{fig:perf-match} (left) shows two of the cases).  If~$B$ has a
short segment of color~$4$, the other short segment and the cut long segment
are matched.  If~$B$ has a long segment of color~$4$, the short segment,
which is neighboring the segment of color~$4$, and the cut segment of the
two remaining segments are matched.  If~$B$ has no segment of color~$4$, for
each pair of neighboring long and short segments the cut segment is matched.
Hence, in every case the face is matched exactly twice.

Since each face is matched exactly twice, its two variable-vertices
are matched exactly once.  It can easily be seen that each corner-face
is matched exactly once, except the corner-face of color~$4$ which is
not matched.  Finally we match the corner-face of color~$4$ with the
equation-vertex corresponding to the non-homogeneous equation, and
obtain a perfect matching. Figure~\ref{fig:perf-match} (right) shows an
example.
\end{claimproof}

Let~$\mathcal{M}_0$ be the set of perfect matchings of~$H_F$ that do
not contain both edges of any pair of crossing edges.

\begin{myclaim} \label{claim:min_matchings_same_sign}
Let~$M_1,M_2 \in \mathcal{M}_0$.
Then~$\sgn(M_1) = \sgn(M_2)$.
\end{myclaim}

\begin{claimproof}
  Since each vertex has degree~$1$ in~$M_1$ and in~$M_2$,
  each vertex has degree~$0$ or~$2$ in the symmetric difference
  of~$M_1$ and~$M_2$. Therefore
  the symmetric difference of~$M_1$ and~$M_2$ is a disjoint union of
  simple cycles.  Let~$C$ be one of these cycles.  Due to
  Lemma~\ref{lemma:perfect_matchings_same_sign} it suffices to show
  that~$\ell(C) \equiv 2 \mod 4$.

Now consider a pair~$v_a w_b , v_c w_d$ of crossing edges that are
both contained in~$C$, i.e., $v_a$ and~$v_c$ are the two variable vertices
of a face and~$w_b,w_d$ correspond to the equations for sidelength
of different pentagons, in particular $v_a w_b v_c w_d$ is a 4-cycle
in $H_F$. Since the edges~$v_a w_b , v_c w_d$ cannot be contained in the same
matching, we can assume that the edge~$v_a w_b$ is contained in~$M_1$ and the
edge~$v_c w_d$ in~$M_2$.

There are two paths~$P_1,P_2$ such that~$C=v_aw_bP_1w_d v_cP_2$.
If~$C$ looked like~$C=v_a w_bP_1v_c w_dP_2$,
the path~$P_1$ would start with an edge of~$M_2$ and end with an
edge of~$M_1$. Therefore~$P_1$ would have even length and the cycle $v_c
w_bP_1$ would have odd length, and that is not possible since~$H_F$ is
bipartite.

Let~$C'$ be the cycle defined by~$C' := v_c w_b P_1 w_d v_a\overline{P_2}$
where~$\overline{P_2}$ denotes the reversed path~$P_2$.  Note
that~$\ell(C') = \ell(C)$ and that~$C'$ is contained in the symmetric
difference of the perfect matchings~$M_1'$ and~$M_2'$ obtained
from~$M_1$ and~$M_2$ in the following way: In~$M_1$ we
replace~$v_a w_b$ by~$v_a w_d$ and in~$M_2$ we replace~$v_c w_d$
by~$v_c w_b$.  Further all edges of~$P_1$ switch the matching.
See Fig.~\ref{fig:matching_crossing}.

\begin{figure}
\centering
\includegraphics[]{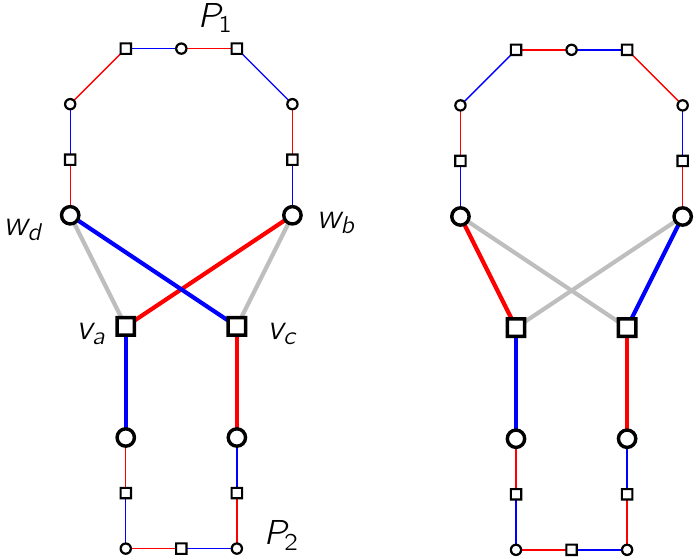}
\caption{Reducing the number of crossings in the symmetric difference
  of two perfect matchings of~$H_F$.}
\label{fig:matching_crossing}
\end{figure}

The swap does not create new crossings since the only new edges~$v_a w_d$
and~$v_c w_b$ have no crossing edges.
Therefore the cycle~$C'$ has one crossing less than~$C$.  By iterating this
process we obtain a cycle~$C''$ with~$\ell(C'')=\ell(C)$ which is
crossing-free and contained in the symmetric difference of two
perfect matchings of~$H_F$. Since~$C''$ is crossing-free, it is also
contained in a spanning subgraph~$H_F'$ of~$H_F$ that keeps all non-crossing
edges of~$H_F$ and exactly one of the two edges of every
crossing pair. Then~$H_F'$ inherits a crossing-free drawing from~$H_F$
where every inner face is a simple cycle of length~$6$ (note that
every face corresponds to a corner of a pentagon in the pentagon
contact representation).  Due to
Lemma~\ref{lemma:perfect_matchings_same_sign} this implies~$\ell(C'')
\equiv 2 \mod 4$.
\end{claimproof}

Now we will consider all perfect matchings of~$H_F$.  We divide them
into equivalence classes according to the following equivalence
relation: Two perfect matchings~$M,M'$ are equivalent if~$M'$ can be
obtained from~$M$ by exchanging pairs of crossing edges with the two
non-crossing edges on the same four vertices (we call them \emph{twin
  edges}).  Exchanges in both directions are allowed.  Note that each
equivalence class contains exactly one of the matchings
of~$\mathcal{M}_0$ we considered in
Claim~\ref{claim:min_matchings_same_sign}.  Let~$M_0\in\mathcal{M}_0$ be
one of these matchings and let~$k=k(M_0)$ be the number of pairs of
twin edges contained in~$M_0$. Because each pair of twin edges can be
exchanged independently with the corresponding pair of crossing edges,
the equivalence class~$\mathcal{A}_{M_0}$ of $M_0$ consists of $2^k$
matchings and for~$i=0,\dotsc,k$, the class contains~$\binom{k}{i}$
matchings with exactly~$i$ pairs of crossing edges.
Note that the entries of~$A_F$ corresponding to twin edges are~$\phi$
and that the entries corresponding to crossing edges are~$1$, see the equations in Lemma~\ref{lemma:face_equations}.
Thus in the product $P_{M_0}$ of entries of $A_F$ corresponding to 
the edges of~$M_0$ we see a contribution of $\phi^2$ for each 
pair of twin edges.
The contribution for a pair of crossing edges is $1$. Therefore
\[
  \sum_{M \in \mathcal{A}_{M_0}} \sgn(M) P_M
    = \sum_{i=0}^{k} \binom{k}{i} \sgn(M_0) (-1)^{i} P_{M_0} \left(\frac{1}{\phi^2}\right)^{i}
    = \sgn(M_0) P_{M_0} \left(1 - \frac{1}{\phi^2}\right)^k
\]
and
\[
  \det(\hat{A}_F)
    = \sum_{M_0 \in \mathcal{M}_0} \sum_{M \in \mathcal{A}_{M_0}} \sgn(M) P_M
    = \sum_{M_0 \in \mathcal{M}_0} \sgn(M_0) \underbrace{P_{M_0}}_{>0}  
    \underbrace{\left(1 - \frac{1}{\phi^2}\right)^{k(M_0)}}_{>0} \enspace .
\]
Because of Claim~\ref{claim:min_matchings_same_sign} we
have~$\sgn(M_0)=\sgn(M_0')$ for any two matchings~$M_0,M_0' \in
\mathcal{M}_0$.  Finally this implies~$\det(\hat{A}_F) \neq 0$.
\end{proof}

The following lemma will help us to prove that a non-negative solution
of the system~$A_F \mathbf{x} = \mathbf{e_1}$ leads to a regular pentagon
contact representation of~$G$.

\begin{lemma} \label{lemma:triangulation_given_face_triangles} 
  Let~$H$ be an inner triangulation of a polygon.  For every inner
  face~$f$ of~$H$ with vertices~$v_1,v_2,v_3$ in clockwise order
  let~$T_f$ be a triangle in the plane whose vertices have coordinates denoted by
  $p(f,v_1),p(f,v_2),p(f,v_3)$ in clockwise order such that the
  following conditions are satisfied:
\begin{enumerate}[(i)]
\item \label{item:vertex_angle_sum} 
  For each inner vertex $v$ of~$H$ with incident
  faces~$f_1,\dotsc,f_k$
    \[\sum_{i=1}^{k} \beta(f_i,v) = 2\pi\]
  where~$\beta(f,v)$ denotes the inner angle of~$T_f$ at~$p(f,v)$.
\item \label{item:outer_vertex_angle_sum}
  For each outer vertex~$v$ of~$H$ with incident faces~$f_1,\dotsc,f_k$
  \[\sum_{i=1}^{k} \beta(f_i,v) \leq \pi \enspace . \]
\item \label{item:edge_same_length} 
  For each inner edge $vw$ of~$H$ with incident faces~$f_1,f_2$ 
    \[p(f_1,v) - p(f_1,w) = p(f_2,v) - p(f_2,w) \enspace , \] 
  i.e., the vector between~$v$ and~$w$ is the same in~$T_{f_1}$
  and~$T_{f_2}$.
\end{enumerate}
Then there exists a crossing-free straight line drawing of~$H$ such
that the drawing of every inner face~$f$ can be obtained from~$T_f$ by
translation.
\end{lemma}

\begin{proof} Let~$H^\ast$ be the dual graph of~$H$ without the vertex
corresponding to the outer face of~$H$.  Further let~$S$ be a spanning
tree of~$H^\ast$.  Then by property~(\ref{item:edge_same_length}) we can glue the triangles~$T_f$ of all inner
faces~$f$ of~$H$ together along the edges of~$S$. 
This determines a unique position for every polygon,
up to a global motion.  We need to show
that the resulting shape has no holes or overlappings.  For the edges
of~$S$ we already know that the triangles of the two incident faces
are touching in the right way.  For the edges of the
complement~$\overline{S}$ of~$S$ we still need to show this.  We
consider~$\overline{S}$ as a subset of the edges of~$H$.  Note
that~$\overline{S}$ is a forest in~$H$.  Let~$e$ be an edge
of~$\overline{S}$ incident to a leaf~$v$ of this forest that is an
inner vertex of~$H$.  Then for all incident edges~$e'\neq e$ of~$v$ we
already know that the triangles of the two incident faces of~$e$ are
touching in the right way.  But then also the two triangles of the two
incident faces of~$e$ are touching in the right way because~$v$
fulfills property~(\ref{item:vertex_angle_sum}).  Since the set of
edges we still need to check is still a forest, we can iterate this
process until all inner edges of~$H$ are checked.

We have to exclude that the resulting polygon has overlappings.
Let~$F$ be the set of faces of~$H$, let~$F_{\operatorname{in}}$ be the set of inner
faces of~$G$, let~$V$ be the set of vertices of~$H$, and
let~$V_{\operatorname{out}}$ be the set of the outer vertices of~$H$. Let~$d=|V_o|$.

\begin{myclaim}\label{claim:outer_angle_sum}
$ \sum_{v\in V_o} \sum_i \beta(f_i,v) = (d-2) \pi $.
\end{myclaim}

\begin{claimproof}
The sum of the inner angles of each triangle~$T_f$ is~$\pi$.
Summing this over all triangles~$T_f$ we obtain
\[ \sum_{f\in F_{\operatorname{in}}} \pi = (|F| - 1) \pi \enspace . \]
Using property~(\ref{item:vertex_angle_sum}) we also have
\[ \sum_{f\in F_{\operatorname{in}}} \pi = (|V| - d)2\pi + \sum_{v\in V_o} \sum_i \beta(f_i,v) \enspace . \]
Using Euler's formula this yields the claim.
\end{claimproof}

Due to Claim~\ref{claim:outer_angle_sum} the sum of the angles at the outer vertices is
just the right value for a simple $d$-gon.
Because of property~(\ref{item:outer_vertex_angle_sum}) the angles
at the boundary are convex.
Thus the resulting shape is a simple convex polygon and therefore non-intersecting
\end{proof}

\begin{theorem}
  The unique solution of the system~$A_F \mathbf{x} = \mathbf{e_1}$ is
  non-negative if and only if the five color forest~$F$ is induced by
  a regular pentagon contact representation of~$G$.
\end{theorem}

\begin{proof}
  Assume there is a regular pentagon representation~$\mathcal{S}$
  of~$G$ that induces the five color forest~$F$.  Then the edge
  lengths given by~$\mathcal{S}$ define a non-negative solution
  of~$A_F \mathbf{x} = \mathbf{e_1}$.

  For the opposite direction, assume the solution of~$A_F \mathbf{x} =
  \mathbf{e_1}$ is non-negative.  To be able to apply
  Lemma~\ref{lemma:triangulation_given_face_triangles} we first construct
  an internally triangulated extension of the skeleton graph of a
  hypothetical regular pentagon contact representation with induced
  five color forest~$F$.  We start with a crossing-free straight-line
  drawing of~$G$. Add a subdivision vertex on each edge of $G$.
  Moreover, for each inner vertex~$v$ draw an edge ending at a
  new vertex inside each face with a missing outgoing edge
  of~$v$. Then connect all the new adjacent vertices of $v$ in the
  cyclic order given by the drawing.

  At this point inner faces of $G$ are subdivided into four triangles
  and a quadrangle (shown gray in Fig.~\ref{fig:extended}).  One of
  the vertices of the quadrangle is the new inner vertex $w$ of the
  face. Connect vertex~$w$ to the subdivision vertex diagonally in the
  quadrangle (shown dashed in Fig.~\ref{fig:extended}).
\calc_figscale{40}%
\begin{figure}[t]
    \centerline{\input{\fpath/extended.pstex_t}}
    \caption{\label{fig:extended}}
    \end{figure}%
VC
{Part of a triangulation with a five color forest, the
resulting skeleton graph, and a realization.}%

Using the edge lengths given by the solution of~$A_F \mathbf{x} = \mathbf{e_1}$
we can compute the side length of each triangle~$T_f$ corresponding to
an inner face~$f$ of this skeleton graph such that an application of
Lemma~\ref{lemma:triangulation_given_face_triangles} gives us a regular
pentagon contact representation of~$G$ with induced five color
forest~$F$.
\end{proof}

The algorithm first computes an arbitrary five color forest of~$G$
(this is possible in linear time by the construction from
Theorem~\ref{thm:fcf_existence} since Schnyder woods can be constructed in
linear time, see for example \cite{brehm20003,kobourov2016canonical}).  Based on the five
color forest the algorithm generates the corresponding system of
linear equations and solves it. Since the size of the system is linear
in the size of the input graph, this can be done in cubic time, for example with Gaussian elimination.  If
the solution is non-negative, we can construct the regular pentagon
contact representation from the edge lengths given by the solution and
we are done.  If the solution has negative variables, we would like to
change the five color forest and proceed with the new one.

We now show a way of changing the five color
forest in the case of a solution with negative variables.

\begin{theorem}
  The negative and non-negative variables of the solution
  of~$A_F \mathbf{x} = \mathbf{e_1}$ are separated by a disjoint union of
  directed simple cycles in the $\alpha_5$-orientation corresponding
  to~$F$.  If there are negative variables, this union is non-empty.
\end{theorem}

\begin{proof}
  We say that an abstract pentagon or an abstract facial quadrilateral has a \emph{sign-change}
  at a point~$p$ on its boundary if one of the two variables corresponding to the two
  boundary edges with common end point~$p$ has a negative solution value and the other
  variable has a non-negative solution value.  
  We call the following two types of edges in the
  $\alpha_5$-orientation \emph{sign-separating edges} (see
  Fig.~\ref{fig:sign_changing_edges} for an illustration):
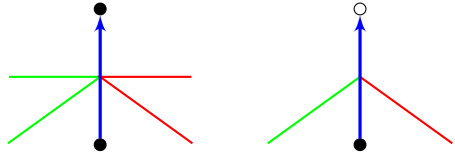
\begin{figure}

\centering

\tikzstyle{posneg}=[font=\footnotesize]
\tikzstyle{positive}=[color=green,posneg]
\tikzstyle{negative}=[color=red,posneg]

\tikzstyle{vertex}=[circle,scale=0.4]
\tikzstyle{primal vertex}=[vertex,draw,fill]
\tikzstyle{dual vertex}=[vertex,draw]
\tikzstyle{out edge}=[-latex',very thick,color=blue]
\tikzstyle{in edge}=[-latex',very thick,color=blue]

\tikzstyle{pos edge}=[green,thick]
\tikzstyle{neg edge}=[red,thick]
\tikzstyle{neutral edge}=[gray!40,thick]

\begin{tikzpicture}[scale=1.5]

\coordinate (O) at (0,0);
\coordinate (Lu) at (216:1);
\coordinate (Ru) at (324:1);
\coordinate (Lo) at (180:.8);
\coordinate (Ro) at (0:.8);

\draw[pos edge] (Lo) -- (O) -- (Lu);
\draw[neg edge] (Ro) -- (O) -- (Ru);

\node[primal vertex] (v) at (-90:.6) {};
\node[primal vertex] (w) at (90:.6) {};

\draw[out edge] (v) -- (w);

\end{tikzpicture}
\qquad
\begin{tikzpicture}[scale=1.5]

\coordinate (O) at (0,0);
\coordinate (Lu) at (216:1);
\coordinate (Ru) at (324:1);
\coordinate (Lo) at (180:.8);
\coordinate (Ro) at (0:.8);

\draw[pos edge] (O) -- (Lu);
\draw[neg edge] (O) -- (Ru);

\node[primal vertex] (v) at (-90:.6) {};
\node[dual vertex] (w) at (90:.6) {};

\draw[out edge] (v) -- (w);

\end{tikzpicture}

\caption{The two types of sign-separating edges.
  Red (green) edges are edges of the skeleton graph with negative (non-negative) solution values.}

\label{fig:sign_changing_edges}

\end{figure}
Edges of the first type are normal edges~$vw$ such that the abstract
pentagon of~$v$ has a sign-change at the contact point with the
abstract pentagon of~$w$, and both abstract facial quadrilaterals
incident to this touching point do not have a sign-change at this
point.  Edges of the second type are stack edges~$vw$ with normal
vertex~$v$ and stack vertex~$w$ such that the abstract pentagon of~$v$
has a sign-change at the corner which is a concave corner of the
abstract quadrilateral of~$w$. Note that both types of sign-separating
edges correspond to a corner~$p$ of the abstract pentagon~$A$ of~$v$ fulfilling
the following property: One of the two sides of~$A$ incident to~$p$ starts with
a non-negative segment at~$p$ and the other side of~$A$ incident to~$p$ starts
with a negative segment at~$p$.

\begin{myclaim} \label{claim:facial_sign_changes} 
  In an abstract facial quadrilateral there is either no sign-change,
  or there are two sign-changes (one at a convex corner and one at the
  concave corner).
\end{myclaim}

\begin{claimproof}
  Let~$x_f^{(1)},\dotsc,x_f^{(4)}$ be the four variables of the facial
  quadrilateral as in the beginning of this section (see
  Fig.~\ref{fig:var_faces} (left)).  Then the equation system contains the two
  equations
\[
x_f^{(3)} = x_f^{(1)} + \phi x_f^{(2)} \enspace , 
\enspace x_f^{(4)} = \phi x_f^{(1)} + x_f^{(2)} 
\]
where~$\phi$ is the golden ratio.
The fact that~$\phi > 1$ immediately implies the claim.
\end{claimproof}

\begin{claim} \label{claim:existence_sign_sep_edge}
If there are negative variables, there exists a sign-separating edge.
\end{claim}

\begin{claimproof}
  Because of the inhomogeneous equation the solution always contains
  positive variables.  Therefore, if there are negative variables, at
  some point of the abstract pentagon contact representation there has
  to be a sign-change.

  We distinguish two cases concerning the two possibilities of Claim~\ref{claim:facial_sign_changes}.
  In the first case there is no sign-change in any abstract facial quadrilateral. Then we can
  distinguish negative faces (the faces with all four variables negative) and non-negative faces
  (the faces with all four variables non-negative). The above observation implies
  that there is at least one face of each kind if there are negative variables. Since the abstract pentagon
  contact representation is connected, there has to be a point where the abstract facial quadrilaterals of a negative
  and a non-negative face touch. At such a point there is a sign-separating edge of the first type.
  
  In the second case there exists at least one abstract facial quadrilateral with a sign-change.
  Due to Claim~\ref{claim:facial_sign_changes} this abstract facial quadrilateral has a sign-change
  at its concave corner. Therefore there is a sign-separating edge of the second type.
\end{claimproof}

For a sign-separating edge~$e=vw$ we will now construct an oriented walk in
the $\alpha_5$-orientation ending in~$e$. Assume that~$x_v \geq 0$
(the other case is symmetric). As we noted above, the edge~$e$ corresponds to a corner~$p$
of the abstract pentagon~$A$ of~$v$ such that one of the two sides
of~$A$ incident to~$p$ starts with a negative segment at~$p$.
If all segments of this side would be negative, this would contradict~$x_v \geq 0$.
Therefore, when we walk along this side starting at~$p$, there has to be
a first point~$q$ which is incident to a negative and a non-negative segment.
In Fig.~\ref{fig:sign_changing_edges_predecessors} we distinguish all possible cases
concerning what the abstract pentagon contact representation locally looks like
at~$q$ (including the signs of the segments) and define the previous
one, two or three edges of the walk according to this case distinction. If several edges are added, only the last
one is a sign-separating edge itself. Since the last added edge
is always a sign-separating edge, the walk can be continued from there
by iterating this process.
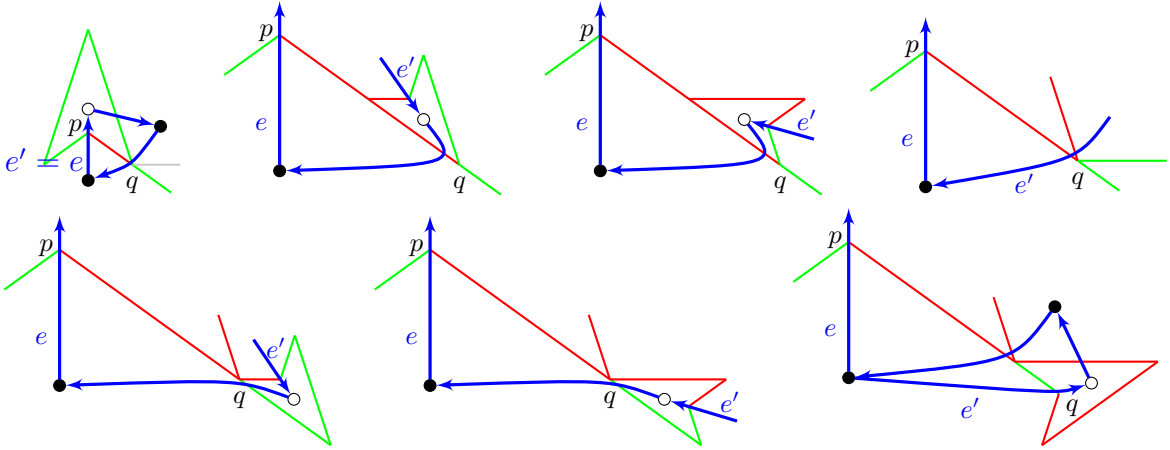
\begin{figure}

\centering

\tikzstyle{posneg}=[font=\footnotesize]
\tikzstyle{positive}=[color=green,posneg]
\tikzstyle{negative}=[color=red,posneg]

\tikzstyle{vertex}=[circle,scale=0.4]
\tikzstyle{primal vertex}=[vertex,draw,fill]
\tikzstyle{dual vertex}=[vertex,draw]
\tikzstyle{out edge}=[-latex',very thick,color=blue]
\tikzstyle{in edge}=[-latex',very thick,color=blue]

\tikzstyle{pos edge}=[green,thick]
\tikzstyle{neg edge}=[red,thick]
\tikzstyle{neutral edge}=[gray!40,thick]

\begin{tikzpicture}[scale=.9]
\coordinate (O) at (0,0);
\coordinate (L) at (216:1);
\coordinate (R) at (324:1.5);

\coordinate (L1) at (216:.8);
\coordinate (R1) at (324:.8);
\coordinate (Lh) at ($(L1)+(72:4)$);
\coordinate (Rh) at($(R1)+(108:4)$);
\coordinate (T) at (intersection of L1--Lh and R1--Rh);
\coordinate (R') at ($(R1)+(0:.7)$);

\draw[pos edge] (O) -- (L1);
\draw[neg edge] (O) -- (R1);
\draw[pos edge] (R1) -- (R);
\draw[neutral edge] (R1) -- (R');
\draw[pos edge] (R1) -- (T);
\draw[pos edge] (L1) -- (T);

\node[dual vertex] (f) at ($(O)+(90:.35)$) {};
\node[primal vertex] (v) at ($(O)+(270:.7)$) {};
\node[primal vertex] (v') at ($(R1)+(54:.7)$) {};

\draw[out edge] (v) -- node[near start, left, xshift=.1cm] {${e'=e}$} (f);
\draw[in edge] (v') .. controls (R1) .. (v);
\draw[out edge] (f) -- (v');

\node at ($(R1)+(270:.3)$) {\footnotesize $q$};
\node at ($(O)+(160:.2)$) {\footnotesize $p$};
\end{tikzpicture}
\quad
\begin{tikzpicture}[scale=.9]
\coordinate (O) at (0,0);
\coordinate (L) at (216:1);
\coordinate (R) at (324:5);

\coordinate (R1) at ($(O)!0.1!(R)$);
\coordinate (R2) at ($(O)!0.25!(R)$);
\coordinate (R3) at ($(O)!0.5!(R)$);
\coordinate (R4) at ($(O)!0.8!(R)$);
\coordinate (R') at ($(R3)!0.5!(R4)$);

\draw[pos edge] (L) -- (O);
\draw[neg edge] (O) -- (R');
\draw[pos edge] (R') -- (R4);

\coordinate (R'') at ($(O)!0.32!(R)$);
\coordinate (A) at ($(R')+(108:1.7)$);
\coordinate (Bh1) at ($(R'')+(0:4)$);
\coordinate (Bh2) at ($(A)+(252:4)$);
\coordinate (B) at (intersection of R''--Bh1 and A--Bh2);

\draw[pos edge] (R') -- (A);
\draw[pos edge] (A) -- (B);
\draw[neg edge] (B) -- (R'');

\node[primal vertex] (v) at ($(O)+(270:2)$) {};
\node[dual vertex] (f) at ($(R3)+(70:0.24)$) {};

\coordinate (O') at ($(O)+(90:0.5)$);
\coordinate (f') at ($(f)!3!(B)$);

\draw[out edge] (v) -- node[near start, left] {\footnotesize $e$} (O');
\draw[in edge] (f) .. controls (R') .. (v);
\draw[in edge] (f') -- (f);

\node at ($(R')+(270:.3)$) {\footnotesize $q$};
\node at ($(O)+(170:.2)$) {\footnotesize $p$};
\node[color=blue] at ($(f')+(-20:.4)$) {\footnotesize $e'$};
\end{tikzpicture}
\quad
\begin{tikzpicture}[scale=.9]
\coordinate (O) at (0,0);
\coordinate (L) at (216:1);
\coordinate (R) at (324:5);

\coordinate (R1) at ($(O)!0.1!(R)$);
\coordinate (R2) at ($(O)!0.25!(R)$);
\coordinate (R3) at ($(O)!0.5!(R)$);
\coordinate (R4) at ($(O)!0.8!(R)$);
\coordinate (R') at ($(R3)!0.5!(R4)$);

\draw[pos edge] (L) -- (O);
\draw[neg edge] (O) -- (R');
\draw[pos edge] (R') -- (R4);

\coordinate (R'') at ($(O)!0.32!(R)$);
\coordinate (A) at ($(R'')+(0:1.7)$);
\coordinate (Bh1) at ($(R')+(108:4)$);
\coordinate (Bh2) at ($(A)+(216:4)$);
\coordinate (B) at (intersection of R'--Bh1 and A--Bh2);

\draw[pos edge] (R') -- (B);
\draw[neg edge] (B) -- (A);
\draw[neg edge] (A) -- (R'');

\node[primal vertex] (v) at ($(O)+(270:2)$) {};
\node[dual vertex] (f) at ($(R3)+(70:0.24)$) {};

\coordinate (O') at ($(O)+(90:0.5)$);
\coordinate (f') at ($(f)!3!(B)$);

\draw[out edge] (v) -- node[near start, left] {\footnotesize $e$} (O');
\draw[in edge] (f) .. controls (R') .. (v);
\draw[in edge] (f') -- (f);

\node at ($(R')+(270:.3)$) {\footnotesize $q$};
\node at ($(O)+(170:.2)$) {\footnotesize $p$};
\node[color=blue] at ($(f')+(110:.3)$) {\footnotesize $e'$};
\end{tikzpicture}
\quad
\begin{tikzpicture}[scale=.9]
\coordinate (O) at (0,0);
\coordinate (L) at (216:1);
\coordinate (R) at (324:5);

\coordinate (R1) at ($(O)!0.1!(R)$);
\coordinate (R2) at ($(O)!0.25!(R)$);
\coordinate (R3) at ($(O)!0.4!(R)$);
\coordinate (R4) at ($(O)!0.7!(R)$);
\coordinate (R') at ($(R3)!0.5!(R4)$);

\draw[pos edge] (L) -- (O);
\draw[neg edge] (O) -- (R');
\draw[pos edge] (R') -- (R4);

\coordinate (A) at ($(R')+(108:1.3)$);
\coordinate (B) at ($(R')+(0:1.3)$);

\draw[neg edge] (A) -- (R');
\draw[pos edge] (R') -- (B);

\node[primal vertex] (v) at ($(O)+(270:2)$) {};

\coordinate (O') at ($(O)+(90:0.5)$);
\coordinate (v') at ($(R')+(414:.8)$);

\draw[out edge] (v) -- node[near start, left] {\footnotesize $e$} (O');
\draw[in edge] (v') .. controls (R') .. (v);

\node at ($(R')+(270:.3)$) {\footnotesize $q$};
\node at ($(O)+(170:.2)$) {\footnotesize $p$};
\node[color=blue] at ($($(R')!.35!(v)$)+(270:.2)$) {\footnotesize $e'$};
\end{tikzpicture}

\begin{tikzpicture}[scale=.9]
\coordinate (O) at (0,0);
\coordinate (L) at (216:1);
\coordinate (R) at (324:5);

\coordinate (R1) at ($(O)!0.1!(R)$);
\coordinate (R2) at ($(O)!0.25!(R)$);
\coordinate (R3) at ($(O)!0.5!(R)$);
\coordinate (R4) at ($(O)!0.8!(R)$);
\coordinate (R') at ($(R3)!0.5!(R4)$);

\coordinate (R'') at ($(O)!0.98!(R)$);
\coordinate (A) at ($(R'')+(108:1.7)$);
\coordinate (Bh1) at ($(R')+(0:4)$);
\coordinate (Bh2) at ($(A)+(252:4)$);
\coordinate (B) at (intersection of R'--Bh1 and A--Bh2);
\coordinate (C) at ($(R')+(108:1)$);

\draw[pos edge] (L) -- (O);
\draw[neg edge] (O) -- (R');
\draw[pos edge] (R') -- (R'');

\draw[neg edge] (C) -- (R');
\draw[neg edge] (R') -- (B);
\draw[pos edge] (B) -- (A);
\draw[pos edge] (A) -- (R'');

\node[primal vertex] (v) at ($(O)+(270:2)$) {};
\node[dual vertex] (f) at ($(R4)+(38:0.24)$) {};

\coordinate (O') at ($(O)+(90:0.5)$);
\coordinate (f') at ($(f)!3!(B)$);

\draw[out edge] (v) -- node[near start, left] {\footnotesize $e$} (O');
\draw[in edge] (f) .. controls (R') .. (v);
\draw[in edge] (f') -- (f);

\node at ($(R')+(270:.3)$) {\footnotesize $q$};
\node at ($(O)+(170:.2)$) {\footnotesize $p$};
\node[color=blue] at ($(f')+(-20:.35)$) {\footnotesize $e'$};
\end{tikzpicture}
\quad
\begin{tikzpicture}[scale=.9]
\coordinate (O) at (0,0);
\coordinate (L) at (216:1);
\coordinate (R) at (324:5);

\coordinate (R1) at ($(O)!0.1!(R)$);
\coordinate (R2) at ($(O)!0.25!(R)$);
\coordinate (R3) at ($(O)!0.5!(R)$);
\coordinate (R4) at ($(O)!0.8!(R)$);
\coordinate (R') at ($(R3)!0.5!(R4)$);

\coordinate (R'') at ($(O)!0.98!(R)$);
\coordinate (A) at ($(R')+(0:1.7)$);
\coordinate (Bh1) at ($(R'')+(108:4)$);
\coordinate (Bh2) at ($(A)+(216:4)$);
\coordinate (B) at (intersection of R''--Bh1 and A--Bh2);
\coordinate (C) at ($(R')+(108:1)$);

\draw[pos edge] (L) -- (O);
\draw[neg edge] (O) -- (R');
\draw[pos edge] (R') -- (R'');

\draw[neg edge] (C) -- (R');
\draw[neg edge] (R') -- (A);
\draw[neg edge] (A) -- (B);
\draw[pos edge] (B) -- (R'');

\node[primal vertex] (v) at ($(O)+(270:2)$) {};
\node[dual vertex] (f) at ($(R4)+(38:0.24)$) {};

\coordinate (O') at ($(O)+(90:0.5)$);
\coordinate (f') at ($(f)!3!(B)$);

\draw[out edge] (v) -- node[near start, left] {\footnotesize $e$} (O');
\draw[in edge] (f) .. controls (R') .. (v);
\draw[in edge] (f') -- (f);

\node at ($(R')+(270:.3)$) {\footnotesize $q$};
\node at ($(O)+(170:.2)$) {\footnotesize $p$};
\node[color=blue] at ($(f')+(110:.3)$) {\footnotesize $e'$};
\end{tikzpicture}
\quad
\begin{tikzpicture}[scale=.9]
\coordinate (O) at (0,0);
\coordinate (L) at (216:1);
\coordinate (R) at (324:5);

\coordinate (R') at ($(O)+(324:3)$);

\coordinate (R'') at ($(R')+(324:.8)$);
\coordinate (C) at ($(R')+(108:1)$);

\coordinate (R''') at ($(R'')+(-108:.8)$);
\coordinate (h1) at ($(R')+(0:4)$);
\coordinate (h2) at ($(R''')+(36:4)$);
\coordinate (D) at (intersection of R'--h1 and R'''--h2);

\draw[pos edge] (L) -- (O);
\draw[neg edge] (O) -- (R');
\draw[pos edge] (R') -- (R'');

\draw[neg edge] (C) -- (R');
\draw[neg edge] (R') -- (D);
\draw[neg edge] (D) -- (R''');

\draw[neg edge] (R'') -- (R''');

\node[primal vertex] (v) at ($(O)+(270:2)$) {};
\node[primal vertex] (v') at ($(R')+(54:1)$) {};
\node[dual vertex] (f) at ($(R'')+(18:0.5)$) {};

\coordinate (O') at ($(O)+(90:0.5)$);

\draw[out edge] (v) -- node[near start, left] {\footnotesize $e$} (O');
\draw[in edge] (v') .. controls (R') .. (v);
\draw[in edge] (f) -- (v');
\draw[in edge] (v) .. controls (R'') .. (f);

\node at ($(R4)+(290:.1)$) {\footnotesize $q$};
\node at ($(O)+(170:.2)$) {\footnotesize $p$};
\node[color=blue] at ($($(R4)!.45!(v)$)+(270:.25)$) {\footnotesize $e'$};
\end{tikzpicture}

\caption{Construction of the predecessors of a sign-separating edge~$e=vw$.
  The last edge~$e'$ assigned as a predecessor is a sign-separating edge itself.
  Therefore the walk can be continued from there by iterating this construction.
  Red (green) edges are edges of the skeleton graph with negative (non-negative) solution values.}

\label{fig:sign_changing_edges_predecessors}

\end{figure}

Let~$E'$ be the set of all edges occurring in any predecessor path,
including the sign-separating edges.  Then we can interpret the
predecessor assignment as an assignment from $E'$ to the same
set~$E'$.

\begin{myclaim} \label{claim:predecessor_assignment_injective}
The predecessor assignment is a permutation of the set~$E'$.
\end{myclaim}

\begin{claimproof}
  We show that the assignment is injective by proving that each
  edge~$e$ of the $\alpha_5$-orientation has an unique successor if it
  has one.  Since~$E'$ is a finite set, this implies that the
  assignment is bijective.

  Let~$e=vw$ be an edge ending in a stack vertex~$w$.  Then~$e$
  corresponds to the concave corner of the abstract quadrilateral~$B$
  of~$w$.  Let~$p$ be the convex corner of~$B$ with a sign-change
  (this corner is unique due to Claim~\ref{claim:facial_sign_changes}),
  let~$A_1$ be the abstract pentagon touching~$p$ with the interior of
  a side, and let~$A_2$ be the abstract pentagon touching~$p$ with
  a corner.  Further, for~$i=1,2$, let~$u_i$ be the normal vertex
  corresponding to~$A_i$.  If~$u_1=v$, we are in the first or the last
  case of Fig.~\ref{fig:sign_changing_edges_predecessors} and the
  successor of~$e$ has to be the edge~$wu_2$.  Otherwise the successor
  of~$e$ has to be the edge~$wu_1$.

  Now let~$e=vw$ be an edge ending in a normal vertex~$w$.  If~$e \in
  E'$, it corresponds to a sign-change at a point~$p$ in the interior
  of a side of the abstract pentagon~$A$ of~$w$.  Let~$x_w \geq 0$
  (the other case is symmetric).  Then the successor of~$e$ has to be
  the edge corresponding to the first corner of~$A$ that we reach when
  we go from~$p$ in the direction of the negative segment.
\end{claimproof}

Due to Claim~\ref{claim:predecessor_assignment_injective} the edge set~$E'$
is a disjoint union of directed simple cycles in the
$\alpha_5$-orientation separating the negative and non-negative
variables.  Due to Claim~\ref{claim:existence_sign_sep_edge} this union is
non-empty if there are negative variables.
\end{proof}

With this theorem at hand we have a way of changing the five color
forest and restart the algorithm. We cannot prove that the iteration
will eventually stop with a non-negative solution. The following 
theorem, however, shows in a very special case that the change of
the five color forest can have the intended effect, i.e., change the
signs of negative variables to positive.

Let~$g$ be an oriented cycle in an $\alpha_5$-orientation.
Then exactly one vertex~$w$ of~$g$ is a stack vertex.
Let~$s_1$ and~$s_2$ be the two edges incident to the concave corner of the abstract facial
quadrilateral corresponding to~$w$. Exactly one edge~$s_i$ of~$s_1$ and~$s_2$ has an endpoint that
is the corner of an abstract pentagon corresponding to a vertex of~$g$.
We call~$s_i$ the segment \emph{surrounded by~$g$} (see Fig.~\ref{fig:flip}).

\begin{theorem}\label{thm:progress}
  Let~$F$ be a five color forest, let~$g$ be an oriented facial cycle in
  the corresponding $\alpha_5$-orientation and let~$F'$ be the five
  color forest obtained from~$F$ by flipping~$g$. Let $\xi$ and $\xi'$
  be the solutions of the equation systems corresponding to~$F$ and
  $F'$, respectively. Let~$s_g$ be the segment surrounded by~$g$ 
  and let $\xi_g$ and $\xi'_g$ be the component of $\xi$
  and~$\xi'$, respectively, which corresponds to $s_g$, i.e., records the
  `length' of $s_g$. Then~$\xi_g$ and~$\xi'_g$ have different signs
  or~$\xi_g=\xi'_g=0$.
\end{theorem}
\begin{proof}
  We denote the equation systems corresponding to~$F$ and $F'$ as~${A_F \mathbf{x}
  = \mathbf{e_1}}$ and ${A_{F'} \mathbf{y} = \mathbf{e_1}}$. Let~$f$ be the face
  of~$G$ containing~$g$. The variable corresponding to $s_g$ in the
  first system is~$x_f^{(i)}$ with $i=1$ or $i=2$ and in the second
  system it is $y_f^{(j)}$ with $j\neq i$ and $j\in\{1,2\}$, i.e.,
  $\xi_g$ is the value of $x_f^{(i)}$ in the solution $\xi$ and
  $\xi'_g$ is the value of $y_f^{(j)}$ in the solution $\xi'$.
  Let~$A_F^{(g)}$ be the matrix obtained from~$A_F$ by replacing the
  column corresponding to~$x_f^{(i)}$ with~$\mathbf{e_1}$, and
  let~$A_{F'}^{(g)}$ be the matrix obtained from~$A_{F'}$ by replacing
  the column corresponding to~$y_f^{(j)}$ with~$\mathbf{e_1}$.
  According to Cramer's rule we have
\[
\xi_g = \frac{\det(A_F^{(g)})}{\det(A_F)} \enspace , \enspace
\xi'_g = \frac{\det(A_{F'}^{(g)})}{\det(A_{F'})} \enspace .
\]

It can be verified that the column of~$A_F$ corresponding
to~$x_f^{(j)}$ and the column of~$A_{F'}$ corresponding to~$x_f^{(i)}$
are equal. We go through the details with the generic example shown
in Fig.~\ref{fig:variables}. In this case $i=2$ and $j=1$.  Consider the
variable $x_f^{(1)}$. It naturally belongs to the equation of color 4
of~$v$ with a coefficient of 1. Due to the substitutions, see the
equations in Lemma~\ref{lemma:face_equations}, it also contributes to the
equations of color 5 at $u$ and of color 2 at $w$, the respective
coefficients are 1 and $\phi$.  Now consider the variable $y_f^{(2)}$. It
naturally belongs to the equation of color 5 of~$u$ with a coefficient
of 1. The substitutions also make it contribute to the equations of
color 4 at $v$ and of color 2 at $w$, the respective coefficients are
1 and $\phi$. Hence, the columns corresponding to $x_f^{(1)}$ and
$y_f^{(2)}$ in their respective systems are equal.

\calc_figscale{40}%
\begin{figure}[t]
    \centerline{\input{\fpath/variables.pstex_t}}
    \caption{\label{fig:variables}}
    \end{figure}%
VC
{The face $f$ before and after the flip at the red segment together with
the variables of the face.}%

If we switch the columns  corresponding to~$y_f^{(i)}$ and~$y_f^{(j)}$
in~$A_{F'}$ to get~$\tilde{A}_{F'}$, then, by the above~$A_F$
and~$\tilde{A}_{F'}$ only differ in the column corresponding to the segment
$s_g$, whence~$A_F^{(g)} = \tilde{A}_{F'}^{(g)}$.

To prove the theorem it remains to show
that~$\det(A_F)$ and~$\det(\tilde{A}_{F'})$ have different signs.  Similar to
the proof of Theorem~\ref{thm:uniquely_solvable} we can do this by showing
that a perfect matching~$M$ of~$H_F$ and a perfect matching~$M'$
of~$H_{F'}$ that both do not contain a pair of crossing
edges, have different signs.

Let~$v_1$ and~$v_2$ be the vertices of~$H_F$ and~$H_{F'}$ corresponding
to face~$f$ such that~$v_1$ corresponds to~$s_g$.  Note that this
makes the local situations around~$f$ in~$H_F$ and~$H_{F'}$ asymmetric
(see Fig.~\ref{fig:HF-graphs}). The asymmetry corresponds to the switch
of columns from~$A_{F'}$ to~$\tilde{A}_{F'}$.

\calc_figscale{40}%
\begin{figure}[t]
    \centerline{\input{\fpath/HF-graphs.pstex_t}}
    \caption{\label{fig:HF-graphs}}
    \end{figure}%
VC
{The local situation around $f$ in $H_F$ and~$H_{F'}$.}%

Let~$w_1$ be the unique equation-vertex
that is adjacent to~$v_1$ and~$v_2$ in~$F$ and~$F'$ (in
Fig.~\ref{fig:variables} vertex $w_1$ would correspond to the equation of
color 2 at~$w$). 
Let~$w_2$ be the equation-vertex that is adjacent to
both of~$v_1$ and~$v_2$ only in~$F$, and in~$F'$ only to~$v_2$ (in
Fig.~\ref{fig:variables} vertex $w_2$ would correspond to the equation of
color 5 at~$u$).  
Let~$w_3$ be the equation-vertex
belonging to the same pentagon as~$w_2$ that is adjacent to~$v_1$
in~$F'$ (in Fig.~\ref{fig:variables} vertex $w_3$ also belongs to $u$ and
has color~1).  
Let~$w_4$ be the equation-vertex that is adjacent to~$v_1$ in~$F$ and
has no adjacency in $F'$ (in Fig.~\ref{fig:variables} vertex $w_4$
corresponds to the equation of color 3 at~$v$).
Finally, let~$w_5$ be the equation-vertex belonging to the same
pentagon as~$w_4$ that is adjacent to both of~$v_1$ and~$v_2$ in~$F'$,
and in~$F$ only to~$v_2$ (in Fig.~\ref{fig:variables} vertex $w_5$
corresponds to the equation of color 4 at~$v$).

The non-crossing condition implies that~$M$ does not contain both of the
edges~$v_1w_1$ and~$v_2w_2$, and that~$M'$ does not contain both of
the edges~$v_1w_1$ and~$v_2w_4$.

We distinguish two cases.  In the first case, both of~$M$ and~$M'$
contain the edge~$v_1w_1$.  Then~$M$ has to contain the edge~$v_2w_5$
and~$M'$ has to contain the edge~$v_2w_2$.  In this case~$M'$
is a matching of~$H_F$ which contains a single crossing while 
$M$ is a matching without crossing. Therefore,~${\sgn(M) \neq \sgn(M')}$.

In the second case, at least one of the matchings~$M$ and~$M'$ does
not contain the edge~$v_1w_1$.  Assume~$M'$ contains the edge~$v_1w_j$
with~$j\in\{ 3,4 \}$ and not the edge~$v_1w_1$.  The case where~$M$
does not contain the edge~$v_1w_1$ is symmetric. Note that we can add
the edge~$v_1w_j$ to $H_F$ without creating an additional crossing and
let $\hat{H}$ be the thus obtained graph. Then~$M$ and~$M'$ are
matchings of~$\hat{H}$.  If we delete one edge of each pair of
crossing edges from~$H_F$, all inner faces are bounded by simple
cycles of length~$6$. By doing the same with $\hat{H}$ one of the
$6$-cycles is divided into two cycles
of length~$4$ by the edge~$v_1w_j$. 

To argue that~${\sgn(M) \neq \sgn(M')}$ we define the
graph~$\tilde{H}$ obtained from~$\hat{H}$ by subdividing the
edge~$v_1w_j$ into three edges~$v_1u_1,u_1u_2,u_2w_j$.  Then, if we
delete one edge of each pair of crossing edges from~$\tilde{H}$, all
inner faces are bounded by simple cycles of length~$6$.
Let~$\tilde{M}$ be the perfect matching of~$\tilde{H}$ obtained
from~$M$ by adding the edge~$u_1u_2$, and let~$\tilde{M'}$ be the
perfect matching of~$\tilde{H}$ obtained from~$M'$ by replacing the
edge~$v_1w_j$ with the edges~$v_1u_1$ and~$u_2w_j$.  Then the
symmetric difference of~$\tilde{M}$ and~$\tilde{M'}$ is a disjoint
union~${C_1,\dotsc,C_k}$ of simple cycles of lengths~${\ell(C_i) \equiv 2
\mod 4}$.  Let~$C_1$ be the cycle containing the
edges~$v_1u_1,u_1u_2,u_2w_j$.  Then the symmetric difference of~$M$
and~$M'$ is the disjoint union of the cycles~$C_2,\dotsc,C_k$ and the
cycle~$C_1'$ obtained from~$C_1$ by replacing the
edges~$v_1u_1,u_1u_2,u_2,w_j$ with the edge~$v_1w_j$.
Therefore~${\ell(C_1') \equiv 0 \mod 4}$ and~${\sgn(M) \neq \sgn(M')}$.
\end{proof}

\section{Concluding remarks}

We cannot prove that the iterations of the algorithm lead to any kind
of progress. Therefore, it may be that the algorithm cycles and runs
forever. However, there are two independent implementations
\cite{raasch2018kontaktdarstellungen,website2018kcontact} of the algorithm
and experiments with these implementations have always been successful.

Similar algorithms for the computation of contact
representations with homothetic squares or triangles have been described
in \cite{felsner2009triangle} and \cite{felsner2013rectangle}.  These
algorithms have also been subject to extensive experiments
\cite{picchetti2011finding,rucker2011kontaktdarstellungen} that have
always been successful. We therefore have the following conjecture.

\begin{conjecture}
  The algorithm described above terminates with a non-negative
  solution for every graph~$G$ which is an inner triangulation of a
  $5$-gon, and for every initial five color forest~$F$ of~$G$.
\end{conjecture}

A proof of this conjecture would imply a new proof for the existence
of pentagon contact representations for these graphs.  Moreover, since
they only depend on the values of the solution of a linear system of
equations, the coordinates for the corners of the pentagons could be
computed exactly. If the proof would come with a polynomial bound on
the number of iterations before termination, then the algorithm would
run in strongly polynomial time when doing arithmetics in the
extension field $\mathbb{Q}[\sqrt{5}]$ of the rationals.



\begin{thebibliography}{10}

\bibitem{bernardi-fusy2012}
Olivier Bernardi and {\'{E}}ric Fusy.
\newblock Schnyder decompositions for regular plane graphs and application to
  drawing.
\newblock {\em Algorithmica}, 62:1159--1197, 2012.

\bibitem{brehm20003}
Enno Brehm.
\newblock 3-orientations and {S}chnyder 3-tree-decompositions.
\newblock {D}iplomarbeit, Freie Universit{\"a}t Berlin, 2000.
\newblock URL:
  \url{page.math.tu-berlin.de/\~felsner/Diplomarbeiten/brehm.ps.gz}.

\bibitem{BSST-40}
Rowland~L. Brooks, Cedric A.~B. Smith, Arthur~H. Stone, and William~T. Tutte.
\newblock The dissection of rectangles into squares.
\newblock {\em Duke Mathematical J.}, 7:312--340, 1940.

\bibitem{de2001topological}
Hubert {de Fraysseix} and Patrice {Ossona de Mendez}.
\newblock On topological aspects of orientations.
\newblock {\em Discrete Mathematics}, 229(1):57--72, 2001.

\bibitem{de1994triangle}
Hubert {de Fraysseix}, Patrice {Ossona de Mendez}, and Pierre Rosenstiehl.
\newblock On triangle contact graphs.
\newblock {\em Combinatorics, Probability and Computing}, 3:233--246, 1994.

\bibitem{felsner2004lattice}
Stefan Felsner.
\newblock Lattice structures from planar graphs.
\newblock {\em Electronic J.\ of Combinatorics}, 11(1):R15, 2004.

\bibitem{felsner2009triangle}
Stefan Felsner.
\newblock Triangle contact representations.
\newblock In {\em Midsummer Combinatorial Workshop, Praha}, 2009.
\newblock URL: \url{page.math.tu-berlin.de/\~felsner/Paper/prag-report.pdf}.

\bibitem{felsner2013rectangle}
Stefan Felsner.
\newblock Rectangle and square representations of planar graphs.
\newblock In {\em Thirty Essays on Geometric Graph Theory}, pages 213--248.
  Springer, 2013.

\bibitem{gonccalves2011triangle}
Daniel Gon\c{c}alves, Benjamin L{\'e}v{\^e}que, and Alexandre Pinlou.
\newblock Triangle contact representations and duality.
\newblock {\em Discrete and Computational Geometry}, 48(1):239--254, 2012.

\bibitem{kobourov2016canonical}
Stephen~G. Kobourov.
\newblock Canonical orders and {S}chnyder realizers.
\newblock In {\em Encyclopedia of Algorithms}, pages 277--283. Springer, 2016.

\bibitem{picchetti2011finding}
Thomas Picchetti.
\newblock Finding a square dual of a graph.
\newblock 2011.
\newblock URL:
  \url{page.math.tu-berlin.de/\~felsner/Diplomarbeiten/rapport\_picchetti.pdf}.

\bibitem{raasch2018kontaktdarstellungen}
Nadine Raasch.
\newblock {K}ontaktdarstellungen planarer {G}raphen mit {F}{\"u}nfecken.
\newblock {M}asterarbeit, {T}echnische {U}niversit{\"a}t {B}erlin, 2018.
\newblock URL:
  \url{page.math.tu-berlin.de/\~felsner/Diplomarbeiten/Masterarbeit_Nadine-Raasch.pdf}.

\bibitem{rucker2011kontaktdarstellungen}
Julia Rucker.
\newblock {K}ontaktdarstellungen von planaren {G}raphen.
\newblock {D}iplomarbeit, {T}echnische {U}niversit{\"a}t {B}erlin, 2011.
\newblock URL:
  \url{page.math.tu-berlin.de/\~felsner/Diplomarbeiten/dipl-Rucker.pdf}.

\bibitem{website2018kcontact}
Manfred Scheucher and Hendrik Schrezenmaier.
\newblock k-{C}ontact {R}epresentations, 2018.
\newblock URL:
  \url{https://www3.math.tu-berlin.de/diskremath/research/kgon-representations/index.html}.

\bibitem{schnyder1990embedding}
Walter Schnyder.
\newblock Embedding planar graphs on the grid.
\newblock In {\em Proc. SODA}, pages 138--148, 1990.

\bibitem{schramm2007combinatorically}
Oded Schramm.
\newblock Combinatorically prescribed packings and applications to conformal
  and quasiconformal maps.
\newblock Modified version of PhD thesis from 1990.
\newblock \href {http://arxiv.org/abs/0709.0710v1} {\path{arXiv:0709.0710v1}}.

\bibitem{schramm1993square}
Oded Schramm.
\newblock Square tilings with prescribed combinatorics.
\newblock {\em Israel J.\ of Mathematics}, 84(1-2):97--118, 1993.

\bibitem{schrezenmaier2016zur}
Hendrik Schrezenmaier.
\newblock Zur {B}erechnung von {K}ontaktdarstellungen.
\newblock {M}asterarbeit, {T}echnische {U}niversit{\"a}t {B}erlin, 2016.
\newblock URL: \url{page.math.tu-berlin.de/\~schrezen/Papers/Masterarbeit.pdf}.

\bibitem{steiner2016existenz}
Raphael Steiner.
\newblock Existenz und {K}onstruktion von {D}reieckszerlegungen triangulierter
  {G}raphen und {S}chnyder woods.
\newblock {B}achelorarbeit, {F}ern{U}niversit{\"a}t in {H}agen, 2016.
\newblock URL:
  \url{www.fernuni-hagen.de/mathematik/DMO/pubs/Bachelorarbeit\_Raphael\_Steiner.pdf}.

\bibitem{Thom-06}
Robin Thomas.
\newblock A survey of {P}faffian orientations of graphs.
\newblock In {\em International {C}ongress of {M}athematicians. {V}ol. {III}},
  pages 963--984. Eur. Math. Soc., 2006.

\end{thebibliography}

\end{document}